\documentclass[11pt]{article}
\usepackage{jheppub}
\usepackage{amsmath,amssymb,amsfonts,amsthm,mathtools,bm}
\usepackage[dvipsnames]{xcolor}
\usepackage{comment}
\usepackage{hyperref}
\usepackage[utf8]{inputenc}
\usepackage[titletoc]{appendix}
\usepackage{tikz}
\usetikzlibrary{shapes,arrows,shadows}
\usepackage[colorinlistoftodos]{todonotes}
\usepackage{xspace}
\usepackage{etoolbox}
\usepackage{listofitems}
\usepackage{subcaption}

\newtheorem{prop}{Proposition}
\newtheorem{conj}{Conjecture}

\makeatletter
\def\@fpheader{\relax}
\makeatother

\newcommand\be{\begin{equation}}
\newcommand\ee{\end{equation}}
\newcommand\bea{\begin{eqnarray}}
\newcommand\eea{\end{eqnarray}}
\newcommand\ba{\begin{array}}
\newcommand\ea{\end{array}}

\newcommand\eref[1]{(\ref{#1})}
\newcommand\bc{\begin{center}}
\newcommand\ec{\end{center}}

\newcommand{\kae}{K{\"a}hler\xspace}

\renewcommand\comment[1]{}

\def\IC{\mathbb{C}}
\def\IP{\mathbb{P}}
\def\ZZ{\mathbb{Z}}
\let\e=\epsilon
\let\l=\lambda
\def\mB#1{\makebox[2em][r]{$#1$}}

\renewcommand{\TH}[1]{{\color{Magenta}$^!$%
               \marginpar{\leavevmode\color{Magenta}\raggedright%
                           \scriptsize{\sf TH:} #1}}}

 \allowdisplaybreaks
 \numberwithin{equation}{section}

\title{
Machine Learned Calabi--Yau Metrics and Curvature}

\author[a]{Per Berglund}
\author[a]{\!\!, Giorgi Butbaia}
\author[b]{\!\!, Tristan H\"ubsch}
\author[c]{\!\!, Vishnu Jejjala}
\author[c]{\!\!, \\ Dami\'an Mayorga Pe\~na}
\author[d]{\!\!, Challenger Mishra}
\author[d]{\!\!, Justin Tan}

\affiliation[\,a]{Department of Physics and Astronomy, University of New Hampshire, Durham, NH 03824, USA}
\affiliation[\,b]{Department of Physics and Astronomy, Howard University, Washington, DC 20059, USA}
\affiliation[\,c]{Mandelstam Institute for Theoretical Physics, School of Physics, NITheCS, and CoE-MaSS,\\
University of the Witwatersrand, Johannesburg, WITS 2050, South Africa}
\affiliation[\,d]{Department of Computer Science \& Technology, University of Cambridge, Cambridge CB3 0FD, UK}

\emailAdd{Per.Berglund@unh.edu}
\emailAdd{Giorgi.Butbaia@unh.edu}
\emailAdd{thubsch@howard.edu}
\emailAdd{v.jejjala@wits.ac.za}
\emailAdd{damian.mayorgapena@wits.ac.za}
\emailAdd{cm2099@cam.ac.uk}
\emailAdd{jt796@cam.ac.uk}

\abstract{
Finding Ricci-flat (Calabi--Yau) metrics is a long standing problem in geometry with deep implications for string theory and phenomenology.
A new attack on this problem uses neural networks to engineer approximations to the Calabi--Yau metric within a given K\"ahler class.
In this paper we investigate numerical Ricci-flat metrics over smooth and singular K3 surfaces and Calabi--Yau threefolds.
Using these Ricci-flat metric approximations for the Cefal\'u family of quartic twofolds and the Dwork family of quintic threefolds, we study characteristic forms on these geometries.
We observe that the numerical stability of the numerically computed topological characteristic is heavily influenced by the choice of the neural network model, in particular, we briefly discuss a different neural network model, namely spectral networks, which correctly approximate the topological characteristic of a Calabi--Yau. Using persistent homology, we show that high curvature regions of the manifolds form clusters near the singular points.
For our neural network approximations, we observe a Bogomolov--Yau type inequality $3c_2 \geq c_1^2$ and observe an identity when our geometries have isolated $A_1$ type singularities.
We sketch a proof that $\chi(X~\smallsetminus~\mathrm{Sing}\,{X}) + 2~|\mathrm{Sing}\,{X}| = 24$ also holds for our numerical approximations.
}

\usetikzlibrary{fit,shapes,positioning}
\def\layersep{2.0cm}
\def\widthNN{4}

\begin{document}

\maketitle
\parskip=5pt

\section{Introduction, rationale, and summary}
\label{s:IRS}
Ricci-flat metrics satisfy the Einstein equations without any energy-momentum tensor source, and so describe empty spacetime.
They are thus of fundamental physical interest and have been studied intensely for over a century.
The closely related \kae metrics on compact, complex geometries of Euclidean signature --- the Calabi--Yau spaces --- describe (to lowest order) the extra spacelike dimensions in string compactifications.

While many Calabi--Yau properties are accessed purely via topology, knowledge of the Ricci-flat metric is crucial for certain explicit computations, such as finding $\alpha^\prime$ corrections, fixing the K\"ahler potential, determining Yukawa couplings, and deducing aspects of supersymmetry breaking in the low-energy effective field theory~\cite{Candelas:1985en,Candelas:1987rx,Green:1987mn}.
All of these are prerequisites for calculating the proverbial electron mass in a Standard Model derived from string compactification.

Yau's proof~\cite{yau1977calabi, Yau:1978} of the Calabi conjecture~\cite{Calabi:1954} is famously not constructive.
We know that a Calabi--Yau manifold has a unique Ricci-flat metric in each \kae class, but not the actual expressions for such metrics.
In recent years, there has been revived interest in obtaining Ricci-flat Calabi--Yau metrics.
On the analytical side, explicit expressions have been obtained for the metric on certain K3 manifolds~\cite{Kachru:2018van, Kachru:2020tat}.\footnote{
The analytic expression for the flat metric on the Calabi--Yau onefold was written by Clifford in 1873~\cite{volkert2013space}.
}
On the computational side, various machine learning techniques have been employed to obtain numerical as well as spectral approximations to flat metrics~\cite{Ashmore:2019wzb, Anderson:2020hux, Douglas:2020hpv, Jejjala:2020wcc, Larfors:2021pbb, Ashmore:2021ohf, Larfors:2022nep, Ashmore:2020ujw, Ashmore:2021qdf}. 
In some of these machine driven approaches, neural networks have provided new representations to study metrics over manifolds with special holonomy. Although there are guarantees to a neural network's ability to approximate reasonably well behaved functions with arbitrary accuracies, there are further mathematical guarantees for neurocomputing stemming from the powerful Kolmogorov--Arnold representation theorem from the middle of the previous century. 
To date, the state of the art in approximating Ricci-flat Calabi--Yau metrics using machine learning methods is the \texttt{cymetric} code~\cite{Larfors:2021pbb,cymetric}, which permits us to obtain approximate flat metrics for complete intersection (CICYs), as well as toric Calabi--Yau spaces derived from triangulations of reflexive polytopes from the Kreuzer--Skarke list~\cite{Kreuzer:2000xy}.
Past efforts to obtain flat metrics include the so-called Donaldson's algorithm~\cite{donaldson2001,donaldson2005some,Douglas:2006rr} as well as applications of the Gauss--Seidel method in~\cite{Headrick:2005ch}, the optimization of energy functionals~\cite{Headrick:2009jz}, and examination of scaling properties~\cite{Cui:2019uhy, Douglas:2021zdn}.
(See also~\cite{Douglas:2006hz, Braun:2007sn, Braun:2008jp, Anderson:2010ke, Anderson:2011ed}.)

In this paper, we apply machine learning techniques to singular and non-singular K3 manifolds and quintic threefolds.
One of the goals of this work is to characterize how good the machine learned metrics are for phenomenology.
In order to perform this assessment, we consider computations of topological quantities such as the Euler character and the Chern classes as obtained from the metric --- as well as the corresponding curvature \emph{distributions}.
Thereby, we seek to determine which parts of the geometry contribute most significantly to the various topological quantities.
To benchmark these observations, it is useful to compare the Fubini--Study metric to the machine learned metric.

We also propose a different neural network model, called \emph{spectral networks}, for approximating the Ricci-flat metric. We find that the numerical invariants computed using the spectral networks exhibit higher numerical stability than standard fully-connected networks which directly use the homogeneous coordinates as input. Furthermore, we briefly discuss the final loss achieved by these networks and find that for Fermat quartic, the lowest $\sigma$-loss is below $10^{-3}$, which is at the same level of accuracy as the method described in~\cite{Headrick:2009jz} using $k=8$.

\comment{\sf\slshape
\begin{itemize}
\item Numerical computations of the metric: what's known and with which we may compare\TH{!`more refs!!!}~\cite{Kachru:2020tat, Headrick:2005ch}
\item {\color{red} List main results \& sketch a road-map through the paper.}
\begin{enumerate}
\item Reproduction/checking of topological data ($\chi_E$) using Fubini--Study and the learned Ricci-flat metric for K3
\item Fubini--Study metric computations vs. (learned) Ricci-flat metric for K3
\item Singularities mess up the (learned) Ricci-flat metric for K3
\item Cautionary note, for using the Ricci-flat computation w/o knowing if the hypersurface is singular
\item Compute $\chi$ for Fermat/Dwork quintics, Tian--Yau with Fubini--Study and learned
\item Toric vs.\ CICY for quintic
\item Compute Yukawas for compactification on Tian--Yau
\item Somewhere we should define $e(J)$
\end{enumerate}
\end{itemize}
}

The organization of this paper is as follows.
Section~\ref{sec:models} describes the deformation families of Calabi--Yau twofolds and threefolds we investigate in this paper and the considered curvature related features.
Section~\ref{sec:nm} briefly summarizes the numerical methods we apply.
Section~\ref{sec:results} considers the machine learned metrics.
Certain details of numerical computations are discussed in the appendices.

\section{The testbed models and their curvature}
\label{sec:models}
We consider several simple, one parameter deformation families of Calabi--Yau twofolds and threefolds.
In each case, we focus on a few curvature related features for which we compare the results obtained with a numerical approximation to the Ricci-flat metric, with those obtained using the pullback of the Fubini--Study metric, as well as with the known exact results. 

\subsection{The deformation families}
\label{s:deformations}
\paragraph{The Cefal{\'u} pencil:}
Consider a complex one parameter deformation family of quartics in $\IP^3$~\cite{Catanese:2021aa}:
\begin{equation}
   \IP^3\supset X_\l\,{:=}\,\big\{p_\l(z)\,{=}\,0\big\}:\quad
   p_\l(z)\!:=\!\sum_{i=0}^3z_i^{4}
                     -\frac\l3\left(\sum_{i=0}^3z_i^{2}\right)^2 ~.
 \label{e:Cefalu}
\end{equation}
The Cefal{\'u} hypersurface~\cite{Catanese:2021aa} is the $\l\,{=}\,1$ case.
We call the general $\lambda$ deformation, the Cefal\'u family (pencil) of quartics.
 While this deformation family of hypersurfaces provides for a rather more detailed analytic analysis~\cite{Candelas:1990rm,Candelas:1990qd,Nahm:1999ps,Wendland:2003ma},
 we focus on a few immediate results for the purpose of comparing with numerical computations of the metric and various metric characteristics on these K3 surfaces.
For each $\l\,{\in}\,\IC$, the defining polynomial $p_\l(z)$ is manifestly invariant under all permutations of the $z_i$, as well as sign changes $z_i\,{\mapsto}\,-z_i$, separately for each $i=0,1,2,3$.
Subject to preserving the holomorphic two-form,
 $\Omega\,{:=}\oint\frac{(z\,{\rm d}^3z)}{p_\l(z)}$, \textit{viz.}, the Calabi--Yau condition, and modulo the $\IP^3$-projectivization, this generates an $S_4\times(\ZZ_2)^2$ symmetry.
 Of the various possible quotients, here we only need $X_0/\ZZ_2$.
 The overall situation is sketched in Figure~\ref{f:Cefalu}, and includes both the $\l$-plane of the hypersurfaces~\eqref{e:Cefalu} as well as the $\l$-plane of their $\ZZ_2$ quotients.
\begin{figure}[htb]
 \begin{center}
   \begin{tikzpicture}
     \path[use as bounding box](-7,-1.7)--(7.5,2.5);
     \fill[yellow!30](-7,-1)--(4,-2)--(7,1.5)--(0,1.5);
      \path(-5.5,-.8)node{\large$\l$};
     \fill[gray!10!cyan!20, opacity=.75](-7,1)--(4,1)--(7,2)--(0,2);
     \draw[gray, -stealth, ultra thick](0,0)--(-10:5.2);
     \draw[densely dotted, ->, thick](0,0)--(0,1.4);
      \path(.1,.9)node[left]{\footnotesize$\big/\ZZ_2$};
     \draw[red!30!gray, double, dotted, thick](0,1.5)--(-10:1.125);
      \path[red!30!gray](.6,.6)node{$\scriptstyle\spadesuit$};
     \draw[densely dotted, ->, thick](-5,-.8)--(-5,1.1);
      \path(-4.9,.2)node[left]{\footnotesize$\big/\ZZ_2$};
     \draw[blue!80!black, dashed, -stealth, very thick]
         (4.35,.75)to[out=175,in=0](.1,1.5);
      \path(-10:5.2)node[right]{$\Re(\l)$};
       \path(2.1,1.33)node[blue!75!black, right, rotate=7]
         {\footnotesize K{\"a}hler class variation~\cite{Headrick:2005ch}};
      \path(0,.2)node[left]{\small$X_0{:=}\{p_0(z)\!=\!0\}$~~};
      \path[red!75!black]{(-10:1.125)+(.1,-.1)}
         node[left, rotate=7]{\small$X_{\frac34}{:=}\{p_{\frac34}(z)\!=\!0\}$~~};
      \path[blue!75!black]{(-10:1.5)+(0,.1)}
         node[right, rotate=13]{\small$X_1{:=}\{p_{1}(z)\!=\!0\}\!\approx\!
                                 \{T^4/\ZZ_2\}$};
      \path[green!60!black]{(-10:2.25)+(.2,-.1)}
         node[left, rotate=15]{\small$X_{\frac32}{:=}\{p_{\frac32}(z)\!=\!0\}$~~};
      \path[magenta!75!black]{(-10:4.5)+(.15,.05)}
         node[left, rotate=22]{\small$X_3{:=}\{p_{3}(z)\!=\!0\}$~~};
      \path(0,1.7)node[left]
           {\small$X_0/\ZZ_2$~\cite{Headrick:2005ch}~~};
     \filldraw[fill=white, very thick](0,0)circle(.7mm);
      \filldraw[fill=red,thick](-10:1.125)circle(.6mm);
      \filldraw[fill=blue!60,thick](-10:1.5)circle(.6mm);
      \filldraw[fill=green,thick](-10:2.25)circle(.6mm);
      \filldraw[fill=purple,thick](-10:4.5)circle(.6mm);
      \filldraw[fill=red!30,thick](0,1.5)circle(.6mm);
    \end{tikzpicture}
 \end{center}
 \caption{The Cefal\'u (complex structure) deformation family of quartics~\eqref{e:Cefalu}
  (lower $\l$-plane), and their $\ZZ_2$ quotients (upper plane),
  together with the identifications
  $X_1\,{\approx}\,\{T^4/\ZZ_2\}$ and 
  $X_{\frac34}\,{\approx}\,\big(X_0/\ZZ_2\big)$,
  the latter identification labeled by ``$\scriptstyle\spadesuit$''.}
 \label{f:Cefalu}
\end{figure}

Variation of $\l$ parametrizes a deformation of the complex structure of the hypersurface $X_\l$, while the subsequently discussed numerical computation of the various metric characteristics is explicitly designed, as in~\eqref{eq:ddphi_model},
to preserve the K{\"a}hler class of the embedding $\IP^3$.
The defining polynomial $p_\l(z)$ fails to be transverse only for
  $\l\,{\to}\,\l^\sharp\,{\in}\,\{\frac34,\, 1,\, \frac32,\, 3\}$
  in the finite\footnote{At $\l\,{\to}\,\infty$, the vanishing of $p_\infty\,{=}\,\big(\sum_i z_i^{~2}\big){}^2$ defines an ``\emph{everywhere singular}'' hypersurface (more properly, a \emph{scheme}), $X_\infty$, since $\vec\partial p_\infty(z)$ vanishes wherever $p_\infty(z)$ does.
  Also, denoting $Z_\infty\,{:=}\,\{\sum_i z_i^{~2}\,{=}\,0\}$, we have that $X_\infty=Z_\infty\,{\cup}\,Z_\infty$ is an everywhere doubled space, which is singular at $X^\sharp_\infty=Z_\infty\,{\cap}\,Z_\infty$ --- indeed, everywhere.
  In this respect, the $\lim_{\l\to\infty}X_\l$ limit is an extremely degenerate case of Tyurin degenerations~\cite{Berglund:2022dgb}.} $\l$-plane ($|\l|\,{<}\,\infty$),
 and the quartic hypersurface, $X_\l$, singularizes there and has, respectively: $8$, $16$, $12$, and $4$ isolated singular points.
 Each of these singular points is an $A_1$-singularity (\textit{i.e.}, \emph{node} or \emph{double point}): the gradient of the defining equation~\eqref{e:Cefalu} vanishes there, but the Hessian (matrix of second derivatives in local coordinates) is regular; for any other $A$-$D$-$E$ (so-called Du~Val) singular point the local Hessian would vanish.\footnote{All the relevant details and facts about these singularities and their (complex structure variation) deformations and (K{\"a}hler class variation) desingularizations are found in Reid's comprehensive survey~\cite{Reid:2022aa}.}
 The singularities of each $X_{\l^\sharp}$ are thus all hypersurface singular points of the form $(0,0,0)\,{\in}\,\{xy\,{=}\,z^2\}\,{\subset}\,\IC^3$ in local coordinates, and are also equivalently described as discrete quotient (orbifold) singularities of the form $(0,0)\,{\in}\,\IC^2/\ZZ_2$.
Deforming any of the singular hypersurfaces $X_{\l^\sharp}\to X_{\l^\sharp+\e}$ for $|\e|\,{\ll}\,1$ changes its complex structure and smooths it by replacing each of its nodes with an $S^2$-like so-called \emph{vanishing cycle} of radius ${\propto}\,|\e|$.

 Alternatively and without changing either $\l$ or the complex structure in general, each node can also be desingularized by a blow-up. This surgically replaces each node with a copy of the (again $S^2$-like) so-called \emph{exceptional set}, which admits a compatible complex structure and is biholomorphic\footnote{In fact, it is the total space of the $\mathcal{O}_{\IP^1}(-2)$ line bundle over this $\IP^1$ that replaces the excised singularity compatibly with the global complex structure of the quartic hypersurface.
 The tangent bundle being $T_{\IP^1}\,{=}\,\mathcal{O}_{\IP^1}(2)$, the adjunction theorem implies that the patching preserves $c_1$. Also, the self-intersection of this exceptional $\IP^1$ is then $-2$.}
 to the \emph{complex} projective space, $\IP^1$ --- and changes the overall K{\"a}hler class by contributing\footnote{In the overall K{\"a}hler metric, $g^{(0)}_{ab}$, defined without/before the blow-up, the exceptional set replacing a point is \emph{null}; $g^{(0)}_{ab}$ continues to ``see'' it as a point.
 To correct this, one varies the metric by adding a local contribution inherent to the exceptional set.} a (variable) multiple of the K{\"a}hler class inherent to the exceptional $\IP^1$.
 Such K{\"a}hler class variations were studied in~\cite{Headrick:2005ch} and were shown by numerical computation to connect the orbifolds $T^4/\ZZ_2$ and $X_0/\ZZ_2$.
 By desingularizing $T^4/\ZZ_2$ via blowup and successively increasing the size of the exceptional sets, these eventually intersect and form new $\ZZ_2$-orbifold singularities, the so-obtained singular space identifiable with $X_0/\ZZ_2$.

Explicit $\l$-deformation connects $X_1\,{\to}\,X_{\frac34}$ by varying the complex structure while holding the K{\"a}hler class constant.
 One would expect the vanishing cycles of $X_{1-\e}$ to grow as $\e\,{\in}\,[0,\frac14)$, come closer to each other and intersect as $\e\,{\to}\,\frac14$, creating the singularities of $X_{\frac34}$.
 On the other hand, the number, type and highly symmetric distribution of singularities in the hypersurfaces $X_{\frac34}$ and $X_1$ suggests identifying these with the two global orbifolds considered in~\cite{Headrick:2005ch}, resulting in an interesting double connection (the ``$\approx$'' symbols denote a likely but not rigorously proven identification):
\begin{equation}
   \{X_0/\ZZ_2\}
     \begin{tikzpicture}
      \path[use as bounding box](0,0)--(0,.9);
      \draw[blue!75!black, densely dashed, <-, thick]
         (0,.3)to[out=30,in=180]++(.8,.3)--++(2.9,0)to[out=0,in=150]++(.8,-.3);
      \path[blue!75!black](2.3,.8)node{\scriptsize K{\"a}hler class variation~\cite{Headrick:2005ch}};
                \end{tikzpicture}
    \approx X_{\frac34} \xrightarrow[\text{cpx.\;str.\;deform}]{~\l~} X_1 \approx \{T^4/\ZZ_2\} ~.
 \label{e:8->16}
\end{equation}
For the \emph{crepant} ($c_1$-preserving~\cite{hubsch1992calabi}) desingularization of both of these orbifolds, the Euler characteristic may be computed by
\begin{equation}
  \chi\big( \widetilde{M/G} \big) =
  \frac1{|G|}\big(\chi(M)-\chi(F)\big) + \chi(N) ~,
 \label{e:OrbiEu}
\end{equation}
where $M$ is a smooth manifold with the discrete group action $G$ for which $F$ is the fixed point set and $N$ the desingularizing surgical replacement of $F$; for a complete and detailed refinement see~\cite{Atiyah:1989ty} and~\cite[\S\,4.5]{hubsch1992calabi}.
Assuming a similar identification with a global finite quotient to be possible also for $X_{\frac32}$ (and $X_3$), since these have $12$ ($4$) isolated $A_1$-singular points, both with $|G\,{=}\,\ZZ_2|\,{=}\,2$ and where $N$ consists of $12$ ($4$) isolated exceptional $\IP^1$s with $\chi(\IP^1)\,{=}\,2$, we compute
\begin{subequations}
  \label{e:OrbiEu2}
\begin{alignat}9
   \frac12\big(\chi(M_{\frac32})-12\big)+12{\cdot}2
    &=24=\chi(\widetilde{X_{\frac32}}) \quad
    &&\Longrightarrow&\quad \chi(M_{\frac32})&=12 ~;\\
   \frac12\big(\chi(M_3)-4\big)+4{\cdot}2
    &=24=\chi(\widetilde{X_3}) \quad
    &&\Longrightarrow&\quad \chi(M_3)&=36 ~.
\end{alignat}
\end{subequations}
Here, $M_{\frac32}$ ($M_3$) denotes a nonsingular complex surface with a $\ZZ_2$ action that has $12$ ($4$) fixed points, the $\ZZ_2$-orbifolds of which may be identified with $X_{\frac32}$ ($X_3$).
Whereas in~\eqref{e:8->16}, we have $M_{\frac34}\,{\approx}\,X_0$ (the Fermat quartic) and $M_1\,{\approx}\,T^4$, we do not have any obvious candidate for $M_{\frac32}$ and $M_3$, but note that these would have to have $h^{2,0}\,{\geqslant}\,1$, and that precisely one holomorphic volume form must remain after the $\ZZ_2$ quotient.
Since $X_{\frac34+\e}$, $X_{1+\e}$, $X_{\frac32+\e}$ and $X_{3+\e}$ are all smoothed by a ``$\l^\sharp\,{\to}\,\l^\sharp\,{+}\,\e$'' (complex structure) deformation, they are conceptual mirrors of the (K{\"a}hler class variation) desingularizations indicated in~\eqref{e:OrbiEu} and~\eqref{e:OrbiEu2}.

\paragraph{The Dwork pencils:}
Analogous to the Dwork pencil of quintics~\cite{Candelas:1990rm}
\begin{equation}
   \IP^4\supset Z_\psi\,{:=}\,\big\{Q_\psi(z)\,{=}\,0\big\}:\quad
   Q_\psi(z)\!:=\!\sum_{i=0}^4z_i^{~5}
                     -5\psi\prod_{i=0}^4z_i ~,
 \label{e:Dwork5}
\end{equation}
we also consider the Dwork pencil of quartics
\begin{equation}
   \IP^3\supset Y_\psi\,{:=}\,\big\{q_\psi(z)\,{=}\,0\big\}:\quad
   q_\psi(z)\!:=\!\sum_{i=0}^3z_i^{~4}
                     -4\psi\prod_{i=0}^3z_i ~.
 \label{e:Dwork}
\end{equation}
Following the by now well known analysis of the quintic~\eqref{e:Dwork5}, it is easy to show that $\psi\,{\simeq}\,\alpha\psi$, with $\alpha^4\,{=}\,1$, so $\mathrm{Arg}(\psi)\,{\in}\,[0,\pi/2]$ provides a fundamental domain, subject to identifying the edges $\{\mathrm{Arg}(\psi)\,{=}\,0\}\,{\simeq}\,\{\mathrm{Arg}(\psi)\,{=}\,\pi/2\}$, thus forming a cone.
The two families~\eqref{e:Cefalu} and~\eqref{e:Dwork} are related: evidently, $Y_0\,{=}\,X_0$.
Furthermore, $q_\psi(z)$ fails to be transverse only for $\psi\,{=}\,\alpha$ with $\alpha^4\,{=}\,1$, where $Y_\alpha$ has $16$ isolated $A_1$-singular points, $(1,\alpha\beta^{2i}\gamma^{3j},\alpha\beta^i\gamma^{2j},\alpha\beta^i\gamma^{3j})$ for $\alpha^4\,{=}\,\beta^4\,{=}\,\gamma^4\,{=}\,1$ and $i,j=0,\dots,3$.
Thus, $Y_1$ has the same number and type of isolated singular points as $X_1$ (albeit in different locations), and as in Figure~\ref{f:Cefalu}, we identify $Y_1\,{\approx}\,X_1\,{\approx}\,\{T^4/\ZZ_2\}$.
This implies that within the K3 complex structure moduli space, the $\l$-plane and the $\psi$-cone have two points in common: $(\l\,{=}\,0)=(\psi\,{=}\,0)$ and $(\l\,{=}\,1)=(\psi\,{=}\,1)$.
 Finally, we note that the $Y_\infty$ model is a complex projective \emph{tetrahedron}, the union of four $\IP^2$s that meet in six $\IP^1$s that meet in four points.
 Unlike $X_\infty$ (which is singular everywhere), $Y_\infty$ is singular only at the union of those six $\IP^1$s.

\subsection{General remarks on \kae geometry}

Since it is a complex, \kae manifold, the Calabi--Yau metric is Hermitian and can be obtained from a \kae potential $K(z,\bar{z})$:
\begin{equation}
g_{a\bar{b}}=\partial_a \partial_{\bar{b}} K(z,\bar{z}) ~.
\end{equation}
The corresponding \kae form, written in terms of the metric, reads
\begin{equation}
J=\frac{\rm i}{2} g_{a\bar{b}}\, dz^a \wedge d\bar{z}^{\bar{b}} ~.
\end{equation}
As a consequence of K\"ahlerity, the only non-zero Christoffel symbols are those for which holomorphic and antiholomorphic indices do not mix:
\begin{equation}
\Gamma^a_{bc}=\overline{\Gamma^{\bar{a}}_{\bar{b}\bar{c}}}
=(\partial_b g_{c\bar{d}})g^{\bar{d}a} ~.
\end{equation}
From this, we readily compute the non-zero components of the Riemann tensor: 
\begin{equation}
R^a_{b\bar{c}d}=-\bar{\partial}_{\bar{c}}\Gamma^a_{bd} ~,\quad R^a_{\bar{b}cd}=\bar{\partial}_{\bar{b}}\Gamma^a_{cd} ~, 
\end{equation}
together with 
\begin{equation}
R^a_{b\bar{c}d}=\overline{R^{\bar{a}}_{\bar{b}c\bar{d}}} ~,\quad R^a_{\bar{b}cd}=\overline{R^{\bar{a}}_{b\bar{c}\bar{d}}} ~. 
\end{equation}
The only non-vanishing entries for the Ricci tensor of a \kae metric are given by
\be
R_{a\bar{b}}=R^c_{ca\bar{b}}=-\partial_a\partial_{\bar{b}}\, \log{\rm det}\, g ~.
\ee
We may sometimes write ${\rm det}\, g =|g|$.
$R$ is closed but not necessarily exact.
It serves to define the Ricci form ${\rm Ric}(J)={\rm i} R_{a\bar{b}} dz^{a}\wedge d\bar{z}^{\bar{b}}$ which is closed by construction.
Since $c_1(J)$ is required by the Calabi--Yau condition to be zero, the Ricci form of a Calabi--Yau is also exact. For further details, the reader should consult the classic references~\cite{candelas1988lectures, hubsch1992calabi, Ballmann}.

The Riemann tensor can be used to construct the curvature form
\begin{equation}
\mathcal{R}^a_b=R^a_{bm\bar{n}}dz^m \wedge d\bar{z}^{\bar{n}} ~.    
\end{equation}
We then have 
\begin{align}
    {\rm Tr}\, \mathcal{R}\,\, & = R^a_{am\bar{n}}dz^m \wedge d\bar{z}^{\bar{n}}=-{\rm i}\,{\rm Ric}(J)\,,  \\
     {\rm Tr}\, \mathcal{R}^2 & = R^a_{bm_1\bar{n}_1}R^b_{am_2 \bar{n}_2}dz^{m_1} \wedge d\bar{z}^{\bar{n}_1}\wedge dz^{m_2} \wedge d\bar{z}^{\bar{n}_2}\,,  \\
      {\rm Tr}\, \mathcal{R}^3 & = R^a_{bm_1\bar{n}_1}R^b_{cm_2 \bar{n}_2} R^c_{am_3 \bar{n}_3}dz^{m_1} \wedge d\bar{z}^{\bar{n}_1}\wedge dz^{m_2} \wedge d\bar{z}^{\bar{n}_2}\wedge dz^{m_3} \wedge d\bar{z}^{\bar{n}_3}\,,  
\end{align}
which can be used to obtain the various Chern characteristic forms $c_i\in\Omega^{i,i}(\mathcal{M})$, resulting from the expansion 
\begin{gather}
	c(t) = \det\left(1+\frac{{\rm i}t}{2\pi}J\right) = c_0 + c_1 t + c_2 t^2 + \dots\,.
\end{gather}
Written in terms of the curvature form, the corresponding Chern forms are given by
\begin{align}
    c_0 &= 1 ~,\\
    c_1 &= \frac{\rm i}{2\pi}{\rm Tr}\, \mathcal{R} ~,\\
    c_2 &= \frac{1}{2(2\pi)^2}({\rm Tr}\, \mathcal{R}^2-({\rm Tr}\, \mathcal{R})^2) ~,\label{eq:c2}\\
    c_3 &= \frac{1}{3}c_1\wedge c_2 +\frac{1}{3(2\pi)^2}c_1 \wedge {\rm Tr}\, \mathcal{R}^2-\frac{\rm i}{3 (2\pi)^3}{\rm Tr}\, \mathcal{R}^3 ~.\label{eq:c3}
\end{align}
For complex $n$-dimensional manifolds, $X$, $\int_Xc_n$ is the Euler characteristic so the top Chern class, $c_n$, is also the Euler (curvature) density $e(J)=c_n(J)$.
For $\dim(X)=2$ CY, $c_2$ may be further identified with the standard volume form multiple of the Kretschmann invariant of Ricci-flat metric, the tensor norm-square of the Riemann tensor; see~\eqref{e:chiK3Kretschmann}, below. As an additional check on Ricci-flatness, this approximation has been studied in~\cite{Headrick:2005ch}.

Restricting to the Calabi--Yau twofold and threefold examples of relevance to this paper, we have that the Euler densities simplify for the Ricci-flat metric due to the condition $c_1=0$.
For K3 we expect the Euler density to be 
\begin{equation}\label{e:chiK3Kretschmann}
    c_2(J^\textsf{CY})= \frac{1}{2(2\pi)^2}{\rm Tr}\, \mathcal{R}^2 = \frac{1}{8\pi^2}\sqrt{g}~R_{a\bar{b}c\bar{d}}\,R^{a\bar{b}c\bar{d}}~{\rm d}^4\!z ~,
\end{equation}
and similarly for a Calabi--Yau threefold
\begin{equation}
    c_3(J^\textsf{CY})=-\frac{\rm i}{3 (2\pi)^3}{\rm Tr}\, \mathcal{R}^3 ~.
\end{equation}
Note however that the above expressions only hold for the Ricci-flat Calabi--Yau metric.
Since we only have numerical approximations to this and since we're interested in checking how well-defined is the resulting metric approximated using neural networks (See Section~\ref{sec:spectralNNs}), our curvature density estimates are always obtained by means of~\eqref{eq:c2} and~\eqref{eq:c3}.

\subsection{Topological checks and the curvature distribution}
In the pursuit of computing the Ricci-flat metric by varying an initial choice such as the Fubini--Study metric on the embedding projective space, it behooves to verify that computationally feasible and otherwise known quantities, such as the Euler number, continue to be evaluated accurately.

For complex surfaces, the second Chern class is the Euler density, which is given by the following:
\begin{equation}
    \chi_{_{\rm E}} = \int_{\rm K3} c_2 = \frac{1}{2(2\pi)^2}\int_{\rm K3} ({\rm Tr}\, \mathcal{R}^2-({\rm Tr}\, \mathcal{R})^2) = \int_{\rm K3}\!{\rm d}^4\!z\, \sqrt{g}\rho - \frac{1}{2(2\pi)^2}\int_{\rm K3} \left({\rm Tr}\, \mathcal{R}\right)^2,
 \label{e:chiK3}
\end{equation}
where $\rho$ is the Kretschmann scalar, which may be more familiar from general relativity, where it is used to distinguish coordinate singularities from physical singularities, and is analogous to the $F_{\mu\nu}F^{\mu\nu}$ term in gauge theory.
 In nearly singular hypersurfaces $X_{\l^\sharp+\e}$, the Euler density receives significant contributions from the vicinity of the vanishing cycles, these being heavily curved and nearly singular. For small enough but nonzero $\e$, these regions are also well separated since they limit to isolated singular points in the $\e\,{\to}\,0$ limit, and it is possible to exclude the contribution from these large curvature regions. 
 
The numerically computed Euler density distribution indeed turns out to be heavily peaked near zero, indicating that a relatively large portion of the hypersurface harbors little curvature. The distribution however also has a long and thin ``tail,'' indicating rather small (rarely sampled) but highly curved regions. The bulk of the contributions sampled at essentially randomly distributed points therefore misses these ``curvature peaks.'' For example, in the actually singular hypersurfaces at the special choices $\l\,{=}\,\l^\sharp$, the Euler number computation using the Fubini--Study metric finds:
\begin{enumerate}\itemsep=-1pt\vspace*{-2mm}
 \item At $\l=\frac34$, $\chi(X_{\frac34})_\textsf{FS}\approx8$,
  missing the contribution of $8$ singular points;
 \item At $\l=1$, $\chi(X_1)_\textsf{FS}\approx-8$,
  missing the contribution of $16$ singular points;
 \item At $\l=\frac32$, $\chi(X_{\frac32})_\textsf{FS}\approx0$,
  missing the contribution of $12$ singular points;
 \item At $\l=3$, $\chi(X_{\frac34})_\textsf{FS}\approx16$,
  missing the contribution of $4$ singular points.
 \vspace*{-2mm}
\end{enumerate}
This is consistent with the fact that all singular points are $A_1$-singularities, the blow-up of each of which contributes $\chi(\IP^1)\,{=}\,2$, completing the result to $\chi(X_\l)\,{=}\,24$ (the ``stringy,'' \textit{i.e.}, \emph{equivariant} Euler number~\cite{Atiyah:1989ty}) even for these singular quartics.

 The numerical computation of the Euler density and its integral (the Euler number) to a predetermined precision shows that a significantly larger number of sampling points is required for the nearly singular $X_\l$ than for $X_\l$ with $\l$ far from the four special values, $\l^\sharp$. 

\subsection{Some useful results for (singular) K3s}
The singular K3 manifolds considered in this work are $\mathbb{P}^3_{\mathbb{C}}$ embedded projective surfaces with isolated singularities. Since the singularities prevent definition of smooth forms (such as curvature form $J$  on $X$), we may consider the smooth locus $X_s\subseteq X$ and induce a curvature form $J$ using the pullback of Fubini--Study metric on $\mathbb{P}_\IC^3$. This action, however bears the cost of $c_n \in \Omega^{n,n}(X_s)$ no longer carrying the topological information of $X$.  In order to study the topological Euler characteristic of the resulting variety, we instead consider Chern--Schwartz--MacPherson classes $c_\text{SM}$~\cite{10.2307/1971080, 10.1155/S1073792894000498, 10.2307/118036}. In particular, we have the following relation between the $c_\text{SM}(X)$ and the topological Euler characteristic $\chi(X)$ of a possibly singular variety $X$~\cite{helmer2016algorithms}:
\begin{gather}\label{eq:csmEuler}
	\chi(X)	 = \deg{c_\text{SM}(X)} ~.
\end{gather}
Furthermore, we may relate the Fulton class $c_F(X)$ to the Chern--Schwartz--MacPherson characteristic class $c_\text{SM}(X)$ using~\cite{aluffi2019chern, parusinski1998characteristic}:
\begin{gather}\label{eq:milnorClass}
	(-1)^{\dim{X}}\deg(c_F(X) - c_\text{SM}(X)) = \sum_{p\in \mathrm{Sing}~{X}}\mu_X(p) ~,
\end{gather}
where $\mu_X(p)$ is the Milnor number of the singularity $p\in \mathrm{Sing}~X$. The actual Euler number can be computed with aid of Proposition~\ref{prop:one}, which we prove.

\begin{prop}\label{prop:pfaffian} Let $X\subseteq \mathbb{P}_\IC^3$ be a possibly singular projective surface with curvature form $J$ defined on the smooth locus $X_s$ of $X$. If $|\mathrm{Sing}~X|<\infty$ and the singularities are of type $A_1$, then:
\begin{gather}
	\int_{X_s}e(J) + 2|\mathrm{Sing}~{X}|= \deg{c_F(X)} ~,
\end{gather}
where $c_F(X)$ is the Fulton class of $X$ and $e(J)$ is given by Pfaffian: $\mathrm{Pf}(J)/(2\pi)^2$. \label{prop:one}
\end{prop}
\begin{proof} Build a stratification of $X$ in the following manner: define $X_0 := X\backslash \mathrm{Sing}~X = X_s$, and for every singular point $p_i\in \mathrm{Sing}~X$, define $X_i :=\{p_i\}$. Note that $X_0$ is the largest stratum, and $X_i$ for $i>0$ are the singular strata. Using the results from~\cite{aluffi2019pfaffian}, we may express the integral of Pfaffian $\mathrm{Pf}(J)$ on $X_s$ in terms of the local Euler obstructions $\mathrm{Eu}_p(X)$~\cite{10.2307/1971080} as:
\begin{gather}
	\frac{1}{4\pi^2}\int_{X_s}\mathrm{Pf}(J) = \chi(X) - \chi(\mathrm{Sing~X}) + \sum_{p\in \mathrm{Sing}X}\mathrm{Eu}_p(X)\,\chi(\{p\}) ~,
\end{gather}
where we have $\chi(X_s) = \chi(X) - \chi(\mathrm{Sing}~X)$ since $X$ is a complex projective variety and $|\mathrm{Sing}~X|<\infty$. Furthermore, $\chi(\mathrm{Sing}~{X}) = |\mathrm{Sing}~{X}|$. Thus, the proof reduces to computation of $\chi(X)$ and the Euler obstructions $\mathrm{Eu}_p(X)$ for $p\in \mathrm{Sing}~X$. Using the results of~\cite{articleCyclesPolaires}, we may express $\mathrm{Eu}_p(X)$ in terms of the Milnor numbers:
\begin{gather}
	\mathrm{Eu}_p(X) = 1- (-1)^{\dim{X}}\mu_{X\cap H}(p) ~,
\end{gather}
where $H$ is a generic hyperplane through $p\in \mathrm{Sing}~X$. In case of $\dim_{\IC}X = 2$, for $A_1$ type singularities $p\in \mathrm{Sing}~X$, we have $\mu_{X\cap H}(p) = 1$ (see~\cite{siersma2022polar, huh2014milnor}), thus, the Euler obstructions vanish. Using the identity~\eqref{eq:milnorClass} and the property of Chern--Schwartz--MacPherson classes~\eqref{eq:csmEuler}, we have:
\begin{gather}
	\chi(X) = \deg{c_\text{SM}(X)}	 = \deg{c_F(X)} - \sum_{p\in \mathrm{Sing}~X}\mu_X(p) = \deg{c_F(X)} - |\mathrm{Sing}~X| ~.
\end{gather}
Thus, by combining the results, we obtain:
\begin{gather}
	\frac{1}{4\pi^2}\int_{X_s}\mathrm{Pf}(J) = \deg{c_F(X)} - 2|\mathrm{Sing}~X| ~,
\end{gather}
which concludes the proof.
\end{proof}

It is easy to see that the singularities of singular $X_\lambda$ varieties in the Cefal\'{u} pencil ($\lambda <\infty$) are of type $A_1$\footnote{By symmetry of the defining polynomial, all singularities have equal Milnor number; moreover, it equals $1$, in agreement with the analysis in Section~\ref{s:deformations}.}. Furthermore, we may regard the Cefal\'{u} pencil as a smoothing family near each singular $X_{\lambda'}$. Using~\cite{aluffi2019pfaffian}, we have the following identity:
\begin{gather}
	\lim_{\lambda\rightarrow \lambda'}\lim_{\delta\rightarrow 0}	 \int_{X_\lambda \cap N_\epsilon(\mathrm{Sing}~X)}e(J_{\lambda}) = (-1)^2\left[\sum_{p\in \mathrm{Sing}~X}\mu_X(p)+ \mu_{X\cap H}(p)\right] = 2|\mathrm{Sing}~X| ~.
\end{gather}
From~\cite{https://doi.org/10.1002/cpa.3160310304, doi:10.1073/pnas.74.5.1798, Guggenheimer1952, roulleau2019generalized} it is known that for surfaces admitting K\"{a}hler--Einstein metric with curvature form $J$, the characteristic forms $c_2(J)$ and $c_1(J)$ satisfy the Bogomolov--Yau (BY) inequality:
\begin{gather}\label{eq:byauIneq}
	3c_2(J) - c_1(J)^2\geq 0 ~.
\end{gather}

\section{Numerical methods}\label{sec:nm}

\subsection{Calabi--Yau metrics from neural networks}

Here we briefly review the existing paradigm of approaches based on direct prediction of the metric tensor by a neural network~\cite{Ashmore:2019wzb, Anderson:2020hux, Douglas:2020hpv, Jejjala:2020wcc, Larfors:2021pbb}.
Focusing, for example, on the Calabi--Yau threefold case, the metric tensor in local coordinates has a representation as a Hermitian $3 \times 3$ matrix, with three independent off-diagonal complex parameters and three independent diagonal real parameters.
The first order approach would be to have the neural network predict nine real components in local coordinate charts such that the metric tensor satisfies the Ricci-flatness conditions on the Calabi--Yau geometry.
However, this may be overly general, and we can consider more restrictive ans\"atze (classes of functions which the neural network may represent) that may help the approximation.

By Yau's theorem~\cite{Yau:1978}, any compact K\"ahler manifold 
of complex dimension $n$ with vanishing first Chern class admits a unique Ricci-flat metric in each of its $h^{1,1}$ K\"ahler classes.
That is to say, given a choice of reference metric $g_{\textsf{ref}}$ with K\"ahler form $J^{\textsf{ref}}$ on the Calabi--Yau $X$, there exists a K\"ahler form $J^{\textsf{CY}}$ associated with the unique Ricci-flat metric. $J^{\textsf{CY}}$ is cohomologous to $J^{\textsf{ref}}$ and hence given by an exact correction:
\begin{equation}
    \label{eq:correction}
    J^{\textsf{CY}} = J^{\textsf{ref}} + {\rm i}\partial \bar{\partial} \phi ~.
\end{equation}
Here, $\phi \in C^{\infty}(X)$ 
is a global smooth real function on $X$.
If $X$ is embedded into a projective space (or a product of projective spaces), the reference K\"ahler form on $X$ may be taken to be the pullback of the canonical Fubini--Study form in the ambient space, $J^{\textsf{FS}}$.
To approximate the true Ricci-flat metric on $X$, one may estimate the correction term in~\eref{eq:correction} by employing a neural network to model $\phi_{\textsf{NN}}: X 
\rightarrow \mathbb{R}$ via a patch invariant scalar function using local coordinate points, and then subsequently pulling back to $X$.
Assuming one has found an embedding of the Calabi--Yau $X$ in the ambient space, $\iota: X \hookrightarrow \mathcal{A}$, the Ansatz for the predicted metric $g_{\textsf{CY}}$ takes the following form:\footnote{With a slight abuse of notation,~\eref{eq:ddphi_model} should be interpreted component-wise.} 
\begin{equation}
    \label{eq:ddphi_model}
    g_{\textsf{CY}} \triangleq \iota^*g_{\textsf{FS}} +  \partial \bar{\partial} \phi_{\textsf{NN}} ~.
\end{equation}
By computing the eigenvalue distribution (see Appendix~\ref{sec:psh}), we verify that the so-defined metric, $g_{\textsf{CY}}$, continues to be positive (Riemannian), but defer the analysis\footnote{We thank Oisin Kim for constructive discussion on this subtlety.} of the necessary and sufficient conditions for that throughout the computational framework. While positivity is not explicitly enforced, the metric is symmetric by construction and its determinant is encouraged to be positive by the objective function. If $\phi$ is a global function, the associated \kae form is remains in the same \kae class as $J^{\textsf{ref}}$.
However, if $\phi$ changes between coordinate patches, then this is not guaranteed, and further precautions must be taken to preserve the \kae class.

To compute the approximations of the Calabi--Yau metric for twofolds, we use the \texttt{cymetric} library~\cite{Larfors:2021pbb, cymetric}.
For Calabi--Yau threefolds, we also use a custom implementation of the routines in the \texttt{cymetric} library in \texttt{JAX}~\cite{jax2018github}, with extended functionality for computing topological quantities. The \texttt{cymetric} package supports different approaches to numerically approximating the Ricci-flat metric. The model~\eqref{eq:ddphi_model} is known as the \texttt{PhiModel}. We use the \texttt{PhiModel} to obtain metric approximations for the K3 and quintic examples and the \texttt{JAX} implementation for the quintic examples only. 
For the underlying neural network model which approximates $\phi$, we use a fully connected/dense network with $3$ hidden layers and $64$ nodes in each hidden layer; see Figure~\ref{fig:NNArchitecture}.\footnote{We train for $50$ epochs, with batch sizes: $(64, 50000)$ using 
\textsf{Adam} optimizer with the default parameters. We use \textsf{gelu} activation functions.}

A variation of this algorithm is to make $\phi_\textsf{NN}$ invariant under $\mathbb{C}^*$ from the beginning. Instead of taking real and imaginary parts of the homogeneous coordinates as inputs, one can take the following: $z_i\bar{z}_i/|z|^2$, ${\rm Re}(z_i\bar{z}_j/|z|^2)$, ${\rm Im}(z_i\bar{z}_i/|z|^2)$  ($j<i$, where $i,j=1,...,N+1$). Thus, instead of $2(N+1)$ we have $(N+1)^2$ inputs for this modified neural network. This results in a globally defined function $\phi_\textsf{NN}$ on $\mathbb{P}_\IC^N$. We briefly describe the results derived from this method in Section~\ref{sec:spectralNNs}, and leave a more detailed survey of these techniques for an upcoming publication~\cite{spectralNetworks}.

\tikzstyle{node}=[very thick,circle,draw=black,minimum size=22,inner sep=0.5,outer sep=0.6]
\tikzstyle{connect}=[->,thick,black,shorten >=1]
\tikzset{
  node 1/.style={node,black,draw=black},
  node 2/.style={node,black,draw=black},
  node 3/.style={node,black,draw=black},
}
\def\nstyle{int(\lay<\Nnodlen?min(2,\lay):3)}

\begin{figure*}[htb]
	\centering
	\begin{tikzpicture}[x=2.4cm,y=1.2cm]
  \readlist\Nnod{3,4,4,1}
  \readlist\Nstr{2n-1,,}
  \readlist\Cstr{z,h,\phi}
  \def\yshift{0.55}
  \foreachitem \N \in \Nnod{
    \def\lay{\Ncnt}
    \pgfmathsetmacro\prev{int(\Ncnt-1)}
    \foreach \i [evaluate={\c=int(\i==\N); \y=\N/2-\i-\c*\yshift;
                 \x=\lay; \n=\nstyle;
                 \index=(\i<\N?int(\i-1):"\Nstr[\n]");}] in {1,...,\N}{
      	\node[node \n] (N\lay-\i) at (\x,\y) {
      	 \ifnum \lay < 2
      		$\strut\Cstr[\n]_{\index}$
      	  \fi
      	  \ifnum \lay > 3
      		$\strut\Cstr[\n]_{\index}$
      	  \fi
      };
      \ifnumcomp{\lay}{>}{1}{
        \foreach \j in {1,...,\Nnod[\prev]}{
          \draw[white,line width=1.2,shorten >=1] (N\prev-\j) -- (N\lay-\i);
          \draw[connect] (N\prev-\j) -- (N\lay-\i);
        }
        \ifnum \lay=\Nnodlen
          \draw[connect] (N\lay-\i) --++ (0.5,0);
        \fi
      }{
        \draw[connect] (0.5,\y) -- (N\lay-\i);
      }
    }
    \ifnum \lay < 4
    	\path (N\lay-\N) --++ (0,1+\yshift) node[midway,scale=1.6] {$\vdots$};
    	\ifnum \lay > 1
    		\path (N\lay-\N) --++ (0,1+\yshift) node[left,pos=0.44,scale=1.1] {64};
    	\fi
    \fi
  }
  
  \node[below=3,align=center] at (N1-1.5) {Input layer:\\[-0.2em] $p\in X\subseteq \mathbb{P}_\IC^n$};
  \node[below=1,align=center] at (N4-1.1) {$g = P^*g_\mathrm{FS} + P^*\partial\overline{\partial}\phi$};
\end{tikzpicture}
\caption{Neural-network architecture for building the \texttt{PhiModel} using \texttt{cymetric}.}\label{fig:NNArchitecture}
\end{figure*}
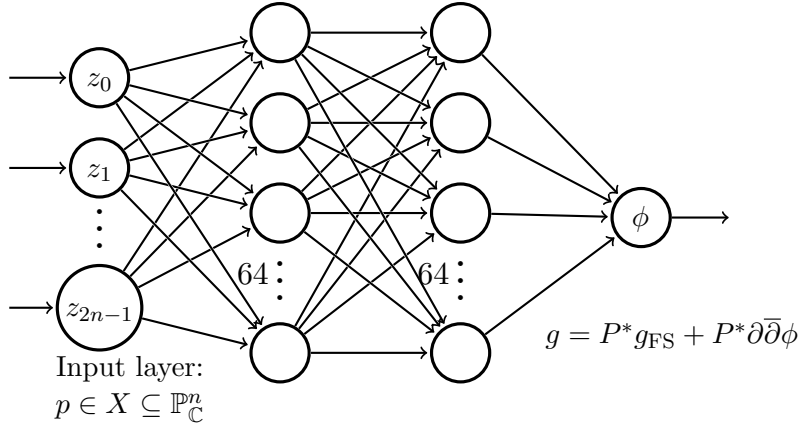


\subsection{Persistent homology}
Let $f: \Sigma\to \mathbb{R}$ be a function on a simplicial complex $\Sigma$ such that whenever $\sigma_1$ is a facet of $\sigma_2$ in $\Sigma$, $f(\sigma_1) \le f(\sigma_2)$.
For $x\in \mathbb{R}$, define the level set $\Sigma_x = f^{-1}(I_x)$, with $I_x = (-\infty,x]$; this is a subcomplex of $\Sigma$.
The ordering of values of $f$ on the simplices in $\Sigma$ defines a filtration,
\begin{equation}
\emptyset = \Sigma_0 \subseteq \Sigma_1 \subseteq \ldots \subseteq \Sigma_n = \Sigma ~.
\end{equation}
For $p\le q$, we have the $k$-th persistent homology group $H_k^{p,q}(\Sigma)$ that is induced by the inclusion $\Sigma_p \hookrightarrow \Sigma_q$.
In particular, we have homomorphisms $f^{p,q}_k: C_k(\Sigma_p) \to C_k(\Sigma_q)$ modulo boundaries, with $C_k(\Sigma_p)$ the free Abelian group generated by $k$-simplices in $\Sigma_p$.
The \emph{persistent homology groups} are the images of these homomorphisms, and the Betti numbers $b_k^{p,q}$ are the dimensions of these groups.
Proceeding across the filtration, topological features are born and die.
For instance, connected components may be added to the space, cycles can form or be filled in, etc.
The \emph{barcode} is a graphical way of visualizing this information.
Persistent homology provides a microscope that images the shape of a dataset and is a key tool in topological data analysis.
We use this to identify high curvature regions on Calabi--Yau spaces.

\section{Results}\label{sec:results}
In our analysis we compute the various curvature forms using both the pullback of the Fubini--Study metric from the ambient space (FS) and the machine learned approximation to the Ricci-flat metric (ML).
This approach has various advantages.
First, having the explicit expressions for the Chern forms in the case of the Fubini--Study metric, we can ensure that the numerical errors in the topological computations are solely due to the Monte Carlo (MC) integration.
The Fubini--Study metric then serves to check convergence of the integration as we increase the number of points.
This also provides a consistency check for the computations leading to the Euler number in the various examples considered.
These results clarify which limitations are due to the numerical integration and which limitations are due to the machine learned approximation.
\begin{figure}[htb]
	\centering
	\includegraphics[width=0.7\textwidth]{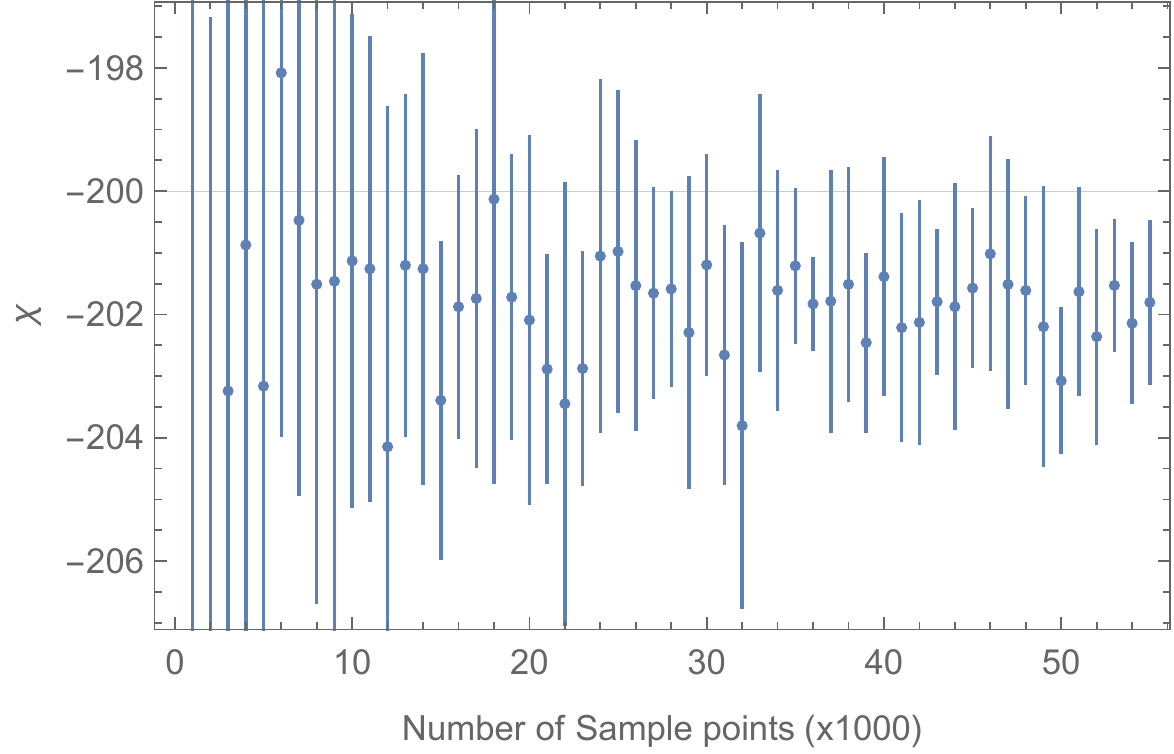}
 \vspace*{2mm}
	\caption{Euler number for the Fermat quintic $Y_{\psi=0}$ computed using different numbers of sample points.}
  \label{fig:EulerFSQuintic}
\end{figure}

As an \textit{hors d'\oe uvre} we present the convergence results for the Monte Carlo integration for the Euler density of the Fubini--Study metric.
Consider the Fermat quintic $Y_{\psi=0}$ and take a dataset of $100,000$ points uniformly distributed with respect to the pullback Fubini--Study metric.
Then we take a subsample of a given size, ranging from $1,000$ to $55,000$, and for each subsample size we repeat the experiment ten times.
For a given subsample, we compute the Euler density and out of the ten repetitions we obtain an estimate of the error (by means of the standard deviation).
The results are shown in Figure~\ref{fig:EulerFSQuintic}.
We observe that the variance reduces as the number of points is increased.
For $100,000$ points we obtain an Euler number of $\chi=-196.43$, roughly $2\%$ off the expected value.
In the simplest compactifications of the heterotic string, the number of generations of particles in the low energy spectrum is given by the index of the Dirac operator and is $\frac12 |\chi|$.
Being off by even $1\%$ for the quintic reports the wrong number of families of elementary particles (in the standard embedding of Heterotic String compactifications).

\subsection{Characteristic forms on the Cefal\'{u} pencil}\label{sec:pencil}

Let $J_\lambda$ be a curvature form on $X_\lambda$. 
 In the case of singular $X_\lambda$,~\eqref{e:chiK3} is no longer true.
 In studies of moduli dependent metrics~\cite{Headrick:2005ch, Anderson:2020hux}, it has been observed that the accuracy of the approximately flat metric is moduli dependent.
 These ideas prompt us to analyze the performance of neural network based approaches for studying Calabi--Yau spaces which are singular or nearly singular.
 In particular, we numerically compute the integral~\eqref{e:chiK3} using the curvature forms $J_\lambda^\textsf{FS}$ and $J_\lambda^\textsf{CY}$ calculated using the induced and numerical Calabi--Yau metrics, respectively. The numerical values of the Euler characteristic for different $X_\lambda$ in the vicinity of singular $X_{\l^\sharp}$ are shown in Figure~\ref{fig:eulerChar}.

\begin{figure*}[htb]
	\centering
	\includegraphics[width=0.8\textwidth]{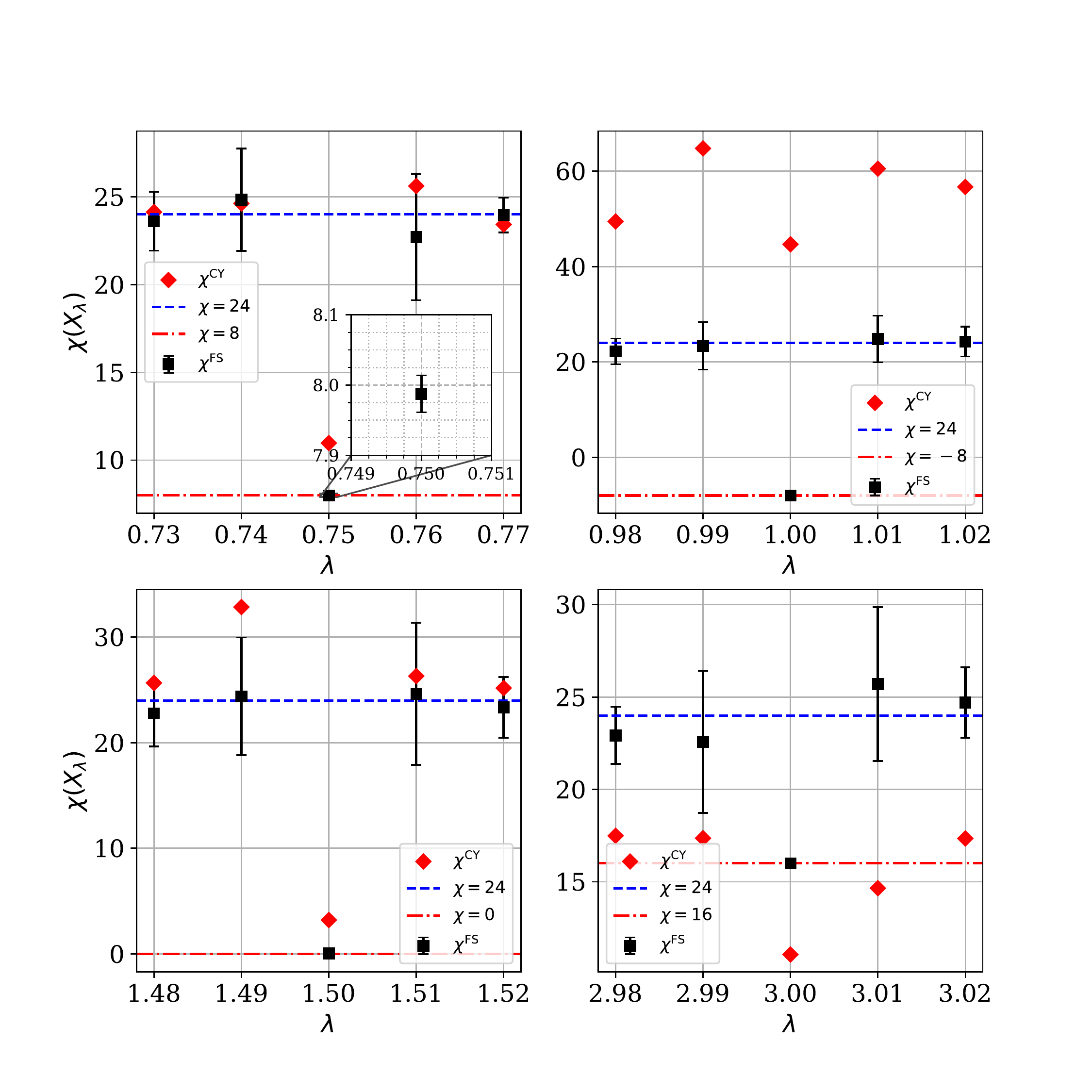}
	\caption{Numerical values of~\eqref{e:chiK3} along the Cefal\'{u} pencil. Black points and error bars showing a $95\%$ confidence interval are associated to Fubini--Study results, while the red dots correspond to the machine learned metric approximation using fully-connected networks. See the Appendix for details on integration.}\label{fig:eulerChar}
\end{figure*}

Using the fully-connected network approximation of $\phi$, we observe a decreasing trend of the accuracy as we approach the singular $X_{\l^\sharp}$ varieties in the pencil. In particular, we note that the Euler characteristic computed using the machine learned approximation to the Calabi--Yau metric deviates significantly from the expected value of $24$ by a margin that is too large to be attributed to the numerical Monte Carlo integration. Furthermore, we see a significant discrepancy between the numerical values computed using $J_\lambda^\textsf{FS}$ and $J_{\lambda}^\textsf{CY}$ at the singular $X_\lambda$. From Figure~\ref{fig:eulerChar}, we notice that for the cases $\lambda=1$ and $\lambda=3$ we obtain the most significant discrepancies between Fubini--Study and machine learned results.

In order to make sense of the results for singular $X_\lambda$'s, let us consider the degrees of the Chern--Schwartz--MacPherson classes which are presented in Table~\ref{table:fultonCSMCefaluPencil} and the Monte Carlo computations for the Fubini--Study metric shown in Table~\ref{table:c1c1c2FS}. To compute the $c_\text{SM}(X_\lambda)$ and $c_F(X_\lambda)$ we use Macaulay2~\cite{M2}. From these tables we can see that the numerical results are in agreement with Proposition~\ref{prop:one}. 
\setlength{\abovecaptionskip}{10pt}
\begin{table}[htb]
\centering
\begin{tabular}{||c c c c||} 
 \hline
$\lambda$ & $\deg{c_\text{SM}(X_\lambda)}$ & $\deg{c_F}(X_\lambda)$ & $|\mathrm{Sing}~{X_\lambda}|$\\ [0.5ex] 
 \hline\hline
 0 &24&24&0\\
3/4&16&24& 8 \\
1&8&24&  16\\
3/2&12&24&  12\\
3&20&24&4\\[1ex] 
 \hline
\end{tabular}
\caption{Degrees of Fulton and Chern--Schwartz--MacPherson classes for Cefal\'{u} pencil.}
\label{table:fultonCSMCefaluPencil}
\end{table}
\setlength{\abovecaptionskip}{-10pt}

\setlength{\abovecaptionskip}{10pt}
\begin{table}[htb]
$$
\begin{array}{||@{~~}c@{~~} c@{~~} c@{~~} c@{~~} c@{~~}||} 
 \hline
\lambda & \#~\text{of sing.\;pt.}
& \deg c_2(J_{\lambda}^\textsf{FS}) & \deg c_1(J_{\lambda}^\textsf{FS})^2 & \deg (3c_2(J_\lambda^\textsf{FS}) - c_1(J_{\lambda}^\textsf{FS})^2)\geq 0 \\ [0.5ex] 
 \hline\hline
 0   &0  &\mB{24}  &\mB{0}  & \text{True}\\
 3/4 &8  &\mB{7.99\pm 0.03}   &\mB{-16.0\pm 0.2}& \text{True}\\
 1   &16 &\mB{-7.99\pm 0.08}  &\mB{-31.9\pm 0.3}& \text{True}\\
 3/2 &12 &\mB{0.0\pm 0.1}   &\mB{-23.9\pm 0.3}& \text{True}\\
 3   &4  &\mB{~16.00\pm 0.09} &\mB{-8.0\pm 0.1} & \text{True}\\[1ex] 
 \hline
\end{array}
$$
\caption{Values of the integrals of the possible top characteristic forms on $X_\lambda$. The integrals were evaluated using MC integration. The uncertainties correspond to $95\%$ confidence interval.}
\label{table:c1c1c2FS}
\end{table}
\setlength{\abovecaptionskip}{-10pt}

We also notice that the Fubini--Study results satisfy the Bogolomov--Yau inequality~\eqref{eq:byauIneq}. Note, however, that for singular $X_\lambda$, we have a non-zero value for the numerical approximation of the integral of the first Chern class.
That is,
\begin{gather}\label{eq:c1c1FS}
	\int_{X_{\lambda^\sharp}\backslash \mathrm{Sing}~X_{\lambda^\sharp}} c_1(J_{\lambda}^\textsf{FS})^2 \neq 0 ~,
\end{gather}
as shown in Table~\ref{table:c1c1c2FS}. In addition, note that Table~\ref{table:c1c1c2FS} exhibits the following property:
\begin{gather}\label{eq:conj1}
	\int_{X_{\lambda^\sharp}\backslash \mathrm{Sing}~X_{\lambda^\sharp}}c_2(J) - c_1(J)^2 = 24 = \chi(\text{K3}) ~,
\end{gather}
which prompts us to formulate the following conjecture.
\begin{conj}\label{conj1}
Let $X\subseteq\mathbb{P}_\IC^3$ be a possibly singular K3 surface, whose smooth locus $X_s$ has curvature form $J$ induced by the Fubini--Study metric on $\mathbb{P}_\IC^3$. If the singularities of $X$ are isolated and of type $A_1$, then~\eqref{eq:conj1} holds true.
\end{conj}

For a singular algebraic surface $X_{\l^\sharp}$ with $\l^\sharp<\infty$, the crepant ($c_1$-preserving) desingularization of an isolated $A_1$-singularity replaces it with an exceptional $\IP^1$-like divisor, $S$, with the self-intersection $[S]^2=-2$.
 The $\deg c_1(J_{\lambda}^\textsf{FS})^2$ column in Table~\ref{table:c1c1c2FS} evidently equals the total sum of these isolated contributions --- as if $X_{\l^\sharp}$ was desingularized. We thus arrive at the next conjecture.
\begin{conj}\label{conj2}
 Each $X_{\l^\sharp}$ in~\eqref{e:Cefalu} with $\l^\sharp<\infty$ may be identified with a global finite quotient, and the $\deg c_2(J_{\lambda}^\textsf{FS})$ column contributions to $\chi(\text{K3\/})$ in Table~\ref{table:c1c1c2FS} appropriately display the leading term in~\eqref{e:OrbiEu}.
\end{conj}
\noindent
This corroborates our expectation that $X_{3/2}$ and $X_3$ are also identifiable as global finite $\ZZ_2$ quotients as shown in~\eqref{e:OrbiEu2}.

\begin{figure*}[htb]
	\centering
     \begin{subfigure}[b]{0.24\textwidth}
     	\centering
        \includegraphics[width=\textwidth]{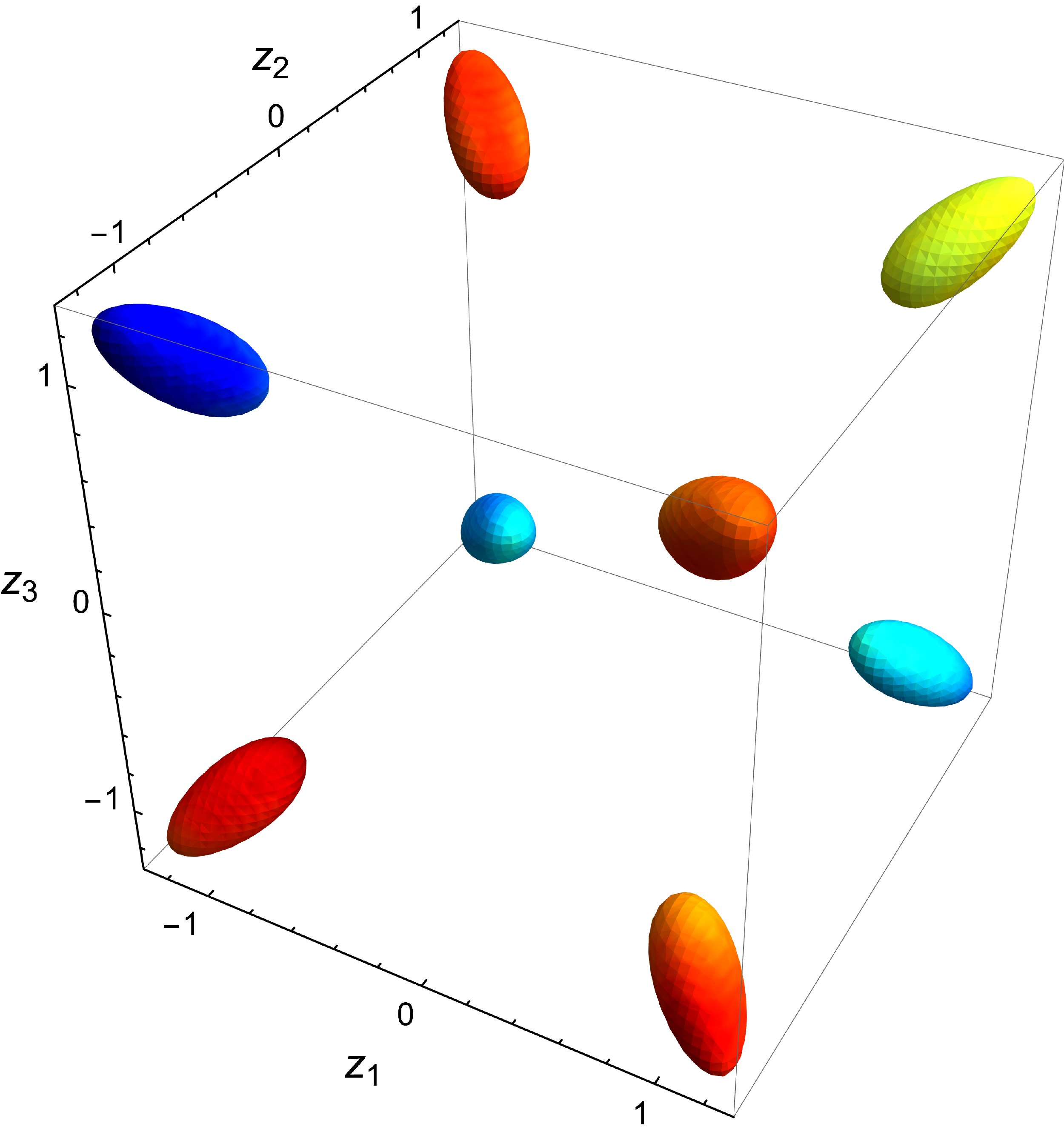}
		\caption{{\small $\lambda= 0.77$}}
	\end{subfigure}
	\begin{subfigure}[b]{0.24\textwidth}  
	\centering 
		\includegraphics[width=\textwidth]{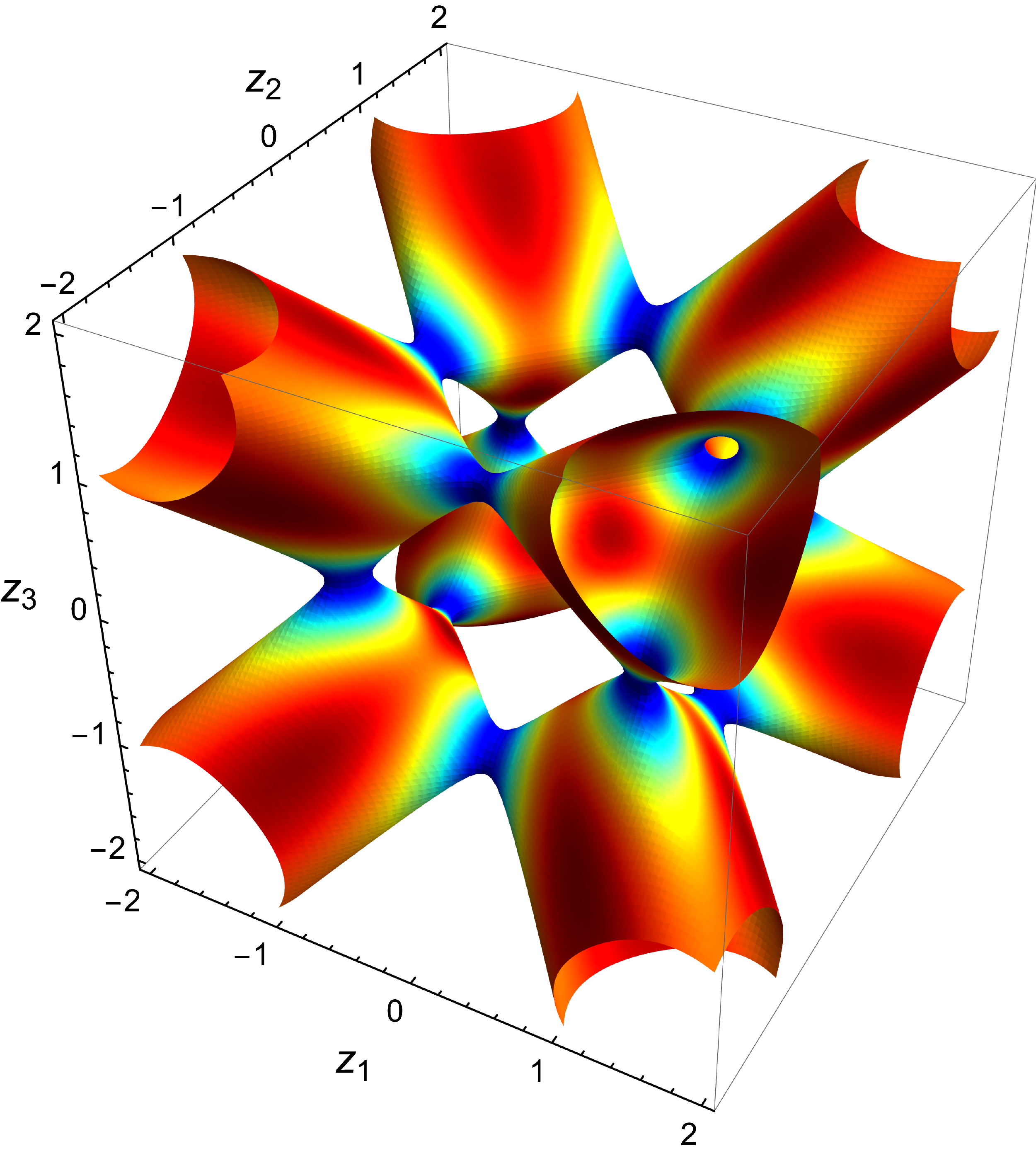}
		\caption{{\small $\lambda = 1.01$}}
	\end{subfigure}
	\begin{subfigure}[b]{0.24\textwidth}   
		\centering 
		\includegraphics[width=\textwidth]{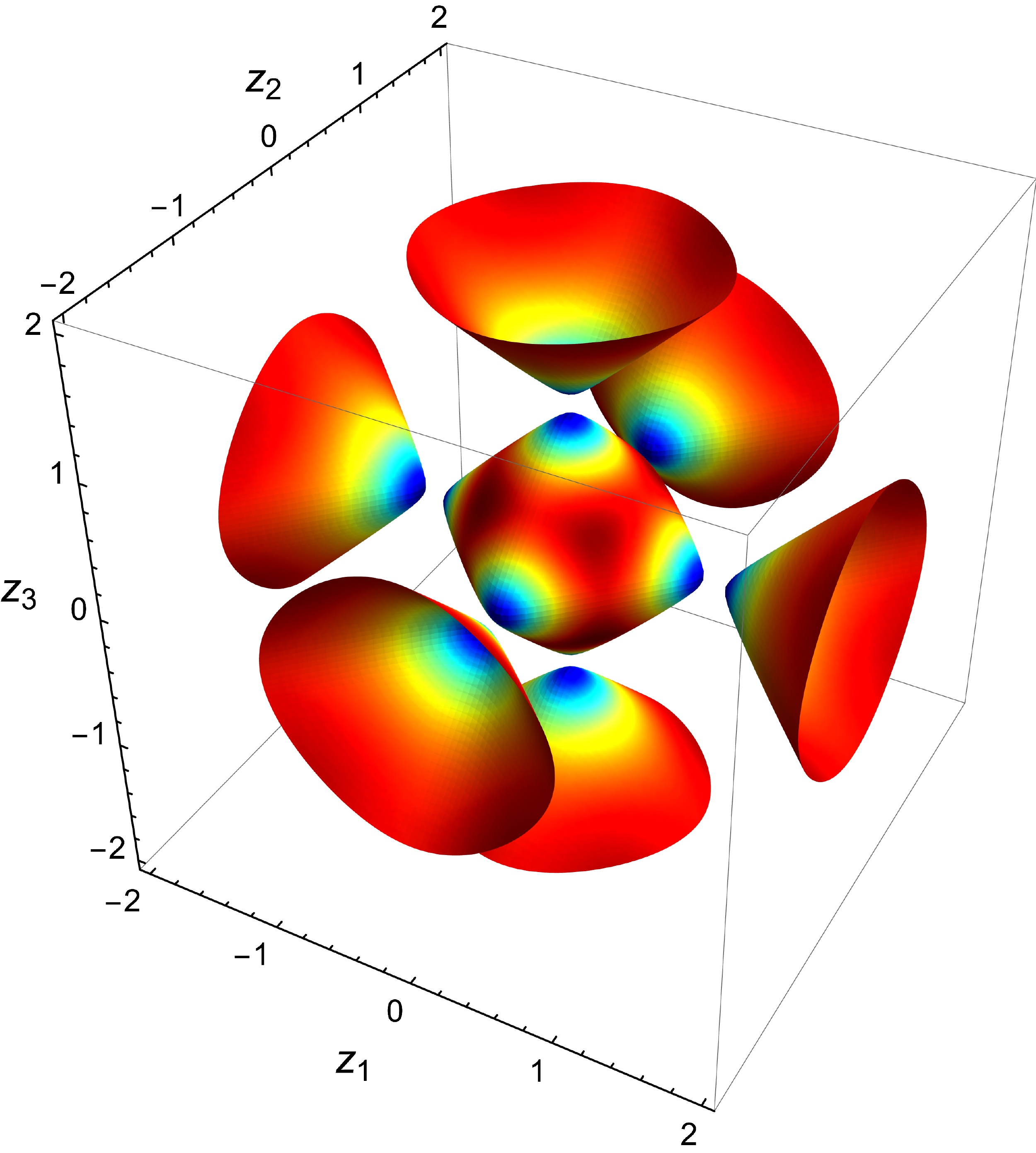}
		\caption{{\small $\lambda = 1.51$}}
	\end{subfigure}
	\begin{subfigure}[b]{0.24\textwidth}   
		\centering 
		\includegraphics[width=\textwidth]{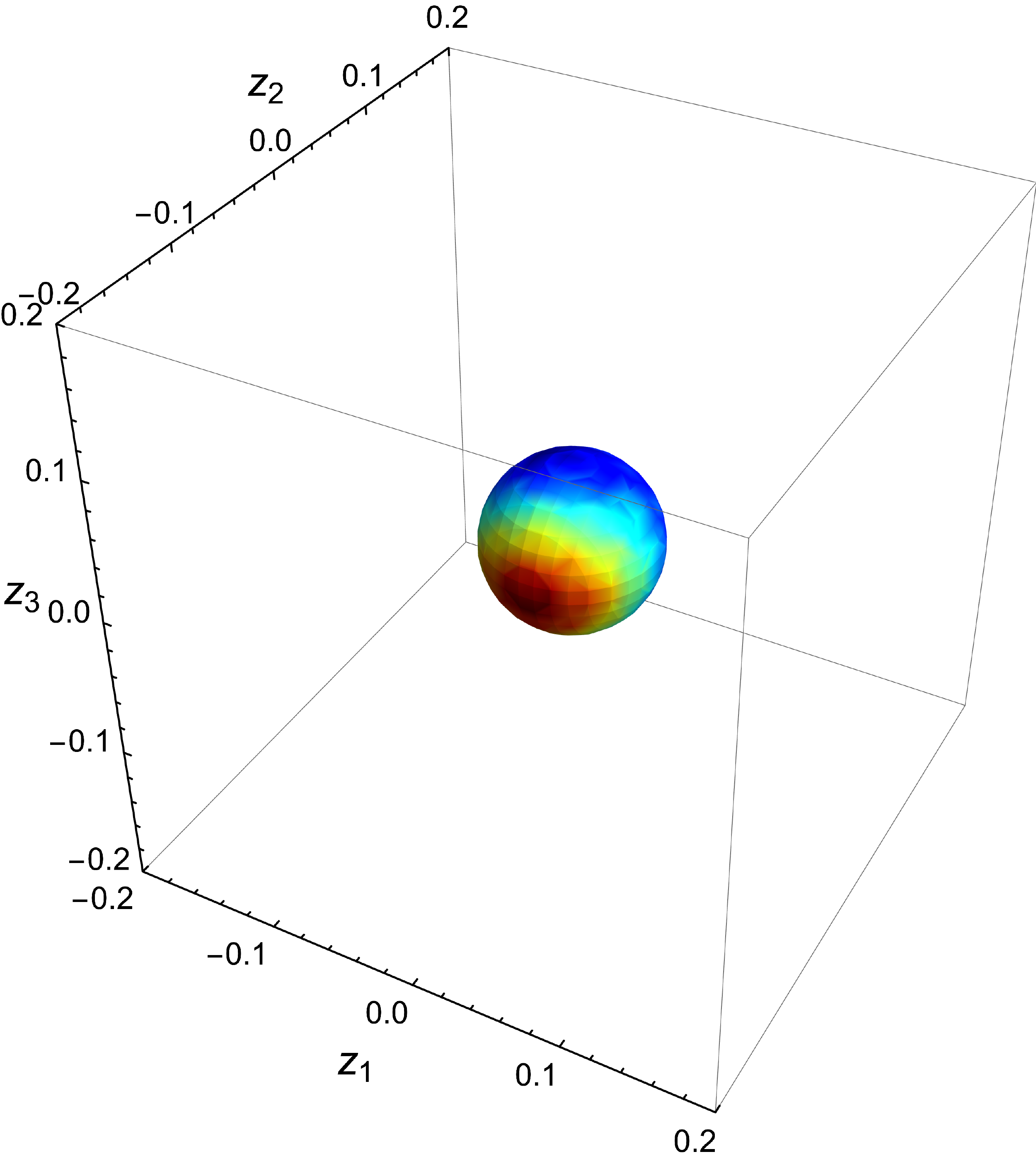}
		\caption{{\small $\lambda = 2.98$}}
	\end{subfigure}
	\vspace*{10mm}
	\caption{Visualization of the real subset of $X_{\lambda^\sharp + \epsilon}$ in a patch $\{z_0 \neq 0\}$. The coloring is defined by the values of the trained spectral network $\phi$. (See Section~\ref{sec:spectralNNs}.)}
 \label{fig:visualizationsCefaluPhi}
\end{figure*}

We note that Conjecture~\ref{conj1} fails if the singular locus of $X$ contains a singularity of dimension greater than zero, which is possible for geometries that combine the Dwork and the Cefal\'{u} deformations.

We shall separately consider each singular $X_{\l^\sharp}$ in the following.
The visualizations of $\phi$ for some near-singular surfaces $X_\lambda$ are shown on Figure~\ref{fig:visualizationsCefaluPhi}.
The training progress is shown on the Figure~\ref{fig:cefalu_sigma_loss}.

\begin{figure*}[htb]
    \centering
    \includegraphics[width=0.7\textwidth]{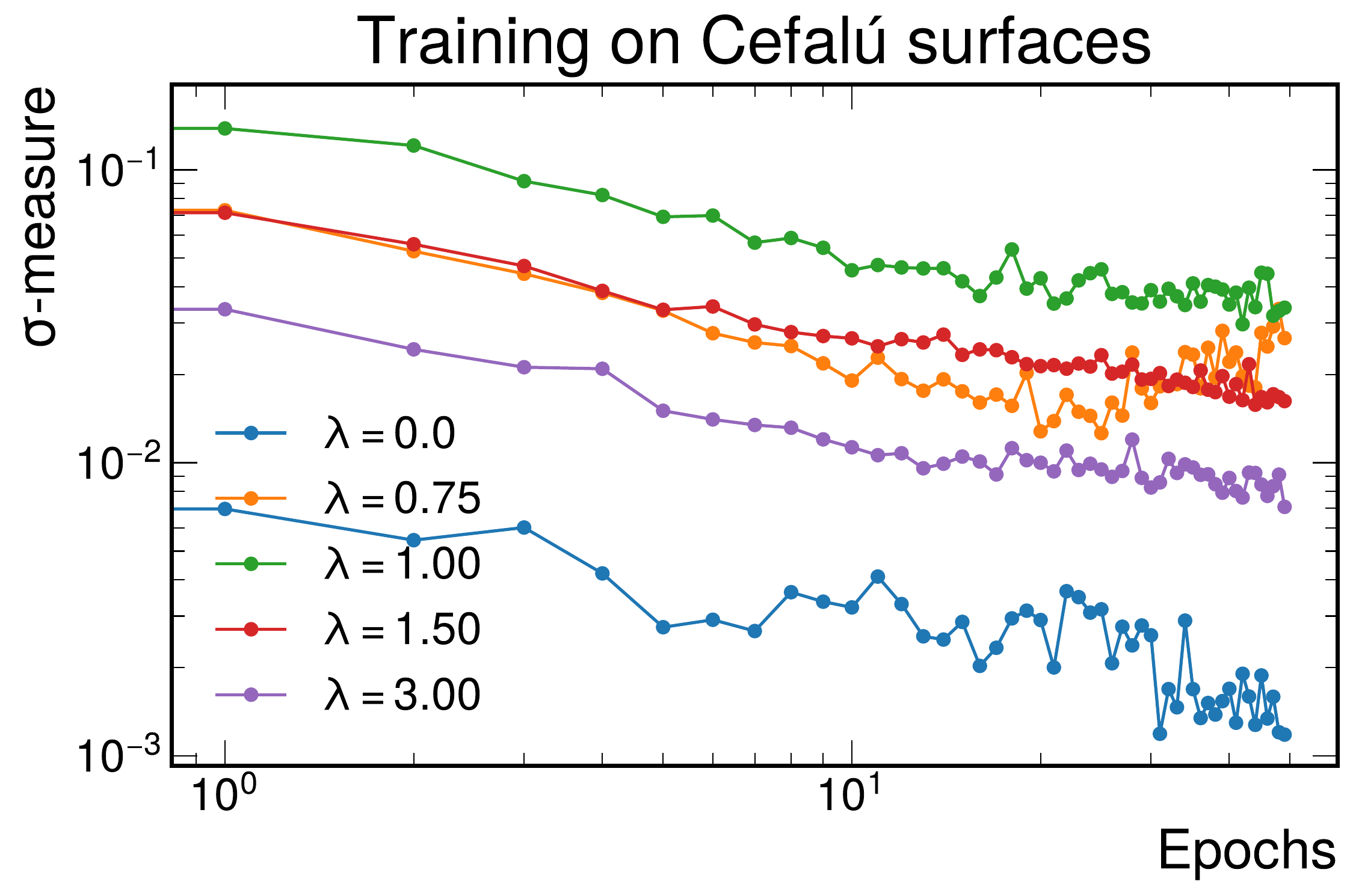}
    \vspace{5mm}
    \caption{Evolution of the $\sigma$ loss when training at $\lambda = \lambda^\sharp$. The $\sigma$-measure is evaluated on validation set.}\label{fig:cefalu_sigma_loss}
\end{figure*}

\subsubsection[\texorpdfstring{$\lambda = 0$}{lambda=0}]{$\bm{\lambda = 0}$}\label{sec:lambda0}

We first consider the smooth K3 obtained from turning off the Cefal\'u deformation.
This supplies a reference point with which to compare calculations on the singular K3 spaces corresponding to the special values $\lambda=\lambda^\sharp \in \{3/4,1,3/2,3\}$. The distribution of the values of $c_1^2$ is shown on the Figure~\ref{fig:quartic_hist_c1c1}. There we see a widespread distribution for the Fubini--Study metric curvature distribution, while for the trained metric the curvature density concentrates in a sharp peak around zero, as expected for a flat metric. Also notice in Figure~\ref{fig:quartic_hist_c2} the Euler density distribution for the learned metric is positive at all points, in accordance to the Bogolomov-Yau inequality. 

\begin{figure*}[htb]
    \centering
    \includegraphics[width=0.7\textwidth]{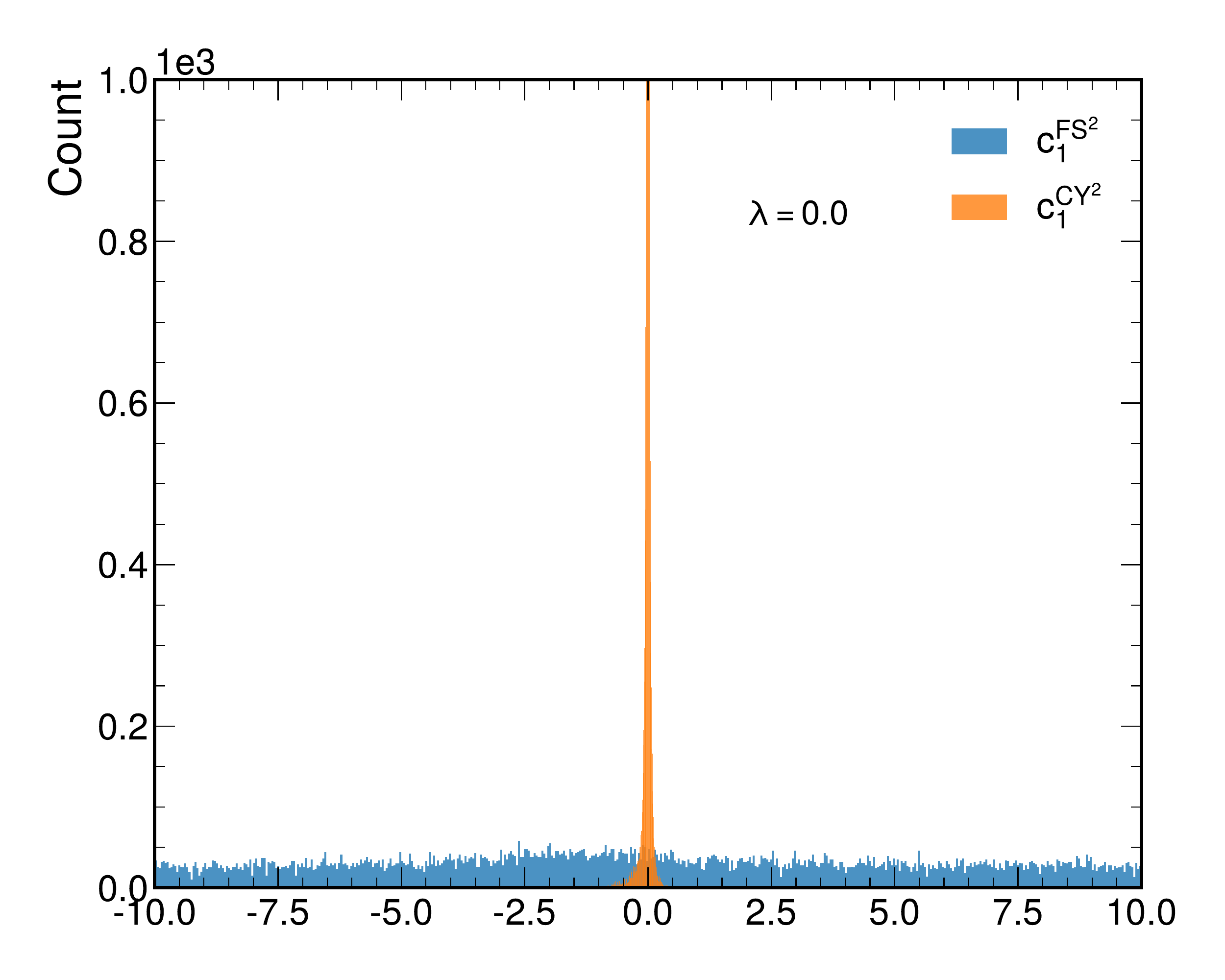}
    \caption{Distribution of the values of $c_1^2$ using both Fubini--Study and machine learned Calabi--Yau metrics for the Fermat quartic.}\label{fig:quartic_hist_c1c1}
\end{figure*}

Similarly, the distribution of the values of $c_2$ is shown on the Figure~\ref{fig:quartic_hist_c2}.

\begin{figure*}[htb]
    \centering
    \includegraphics[width=0.7\textwidth]{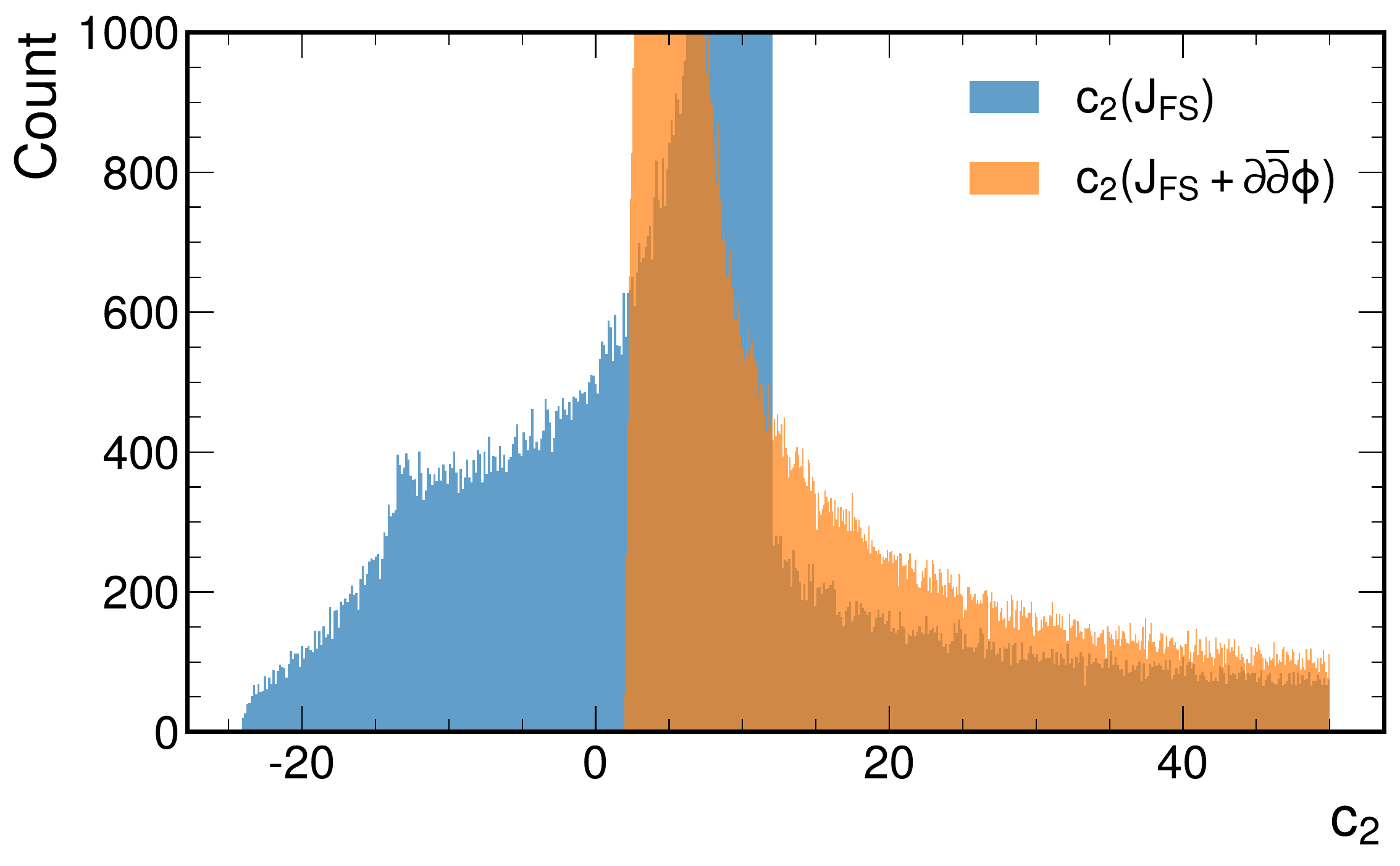}
    \caption{Distribution of the values of $c_2$ using both Fubini--Study and machine learned Calabi--Yau metrics for the Fermat quartic.}\label{fig:quartic_hist_c2} 
\end{figure*}

The integral of the Euler density gives the following values:
\begin{gather}
    \int_{X_0}c_2(J_0^\mathrm{FS}) \approx 24.23 ~,\quad
    \int_{X_0}c_2(J_0^\mathrm{CY}) \approx 24.04 ~,
\end{gather}
where the integral is approximated using $60,000$ points.

\subsubsection[\texorpdfstring{$\lambda = 3/4$}{lambda=3/4}]{$\bm{\lambda = 3/4}$}\label{sec:lambda34}
The number of points in the singular locus $\mathrm{Sing}~{X_{3/4}}$ is $8$. The singular points of $X_{3/4}$ are of form:
\begin{gather}\label{eq:singX34}
	\mathrm{Sing}~{X_{3/4}} = \left\{[\pm 1:\pm 1:\pm 1:\pm 1],\quad\dots\right\} ~.
\end{gather}
We consider a small deviation from $\lambda = 3/4$ by considering varieties $X_{3/4\pm\epsilon}$ for some sufficiently small $\epsilon > 0$. This allows us to study the behavior of the Ricci-flat metric and the curvature thereof as we approach  $X_{3/4}$. In particular, by observing the histogram of the Euler density on Figure~\ref{fig:c2Density075}, we see that the varieties are not uniformly curved. The histograms highlight the contrast between the Fubini--Study and machine learned curvature distributions. A persistent feature is the positivity of the curvature for the machine learned flat approximation. Na\"{\i}vely, this should be expected from the Bogolomov--Yau inequality, because a nearly Ricci-flat metric must have $c_1(J)^2$ close to zero. This feature is persistent for all values of $\lambda$ considered in this work. 
\begin{figure*}[htb]
	\centering
	\includegraphics[width=0.7\textwidth]{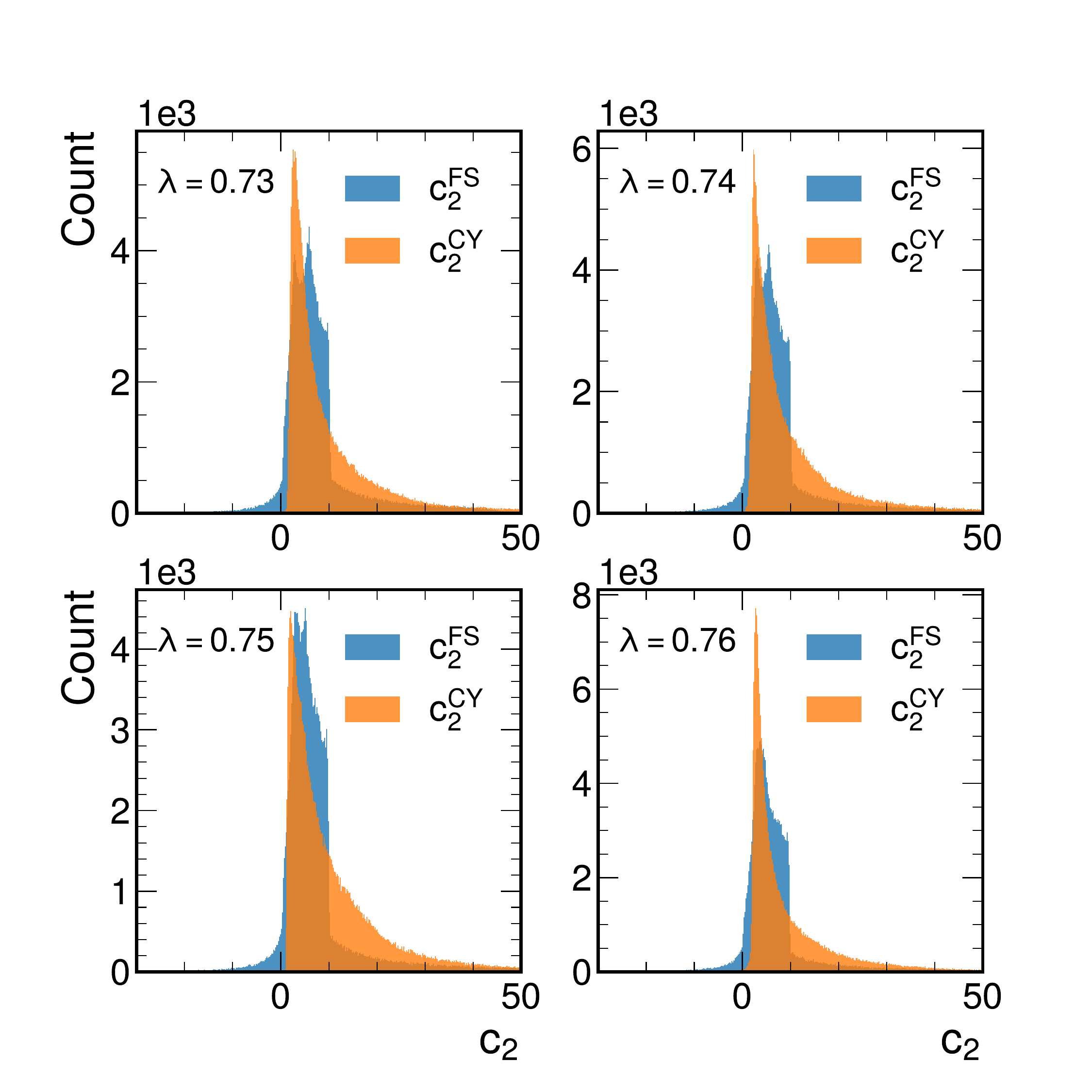}
	\caption{Distribution of the values of the Euler density $e(J_{3/4\pm\epsilon})$ using both $J_{3/4\pm\epsilon}^\textsf{FS}$ and $J_{3/4\pm\epsilon}^\textsf{CY}$.}\label{fig:c2Density075}
\end{figure*}

Let us now explore the curvature tails, \textit{i.e.}, the high curvature regions.
In particular, let $e(J_{3/4\pm\epsilon})$ denote the Euler density corresponding to the curvature form $J_{3/4\pm\epsilon}$ and consider a subset of $X_{3/4\pm\epsilon}$ defined by a one parameter family:
\begin{gather}
	X^\delta_{3/4\pm\epsilon} := \{p\in X_{3/4\pm \epsilon}~\vert~ |e(J_{3/4\pm\epsilon})(p)| \geq \delta\} ~,
\end{gather}
where the cut-off parameter $\delta\in\mathbb{R}$ for the Euler density limits to the high curvature regions in $X_{\lambda}$.
The motivation for defining $\delta$ is to construct a multi-parameter persistent homology, with $\delta$ controlling the cut-off of curvature.
We observe that there exists a sufficiently large value of $\delta > 0$ and sufficiently small value of $\epsilon > 0$ such that that the high curvature regions become disconnected in $X_{3/4\pm \epsilon}^\delta$.
 To study the high Euler density regions, we consider zeroth persistent homology groups $H_0^r(X^\delta_{3/4\pm\epsilon})$.
In this expression, $r$ is the filtration parameter, 
\textit{i.e.}, the radius of the sphere around each point on the Calabi--Yau, $\delta$ is the lower cutoff of the curvature, and $\epsilon$ is the displacement from the singular locus $\lambda=3/4$ in complex structure moduli space.
Hence, the high curvature regions close to the singularity are at $\delta\gg1$ and $0<\epsilon\ll1$. 
For $\lambda=3/4$, we chose $\epsilon=0.02$ and normalized $\delta=0.5$ (normalization is such that the value of $1$ corresponds to the maximum value of Euler density), and computed the filtration in the range $0\leq r\leq 2$.
For larger values of $r$, the persistence diagram of $H_0^r(X^\delta_{3/4\pm\epsilon})$ has stabilized to a single connected component.
For simplicity, we compute the persistent homology in each patch separately using the Euclidean metric.
Specifically, we cover the ambient space $\mathbb{P}^3$ by sets $D_i = \{z_i\neq 0\}$ and consider homology groups of $X_{3/4\pm\epsilon}^{\delta}\cap D_i$. For $\epsilon = 0$, using~\eqref{eq:singX34}, note that there are total of $8$ singular points in each such intersection. In order to ensure that the variety admits well-defined Ricci-flat metric, we consider non-zero $\epsilon > 0$. Thus, the number of generators for different values of $\delta>0$ at different points in the filtration at $\lambda = 0.73$ is shown in Figure~\ref{fig:persMulti073} whereas, the persistence barcode for some sufficiently large $\delta>0$ (in the sense as defined above) is shown in Figure~\ref{fig:persBarcode073}. There, for each patch $D_i$, we observe $8$ cycles with large persistence, each corresponding to some neighborhood of a point $p\in \mathrm{Sing}~X_{3/4}$. Typical points in each cycle are shown in the Table~\ref{table:persistencePoints073} and are thus consistent with points in $\mathrm{Sing}~X_{3/4}$.

\setlength{\abovecaptionskip}{10pt}
\begin{table}[htb]
\centering
{\small
\begin{tabular}{||c l l||} 
 \hline
 Patch & Point & Closest point in $\mathrm{Sing}~X_{3/4}$\\ [0.5ex] 
 \hline\hline
 $D_0$ & $[1.0:~0.87:-0.93:-0.84]$ & $[1:~1:-1:-1]$  \\ 
  & $[ 1.0:-0.86:-0.83:~0.96]$ & $[1:-1:-1:~1]$ \\
  & $[ 1.0:-0.82:-0.87:-0.82]$ & $[1:-1:-1:-1]$ \\
  & $[ 1.0:-0.95:~0.98:~0.95]$ & $[1:-1:~1:~1]$ \\
  & $[ 1.0:-0.95:~0.80:-0.87]$ & $[1:-1:~1:-1]$ \\
  & $[ 1.0:~0.97:-0.95:~0.98]$ & $[1:~1:-1:~1]$ \\
  & $[1.0:~0.84:~0.91:~0.96]$ & $[1:~1:~1:~1]$\\
  & $[ 1.0:~0.92:~0.82:-0.89]$ & $[1:~1:~1:-1]$ \\[1ex] 
 \hline
\end{tabular}}
\caption{Typical points in each of the large persistence cycle of $H_0^r(X_{0.73}^\delta \cap D_0)$.}
\label{table:persistencePoints073}
\end{table}
\setlength{\abovecaptionskip}{-10pt}

\begin{figure*}[htb]
    \centering
    \includegraphics[width=0.8\textwidth]{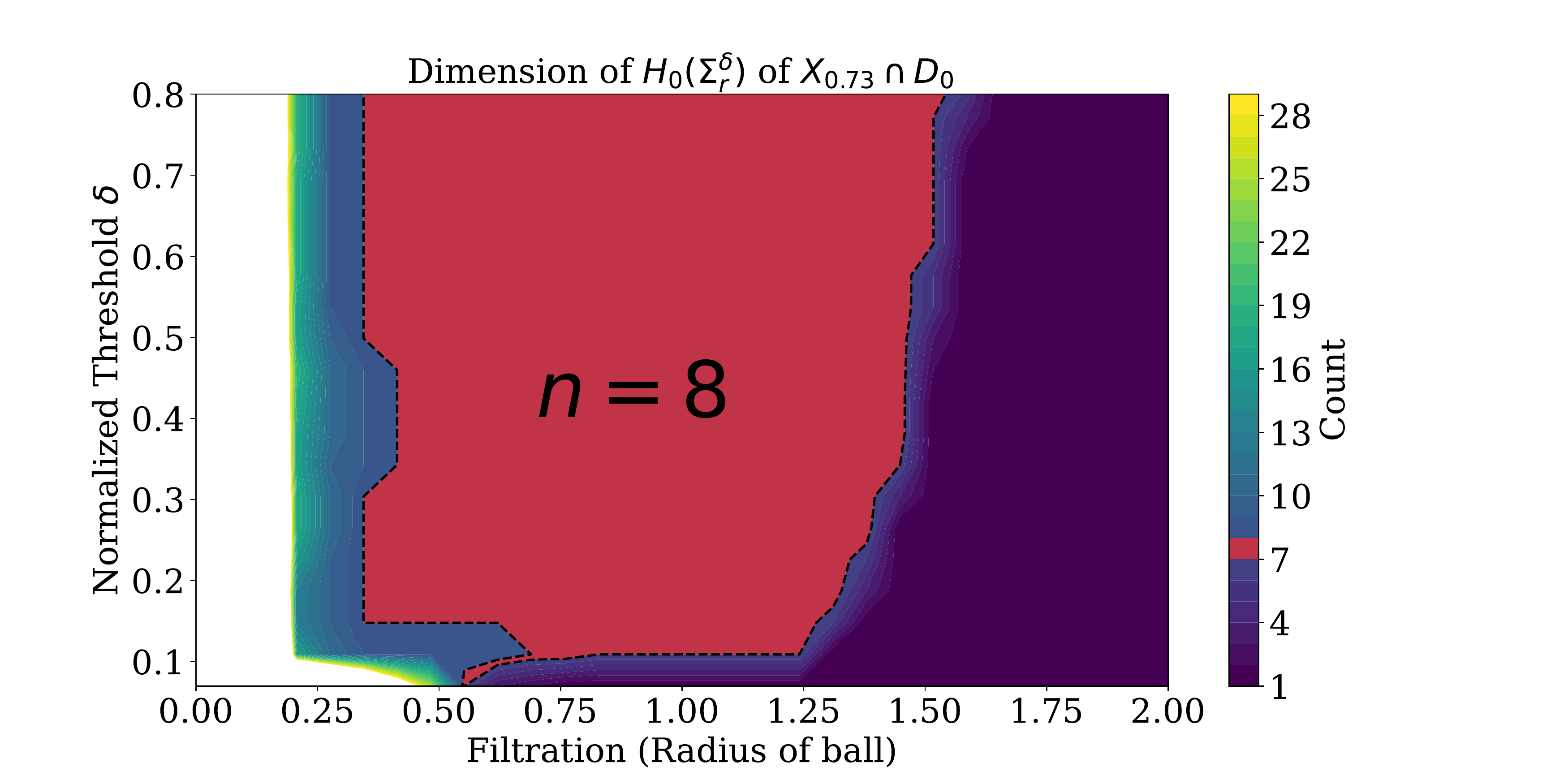}
    \vspace{3mm}
    \caption{Dimension of $H_0(\Sigma_r^\delta)$ for different values of $\delta>0$, where $\Sigma_r^\delta$ is a Vietoris--Rips filtration of $X^\delta_{0.73}\cap D_0$. Notice sharp jump in persistence (difference between death and birth indices in Vietoris--Rips filtration) of $n=8$ generators. The normalized threshold $\delta$ is such that normalized $\delta=1$ corresponds to the largest value of the Euler density.}
    \label{fig:persMulti073}
\end{figure*}

\begin{figure*}[htb]
	\centering
	\includegraphics[width=0.8\textwidth]{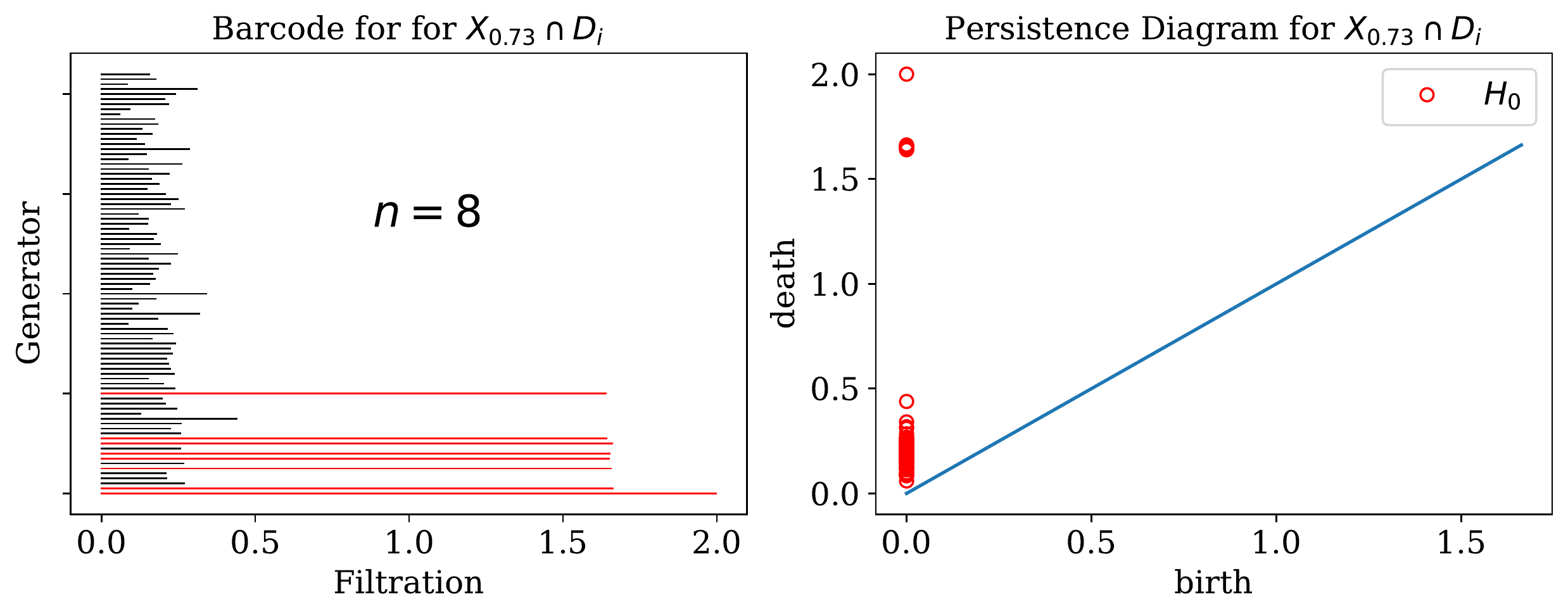}
 \vspace{3mm}
	\caption{Persistence barcode and diagram of $H_0^r(X_{0.73}^\delta \cap D_0)$ where $D_0=\{z_0\neq 0\}$. The number $n$ indicates the number of generators $\gamma$ with persistence $\mathrm{pers}(\gamma) > 0.7$ (colored as red). }\label{fig:persBarcode073}
\end{figure*}

\subsubsection[\texorpdfstring{$\lambda = 1$}{lambda=1}]{$\bm{\lambda = 1}$}
The number of points in the singular locus $\mathrm{Sing}~X_1$ is $16$. The singular points of $X_1$ are of form:
\begin{gather}\label{eq:sing1}
	\mathrm{Sing}~X_1 = \left\{[\pm 1: \pm1 : \pm 1 : 0],~[\pm 1: \pm 1: 0: \pm 1],~\dots\right\}	~.
\end{gather}
Similarly as in Section~\ref{sec:lambda34}, we contemplate deviations from $\lambda = 1$ by considering manifolds $X_{1\pm\epsilon}$ for some sufficiently small $\epsilon >0$. Histograms of the Euler density $e(J_{1\pm\epsilon})$ for the Fubini--Study metric as well as the machine learned approximation are shown in 
Figure~\ref{fig:c2Density1}. Similarly, we see that $X_{1\pm\epsilon}$ is not uniformly curved, thus, we consider $X_{1\pm\epsilon}^\delta$ defined similarly as a one parameter family:
\begin{gather}
	X_{1\pm\epsilon}^\delta := \{p\in X_{1\pm \epsilon}~\vert~ |e(J_{1\pm\epsilon})(p)| \geq \delta\} ~.
\end{gather}
The threshold parameter $\delta\in\mathbb{R}$ is not necessarily the same as in Section~\ref{sec:lambda34}. For studying the high Euler density regions, we analogously consider persistent homology groups $H_0^r(X_{1\pm\epsilon}^\delta \cap D_k)$ for each patch $D_k$. For $\epsilon = 0$, each patch contains: $3\times 2^{2}=12$ singular points. The persistence barcode at $\lambda = 0.98$ is shown in Figure~\ref{fig:persBarcode1}.

\begin{figure*}[htb]
	\centering
	\includegraphics[width=0.7\textwidth]{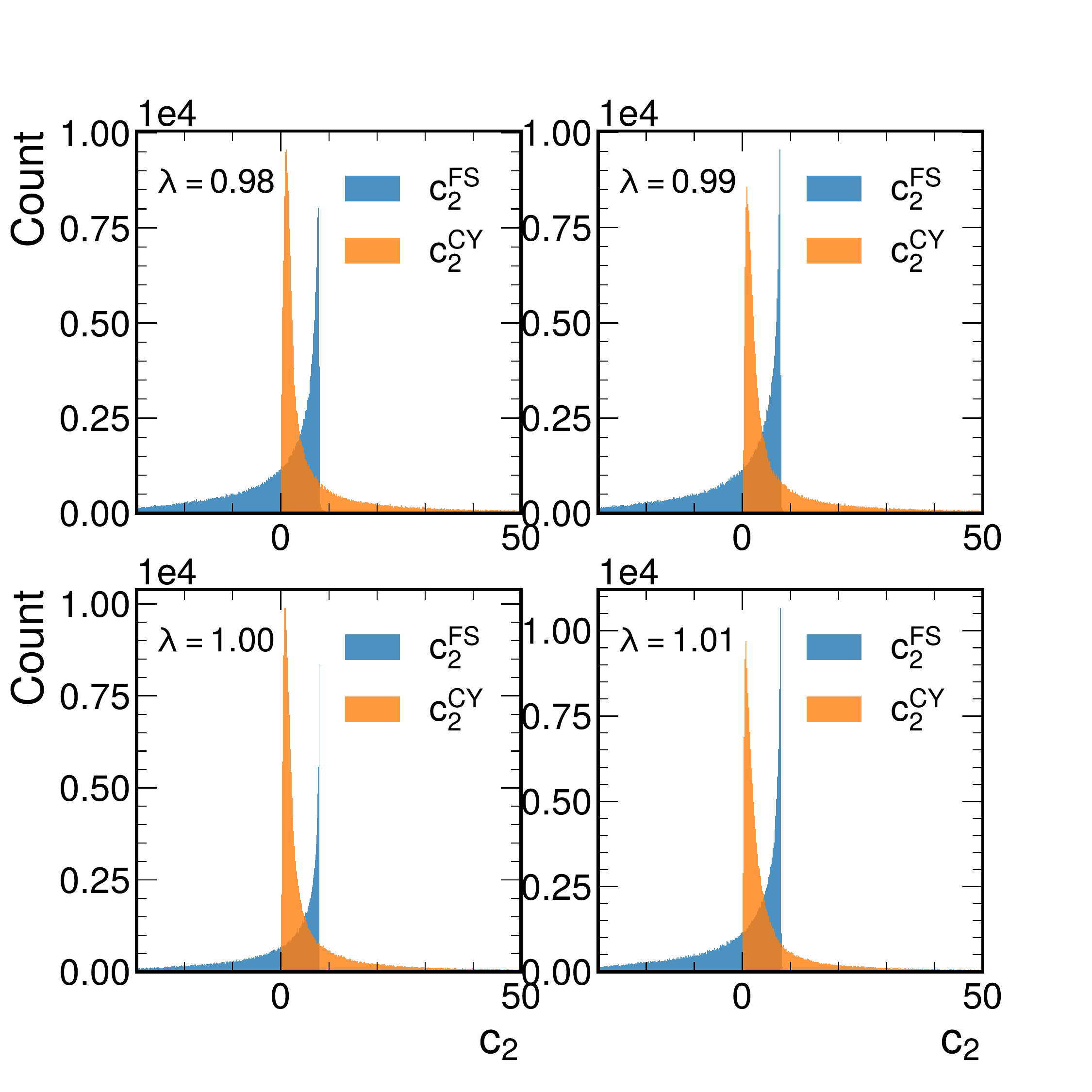}
	\caption{Distribution of the values of the Euler density $e(J_{1\pm\epsilon})$ using both $J_{1\pm\epsilon}^\textsf{FS}$ and $J_{1\pm\epsilon}^\textsf{CY}$.}\label{fig:c2Density1}
\end{figure*}

\begin{figure*}[htb]
	\centering
	\includegraphics[width=0.8\textwidth]{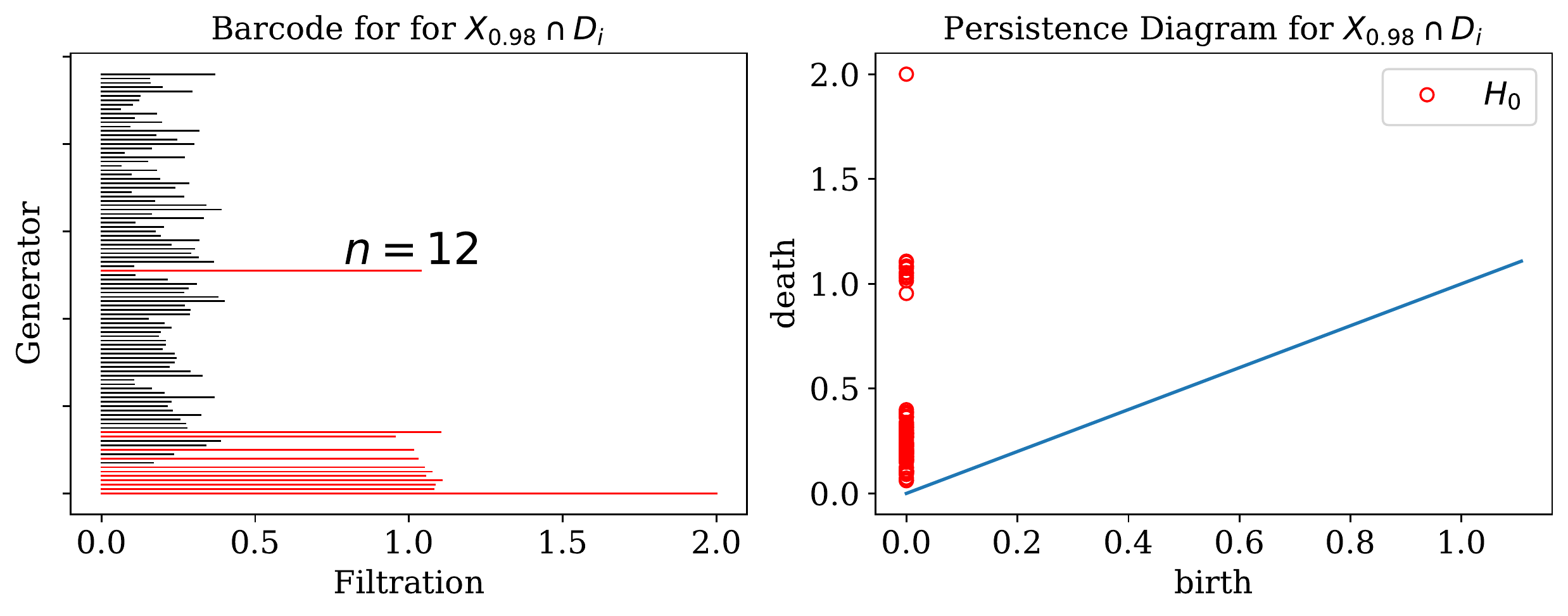}
  \vspace{3mm}
	\caption{Persistence barcode of $H_0^r(X_{0.98}^\delta \cap D_0)$ where $D_0=\{z_0\neq 0\}$. The number $n$ indicates the number of generators $\gamma$ with persistence $\mathrm{pers}(\gamma) > 0.7$ (colored as red).}\label{fig:persBarcode1}.
\end{figure*}

 For each patch $D_i$, we observe $12$ cycles with large persistence, each corresponding to some neighborhood of a point $p\in\mathrm{Sing}~X_1$. Typical points in each cycle are shown in 
 Table~\ref{table:persistencePoints1} and are thus also consistent with the points~\eqref{eq:sing1} in $\mathrm{Sing}~X_1$.
 
Furthermore, note that the Euler density $e(J_{1\pm\epsilon}^\textsf{CY})$ also satisfies the same positivity property as $X_{3/4}$, that is: $e(J_{1\pm\epsilon}^\textsf{CY})\gtrsim 0$, even in the case when $\epsilon = 0$.

\setlength{\abovecaptionskip}{10pt}
\begin{table}[htb]
\centering
{\small
\begin{tabular}{||c l l||} 
 \hline
 Patch & Point & Closest point in $\mathrm{Sing}~X_{1}$\\ [0.5ex] 
 \hline\hline
 $D_0$ & $[ 1.:-0.96:~0.97:~0.04]$ & $[1:-1:~1:~0]$  \\
 & $[1.:~0.98:~0.95:~0.04]$ & $[1:~1:~1:~0]$\\
 & $[ 1.:-0.05:-0.99:~0.95]$ & $[1:~0:-1:~1]$ \\
 & $[1.:~0.85:~0.22:~1.]$ & $[1:~1:~0:~1]$ \\
 & $[ 1.:-0.92:~0.02:-0.97]$& $[1:-1:~0:-1]$ \\
 & $[ 1.:-0.19:~0.86:-0.99]$& $[1:~0:~1:-1]$ \\
 & $[ 1.:~0.19:-0.99:-0.86]$& $[1:~0:-1:-1]$ \\
 & $[ 1.:~0.90:-0.03:-0.91]$& $[1:~1:~0:-1]$ \\
 & $[ 1.:-0.99:-0.96:-0.08]$& $[1:-1:-1:~0]$ \\
 & $[1.:~0.15:~0.88:~0.98]$& $[1:~0:~1:~1]$ \\
 & $[ 1.:0.96:-0.85:-0.16]$& $[1:~1:-1:~0]$ \\
 &$[ 1.:-0.94:-0.10:~0.98]$& $[1:-1:~0:~1]$\\[1ex] 
 \hline
\end{tabular}}
\caption{Typical points in each of the large persistence cycle of $H_0^r(X_{0.98}^\delta \cap D_0)$.}
\label{table:persistencePoints1}
\end{table}
\setlength{\abovecaptionskip}{-10pt}

\newpage
\subsubsection[\texorpdfstring{$\lambda = 3/2$}{lambda=3/2}]{$\bm{\lambda = 3/2}$}
The number of points in the singular locus $\mathrm{Sing}~X_{3/2}$ is $12$. The singular points of $X_{3/2}$ are of form:
\begin{gather}\label{eq:singX32}
	\mathrm{Sing}~X_{3/2} = \left\{[\pm 1,\pm 1,0,0],~\dots \right\} ~.
\end{gather}
Similarly, we consider deviation from $\lambda = 3/2$ by considering $X_{3/2\pm\epsilon}$ for some sufficiently small $\epsilon > 0$. The histogram of the Euler density $e(J_{3/2\pm\epsilon})$ is shown in Figure~\ref{fig:c2Density15}.
\begin{figure*}[htb]
	\centering
	\includegraphics[width=0.7\textwidth]{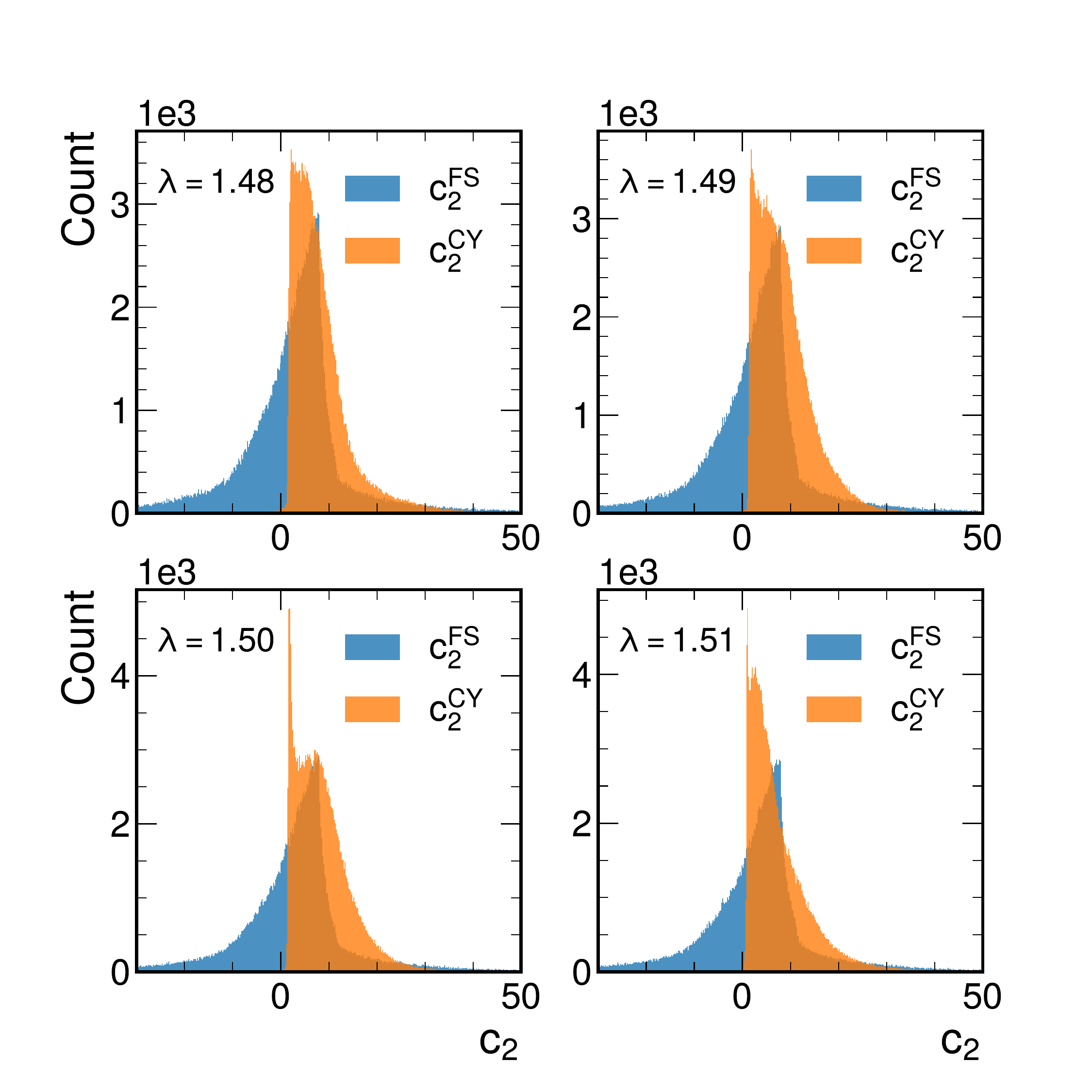}
	\caption{Distribution of the values of the Euler density $e(J_{3/2\pm\epsilon})$ using both $J_{3/2\pm\epsilon}^\textsf{FS}$ and $J_{3/2\pm\epsilon}^\textsf{CY}$.}\label{fig:c2Density15}
\end{figure*}

For each $k$, the patch $X_{3/2}\cap D_k$ contains total of $3\times 2 = 6$ singular points.
The persistence barcode for $X_{1.48}^\delta \cap D_0$ is shown in Figure~\ref{fig:c2Density15}.

\begin{figure*}[htb]
	\centering
	\includegraphics[width=0.8\textwidth]{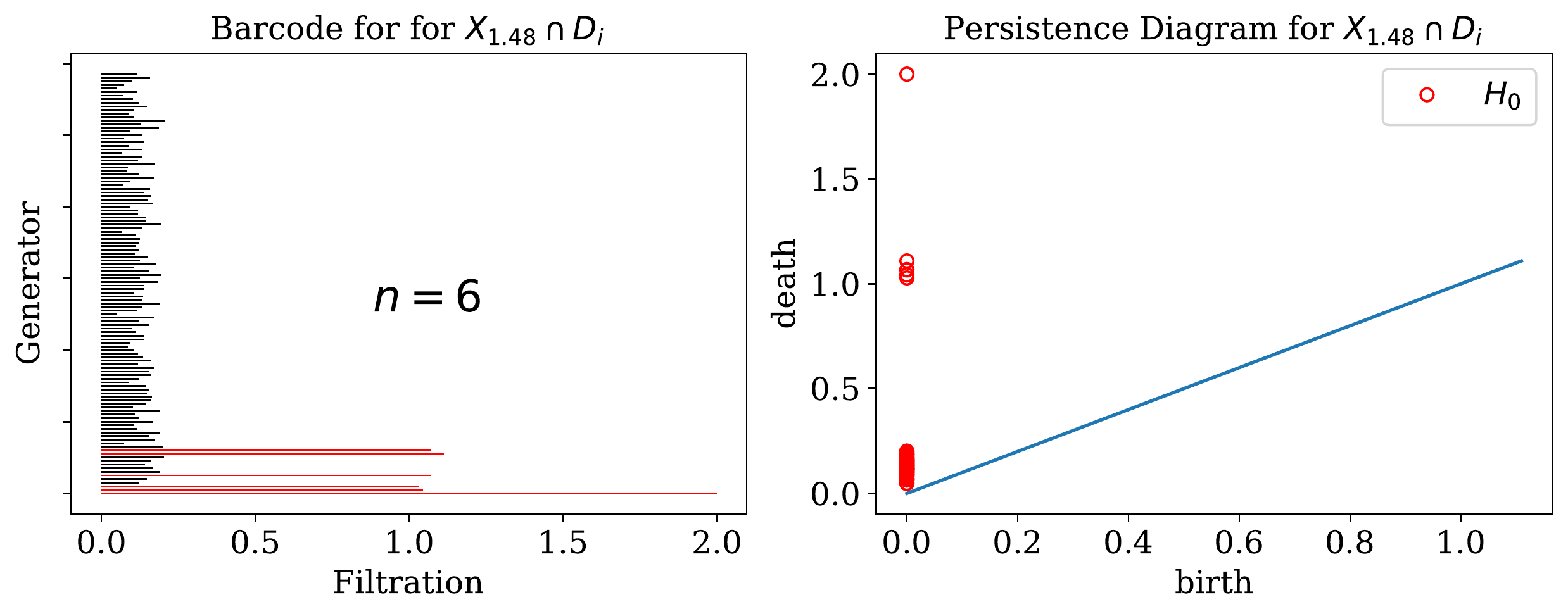}
  \vspace{3mm}
	\caption{Persistence barcode of $H_0^r(X_{1.48}^\delta \cap D_0)$ where $D_0=\{z_0\neq 0\}$. The number $n$ indicates the number of generators $\gamma$ with persistence $\mathrm{pers}(\gamma) > 0.7$ (colored as red).}\label{fig:persBarcode15}.
\end{figure*}
For each patch $D_i$ we observe $6$ cycles with large persistence, each corresponding some neighborhood of a point $p\in \mathrm{Sing}~X_{3/2}$. Typical points in each cycle are shown in the Table~\ref{table:persistencePoints15} and are consistent with the points~\eqref{eq:singX32}.

\setlength{\abovecaptionskip}{10pt}
\begin{table}[htb]
\centering
{\small
\begin{tabular}{||c l l||} 
 \hline
 Patch & Point & Closest point in $\mathrm{Sing}~X_{3/2}$\\ [0.5ex] 
 \hline\hline
 $D_0$ & $[ 1.:~0.04:-0.96:~0.15]$ & $[1:~0:-1:~0]$  \\
  & $[ 1.:-0.97:-0.:~0.06]$ & $[1:-1:~0:~0]$ \\
  & $[1.:~0.93:~0.11:~0.03]$ & $[1:~1:~0:~0]$ \\
  & $[ 1.:-0.06:-0.01:~0.96]$& $[1:~0:~0:~1]$ \\
  & $[ 1.:-0.02:~0.90:-0.14 ]$& $[1:~0:~1:~0]$ \\
 & $[ 1.:0.17:0.04:-1.]$& $[1:~0:~0:-1]$ \\[1ex] 
 \hline
\end{tabular}}
\caption{Typical points in each of the large persistence cycle of $H_0^r(X_{1.48}^\delta \cap D_0)$.}
\label{table:persistencePoints15}
\end{table}
\setlength{\abovecaptionskip}{-10pt}

\subsubsection[\texorpdfstring{$\lambda = 3$}{lambda=3}]{$\bm{\lambda = 3}$}
The number of points in the singular locus $\mathrm{Sing}~X_{3}$ is $4$. The singular points of $X_3$ are of form:
\begin{gather}\label{eq:singX3}
	\mathrm{Sing}~X_{3} = \left\{[1:0:0:0],~[0:1:0:0],~\dots\right\} ~.
\end{gather}
Similarly, we consider deviation from $\lambda = 3$ by considering $X_{3}$ for sufficiently small $\epsilon>0$. The histogram of the Euler density $e(J_{3\pm\epsilon})$ is shown in Figure~\ref{fig:c2Density3}.
\begin{figure*}[htb]
	\centering
	\includegraphics[width=0.7\textwidth]{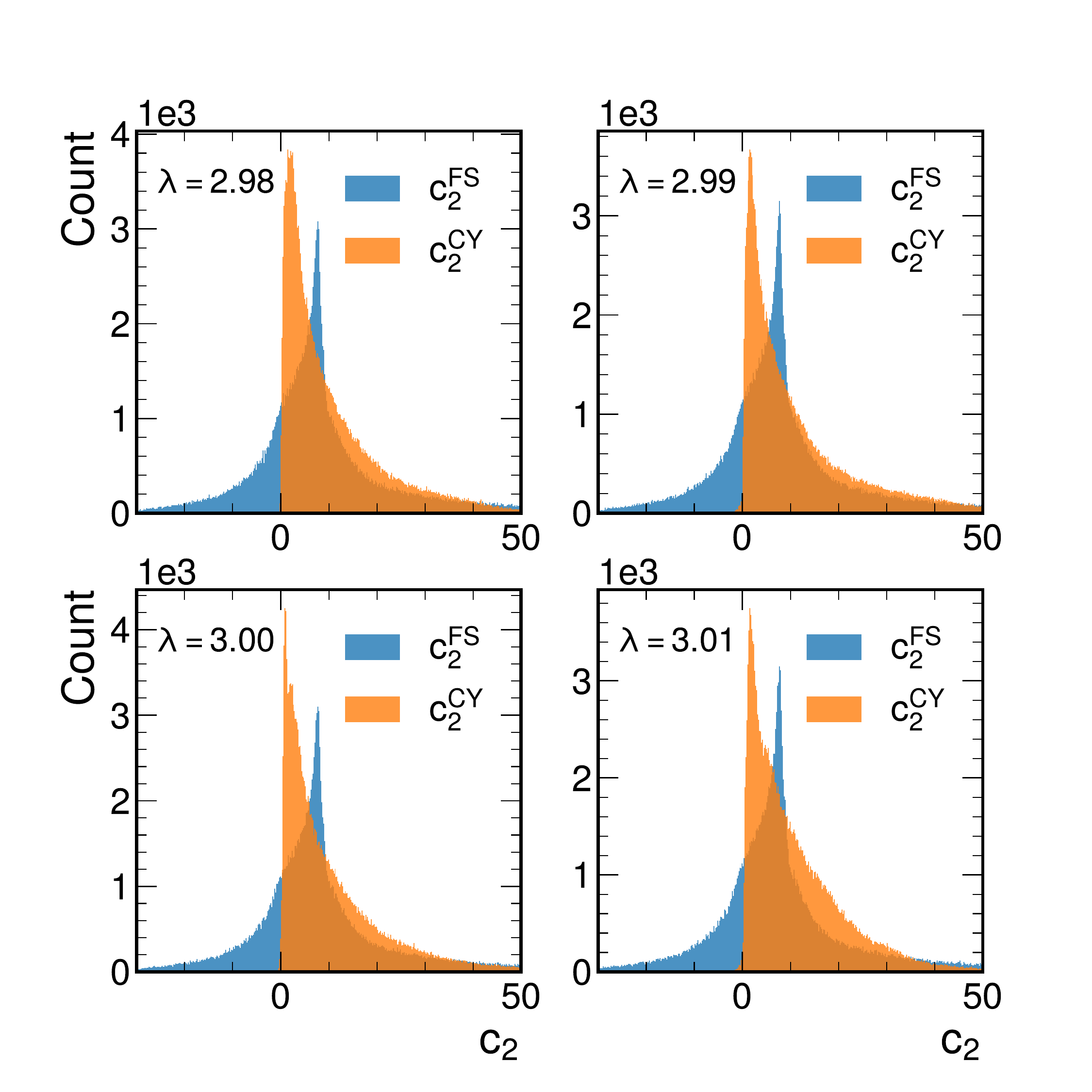}
	\caption{Distribution of the values of the Euler density $e(J_{3\pm\epsilon})$ using both $J_{3\pm\epsilon}^\textsf{FS}$ and $J_{3\pm\epsilon}^\textsf{CY}$.}\label{fig:c2Density3}
\end{figure*}
For each $k$, the patch $X_{3}\cap D_k$ contains only a single singular point, thus, instead we consider the persistence barcode of $X_{2.98}^\delta$, shown on the Figure~\ref{fig:persBarcode3}.
\begin{figure*}[htb]
	\centering
	\includegraphics[width=0.8\textwidth]{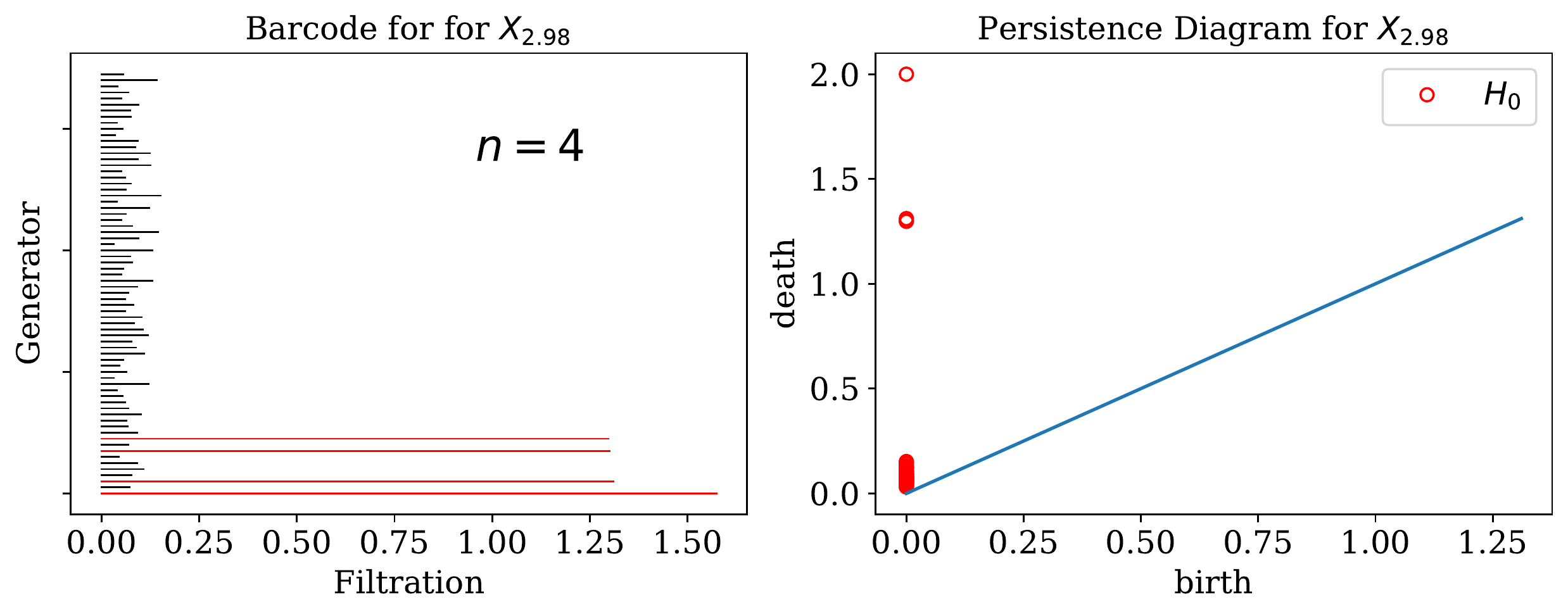}
  \vspace{3mm}
	\caption{Persistence barcode of $H_0^r(X_{2.98}^\delta)$. The number $n$ indicates the number of generators $\gamma$ with persistence $\mathrm{pers}(\gamma) > 0.7$ (colored as red).}\label{fig:persBarcode3}
\end{figure*}

\setlength{\abovecaptionskip}{10pt}
\begin{table}[htb]
\centering
{\small
\begin{tabular}{||l l||} 
 \hline
 Point & Closest point in $\mathrm{Sing}~X_{3}$\\ [0.5ex] 
 \hline\hline
 $[0.04:-0.03:~1.~~~:-0.07]$ & $[0:~0:~1:~0]$  \\
  $[0.03:-0.10:~0.08:~1.~~~~]$ & $[0:~0:~0:~1]$\\
  $[1.~~~:~~0.06:~0.01:~0.0~~~]$ & $[1:~0:~0:~0]$\\
  $[0.00:~1.~~~~:~0.04:-0.05]$ & $[0:~1:~0:~0]$ \\[1ex] 
 \hline
\end{tabular}}
\caption{Typical points in each of the large ($\mathrm{pers}(\gamma) > 0.7$) persistence cycle of $H_0^r(X_{2.98}^\delta)$.}
\label{table:persistencePoints3}
\end{table}
\setlength{\abovecaptionskip}{-10pt}

The typical points in each large persistence generator of $H_0^r(X_{2.98})$ are shown in Table~\ref{table:persistencePoints3}, which is consistent with the singular locus~\eqref{eq:singX3}.
\begin{figure*}[htb]
    \centering
    \includegraphics[width=0.8\textwidth]{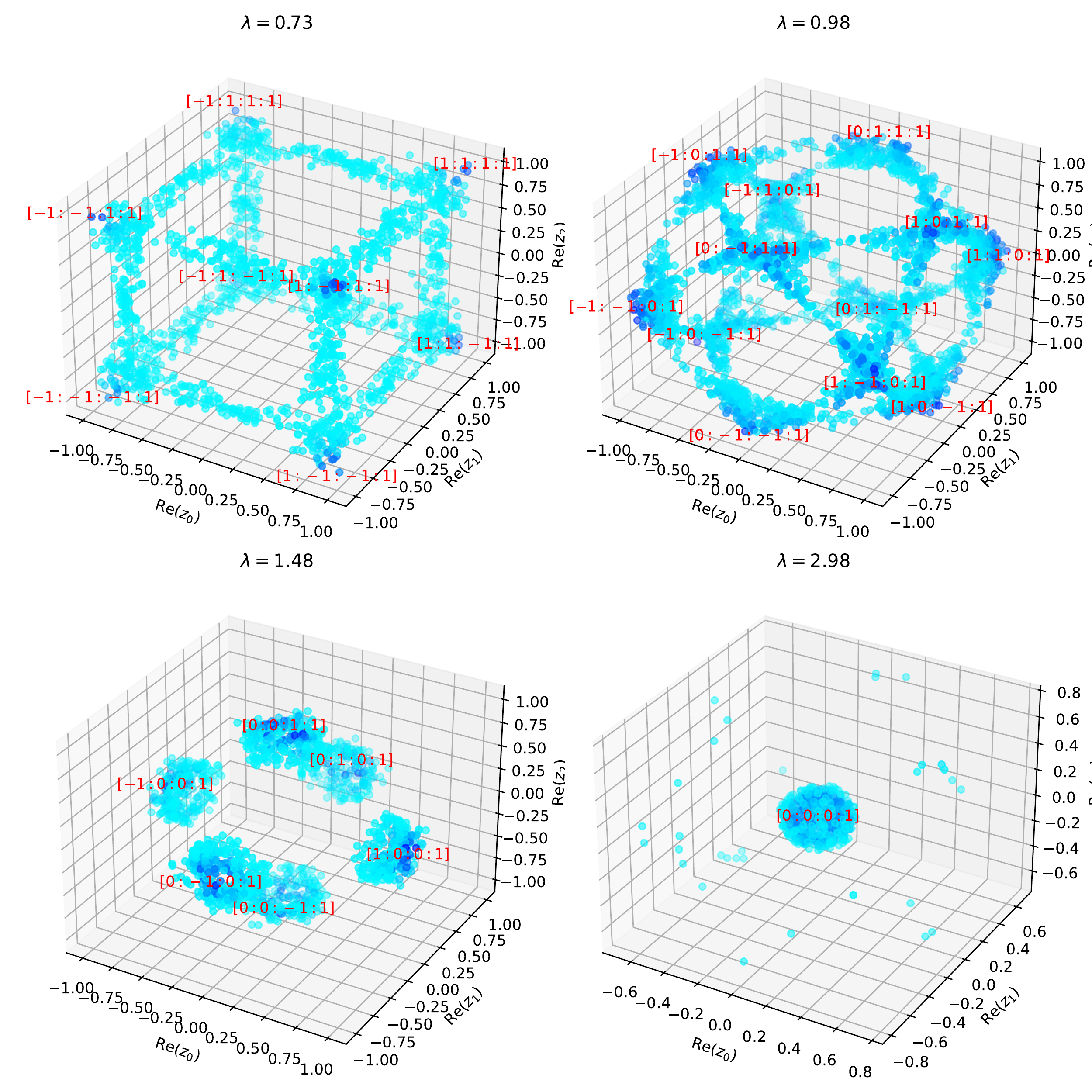}
    \vspace*{3mm}
    \caption{Projection of $X_{\lambda^\sharp+\epsilon,t}\cap D_3$ to $(\mathrm{Re}(z_0),~\mathrm{Re}(z_1),~\mathrm{Re}(z_2))$ for some large $t>0$. The shading is induced by the Euler density: $e(J_{\lambda^\sharp+\epsilon})$. The red labels indicate the singular points of $X_{\lambda^\sharp}\cap D_{3}$.}\label{fig:clusters}.
\end{figure*}
\subsubsection{Asymptotics near large Euler density regions}\label{sec:convergence}
In the previous sections we have considered the large curvature regions in the vicinity of the singularities for the various $X_{\l^\sharp}$ in the Cefal\'{u} pencil. Taking points in  $X_{\l^\sharp\pm\epsilon}$ with a given curvature density $e(J)$ larger than a given parameter $t$, we managed to identify a clustering of points consistent with the singularity distribution for $X_{\l^\sharp}$. In addition to the Euler density $c_2(J_\lambda)$, we may also consider $c_1(J_\lambda)^2$.
 The top forms $c_2(J_\lambda)$ and $c_1(J_\lambda)^2$ induce filtrations of $X_\lambda$ defined by:
\begin{equation}
\begin{aligned}
	X_{\lambda,t_1}^1 &:= \left\{p\in X_{\lambda}~\vert~\left|c_1(J_\lambda)^2(p)\right| < t_1\right\} ~,\\
	X_{\lambda,t_2}^2 &:= \left\{p\in X_{\lambda}~\vert~\left|c_2(J_\lambda)(p)\right| < t_2\right\} ~.
\end{aligned}
\end{equation}
To better visualize the occurrence of large curvature values, 
we pick a patch $D_3 = \{z_3 \neq 0\}$ and for all $\lambda^\sharp$ we construct the scatter plots of $X_{\lambda, t}\cap D_3$, by projecting each point $p= [z_0:z_1:z_2:1]\in X_{\lambda, t}$ to $(\mathrm{Re}(z_0),\mathrm{Re}(z_1), \mathrm{Re}(z_2))$. This allows us to visualize the clustering behavior observed in the persistence diagrams shown on Figures~\ref{fig:persBarcode073},~\ref{fig:persBarcode1},~\ref{fig:persBarcode15} and~\ref{fig:persBarcode3}. We pick values of $\lambda$ to be $\lambda^\sharp+\epsilon$ for some small $|\epsilon|>0$ and pick sufficiently large threshold $t>0$ to highlight the clustering behavior. The results of this are shown on Figure~\ref{fig:clusters}. We also include similar plots of the curvature on for the real parts of $X_\lambda$ in the Figures~\ref{fig:realXcurvatures} and~\ref{fig:realXcurvatures2}.

\begin{figure*}[htb]
    \centering
    \begin{subfigure}[b]{0.5\textwidth}
        \centering
        \includegraphics[width=0.8\textwidth]{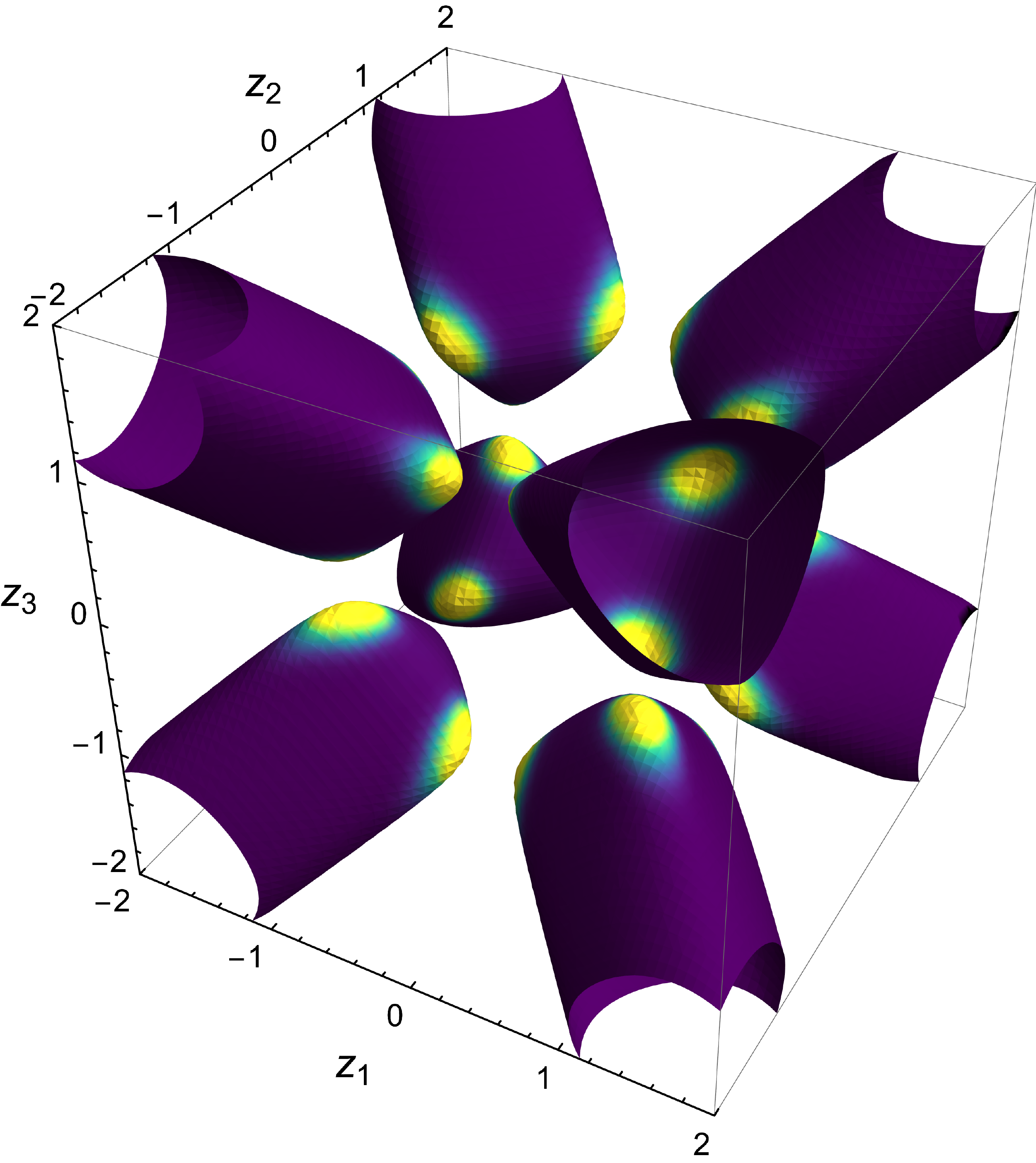}
        \caption{$\lambda = 0.98$, $J_\mathrm{CY}$}
    \end{subfigure}%
    \begin{subfigure}[b]{0.5\textwidth}
        \centering
        \includegraphics[width=0.8\textwidth]{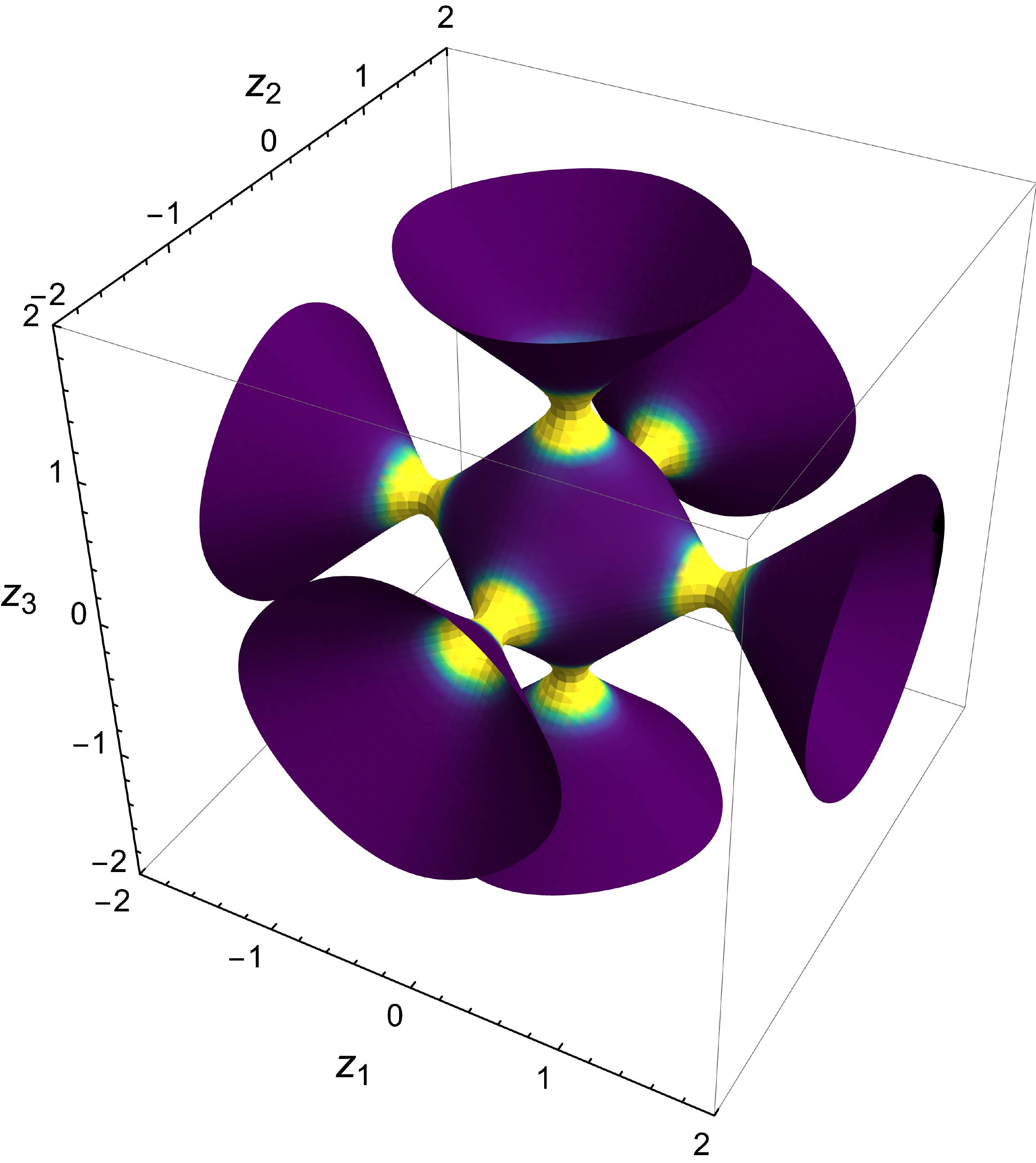}
        \caption{$\lambda = 1.48$, $J_\mathrm{CY}$}
    \end{subfigure}
    \\
    \begin{subfigure}[b]{0.5\textwidth}
        \centering
        \includegraphics[width=0.8\textwidth]{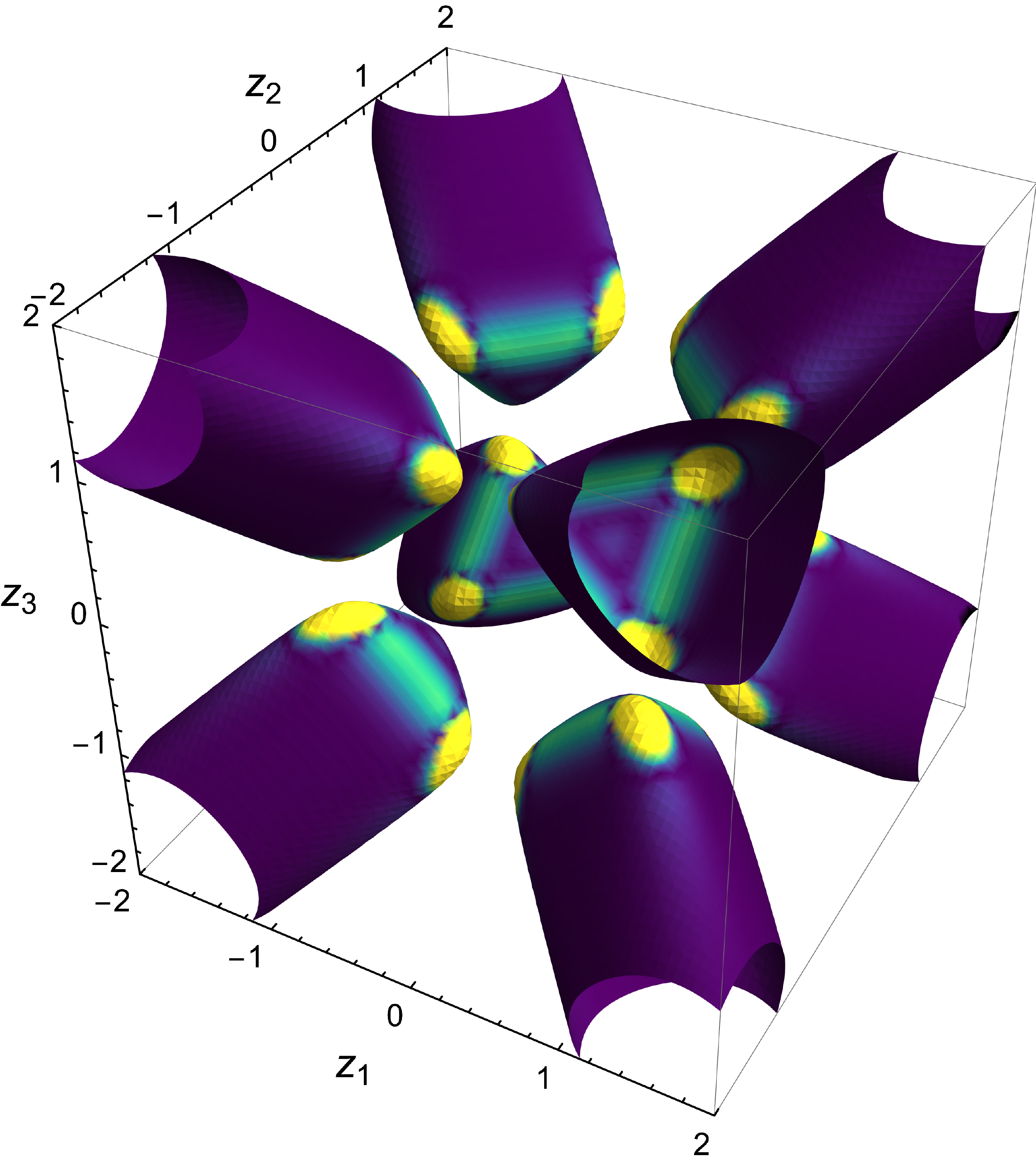}
        \caption{$\lambda = 0.98$, $J_\mathrm{FS}$}
    \end{subfigure}%
    \begin{subfigure}[b]{0.5\textwidth}
        \centering
        \includegraphics[width=0.8\textwidth]{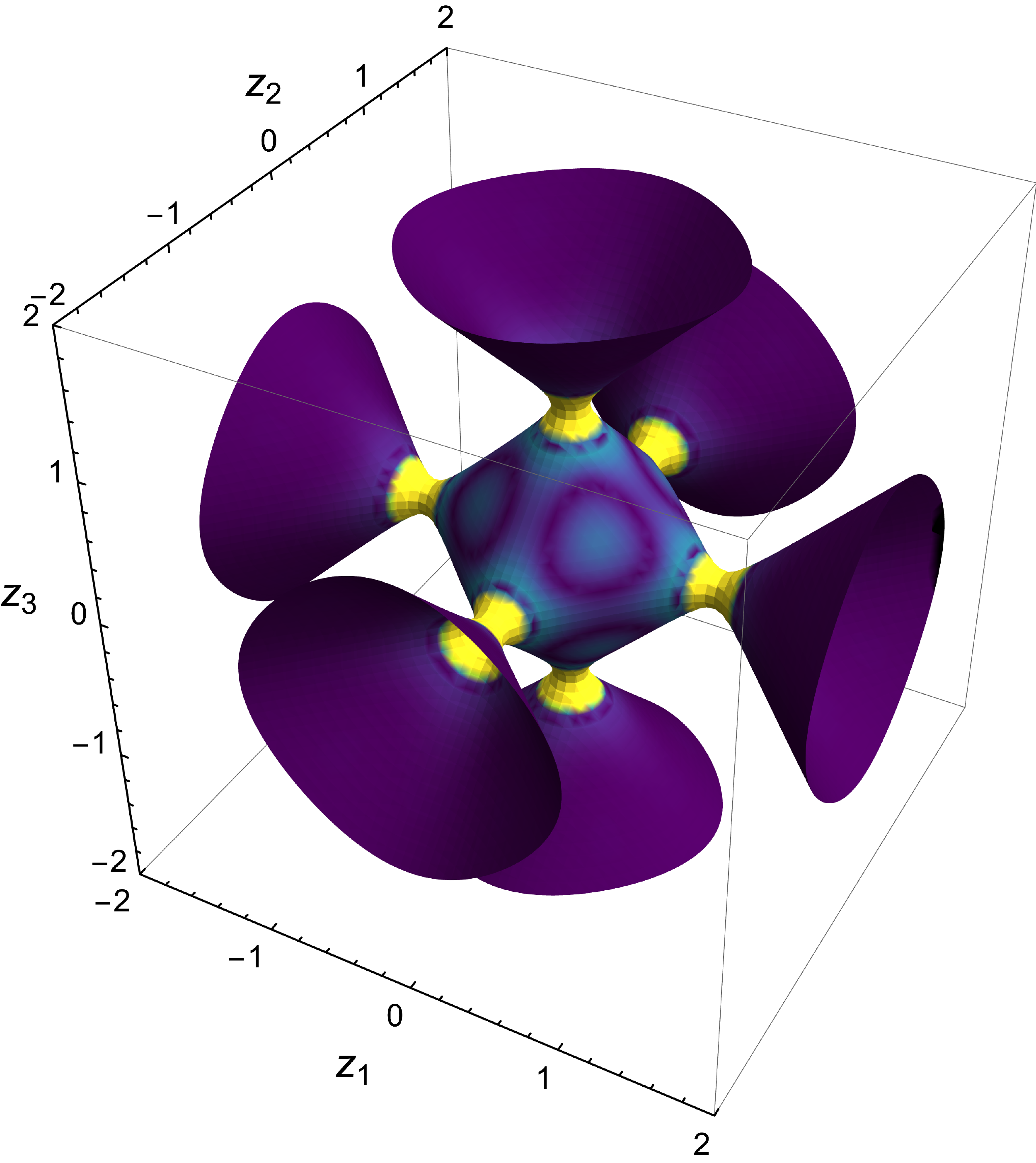}
        \caption{$\lambda = 1.48$, $J_\mathrm{FS}$}
    \end{subfigure}
    \vspace{10mm}
    \caption{Plot of $X_{\lambda^\sharp}\cap D_3 \cap \mathbb{R}^3$ with shading induced by the Euler density. Brighter colors denote larger value of $e(J)$. In order to correctly take into account the identification $\lambda p \sim p$ for $\l\in\IC^\times$, we have used spectral networks to compute $\phi$ (See~\ref{sec:spectralNNs}).}\label{fig:realXcurvatures}

\end{figure*}

\begin{figure*}[htb]
    \centering
    \begin{subfigure}[b]{0.5\textwidth}
        \centering
        \includegraphics[width=0.8\textwidth]{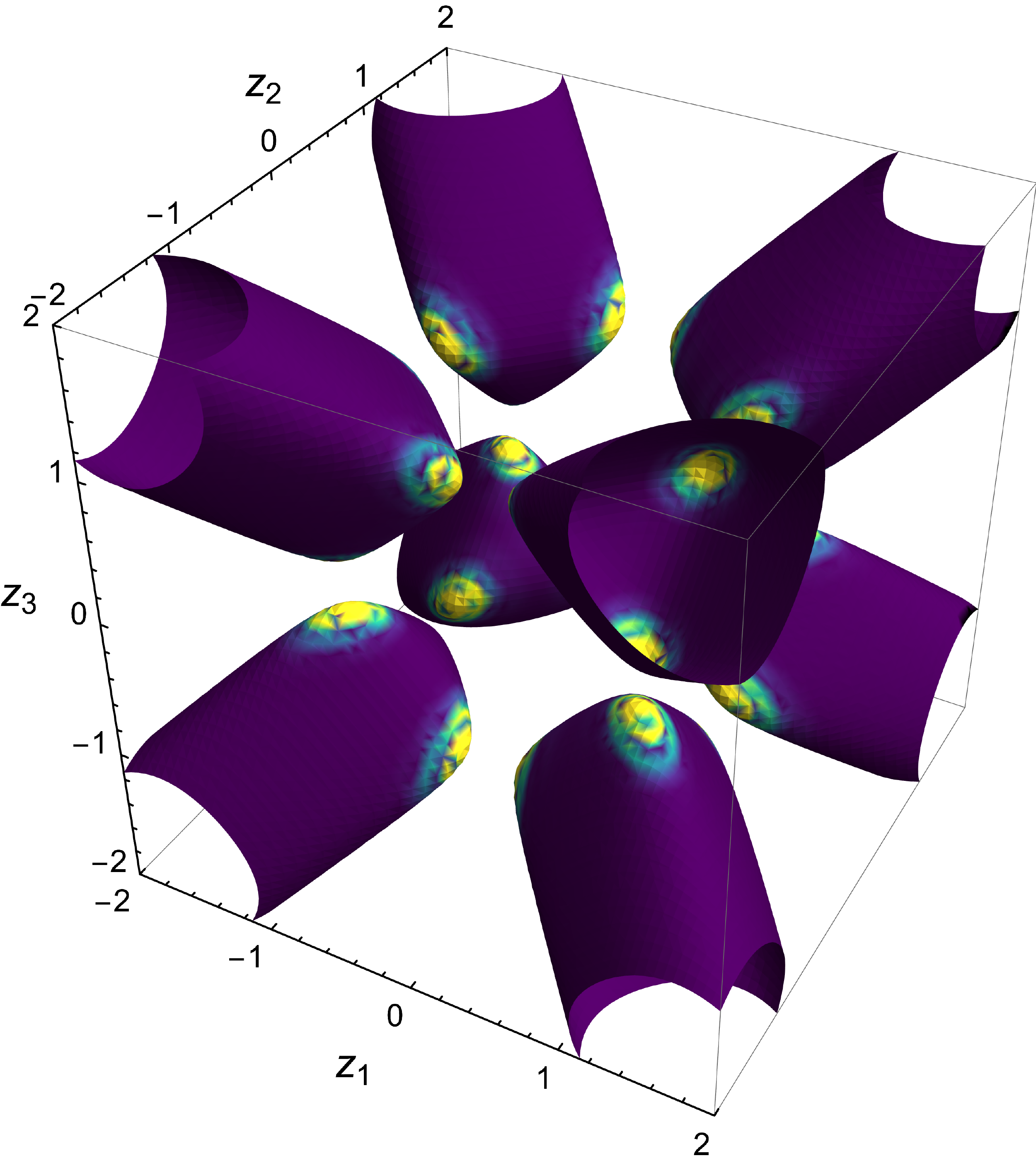}
        \caption{$\lambda = 0.98$, $J_\mathrm{CY}$}
    \end{subfigure}%
    \begin{subfigure}[b]{0.5\textwidth}
        \centering
        \includegraphics[width=0.8\textwidth]{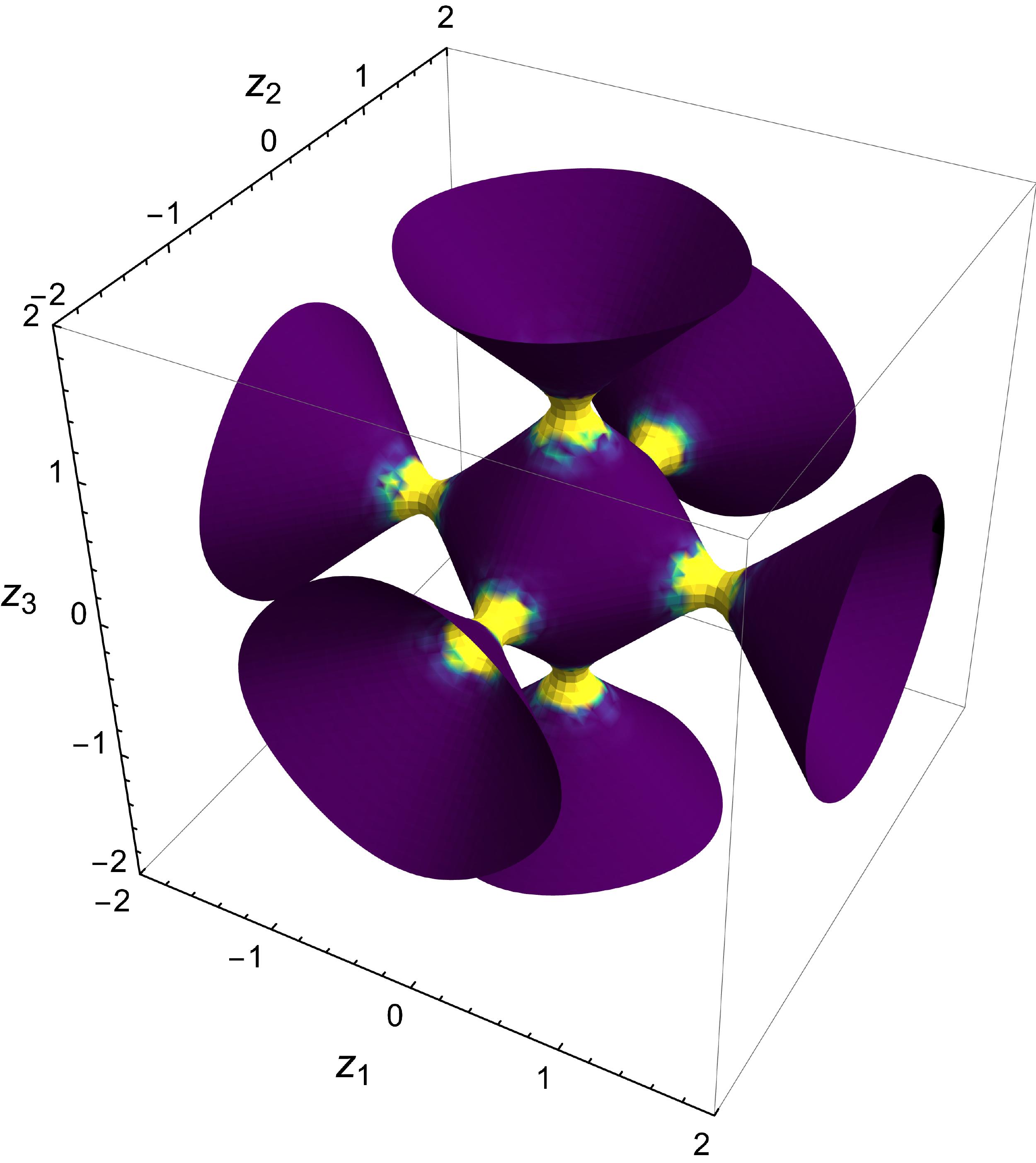}
        \caption{$\lambda = 1.48$, $J_\mathrm{CY}$}
    \end{subfigure}
    \\
    \begin{subfigure}[b]{0.5\textwidth}
        \centering
        \includegraphics[width=0.8\textwidth]{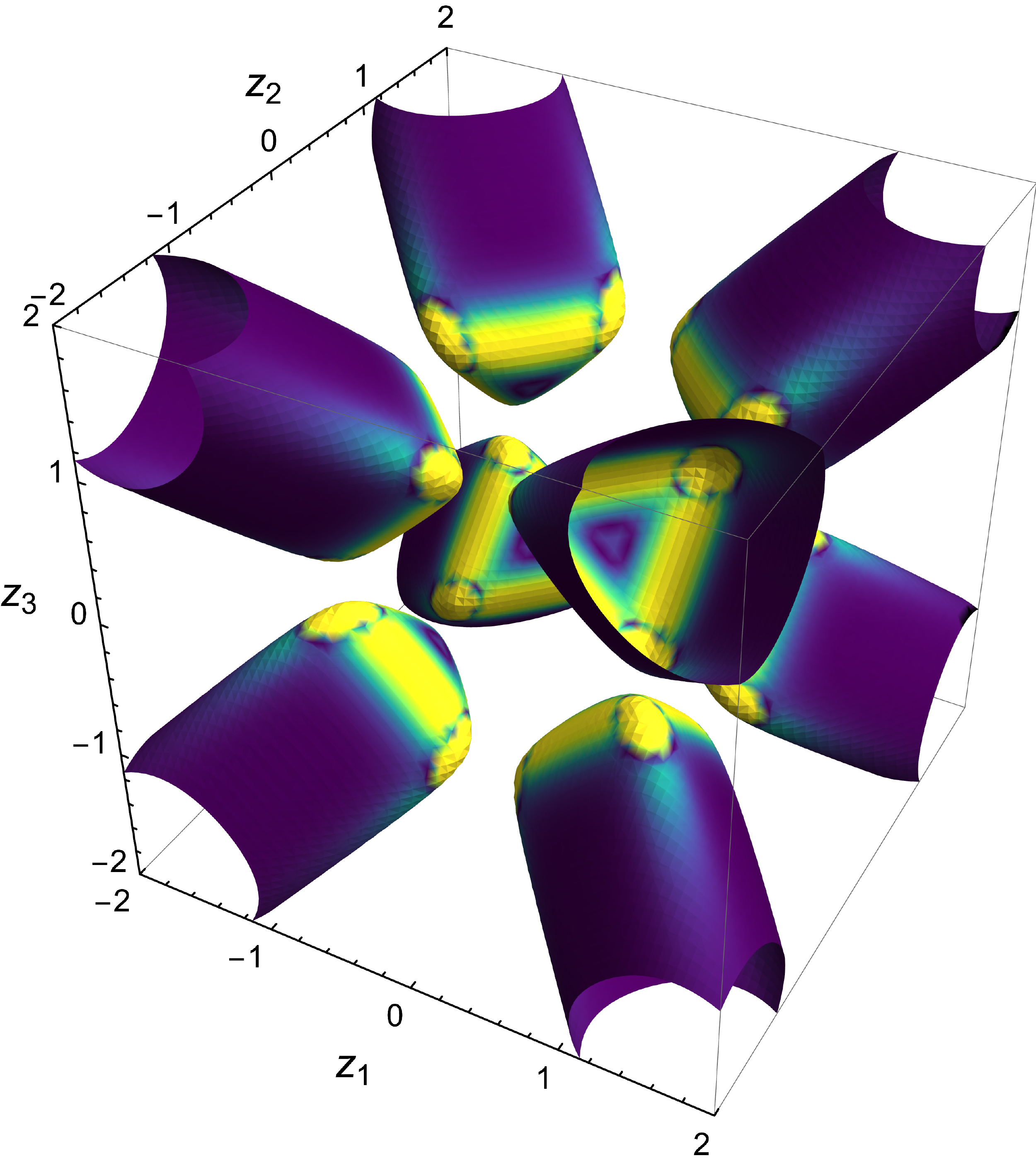}
        \caption{$\lambda = 0.98$, $J_\mathrm{FS}$}
    \end{subfigure}%
    \begin{subfigure}[b]{0.5\textwidth}
        \centering
        \includegraphics[width=0.8\textwidth]{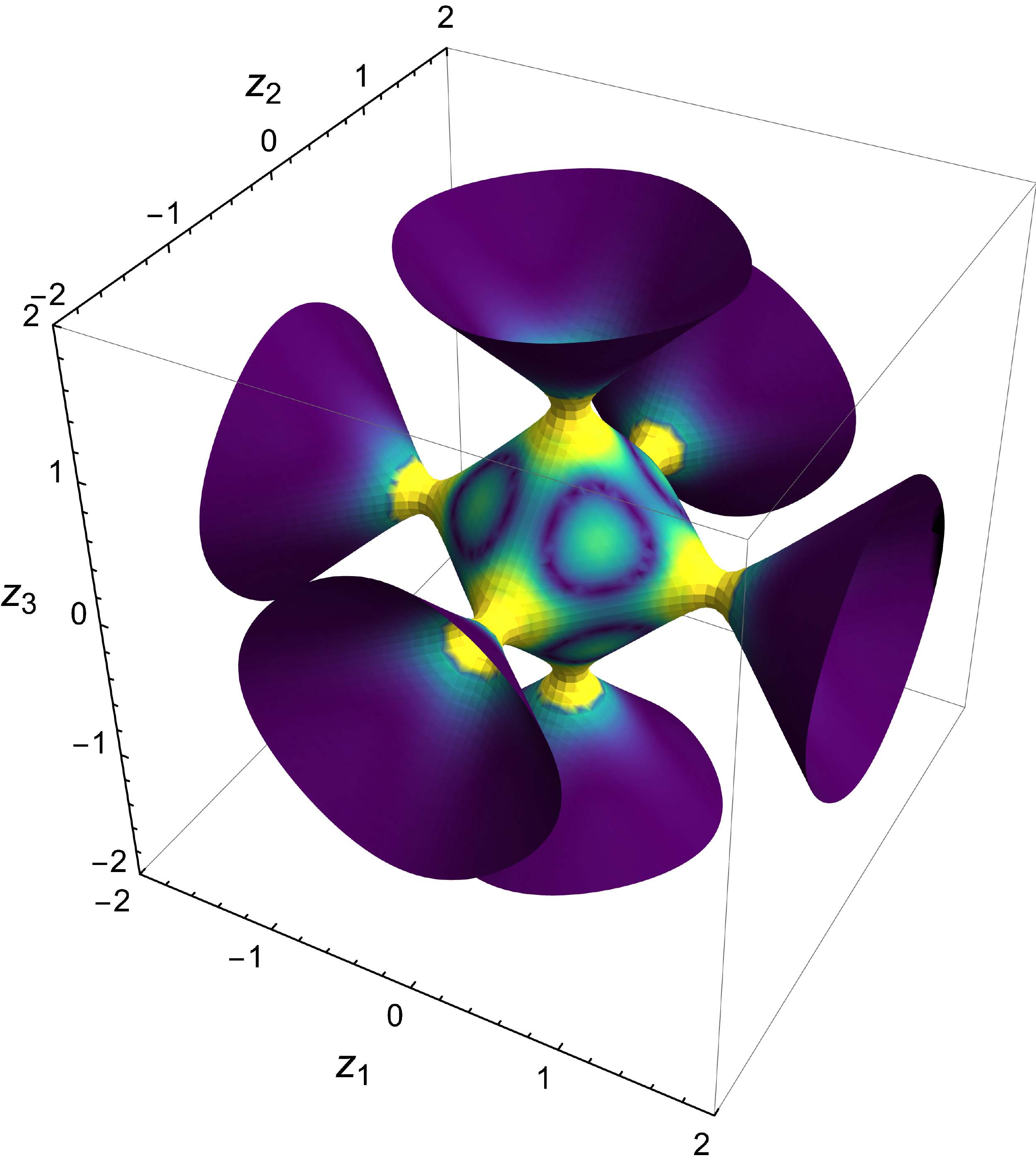}
        \caption{$\lambda = 1.48$, $J_\mathrm{FS}$}
    \end{subfigure}
    \vspace{10mm}
    \caption{Plot of $X_{\lambda^\sharp}\cap D_3 \cap \mathbb{R}^3$ with shading induced by the $c_1^2$. Brighter colors denote larger value of $c_1^2(J)$. In order to correctly take into account the identification $\lambda p \sim p$ for $\l\in\IC^\times$, we have used spectral networks to compute $\phi$ (See~\ref{sec:spectralNNs}).}\label{fig:realXcurvatures2}

\end{figure*}

Using the definitions for $X_{\lambda,t_1}^1$ and $X_{\lambda,t_2}^2$, we may study the asymptotic behavior of $F_1$ and $F_2$ defined by:
\begin{gather}
	F_1(t_1) := \int_{X_{\lambda,t_1}^1}	c_1(J_\lambda)^2\quad\text{and}\quad F_2(t_2) := \int_{X_{\lambda, t_2}^2}c_2(J_\lambda) ~.
\end{gather}
The numerical result for $F_2$ is shown on Figure~\ref{fig:f2Plots}.
\begin{figure*}[htb]
	\centering
	\includegraphics[width=0.7\textwidth]{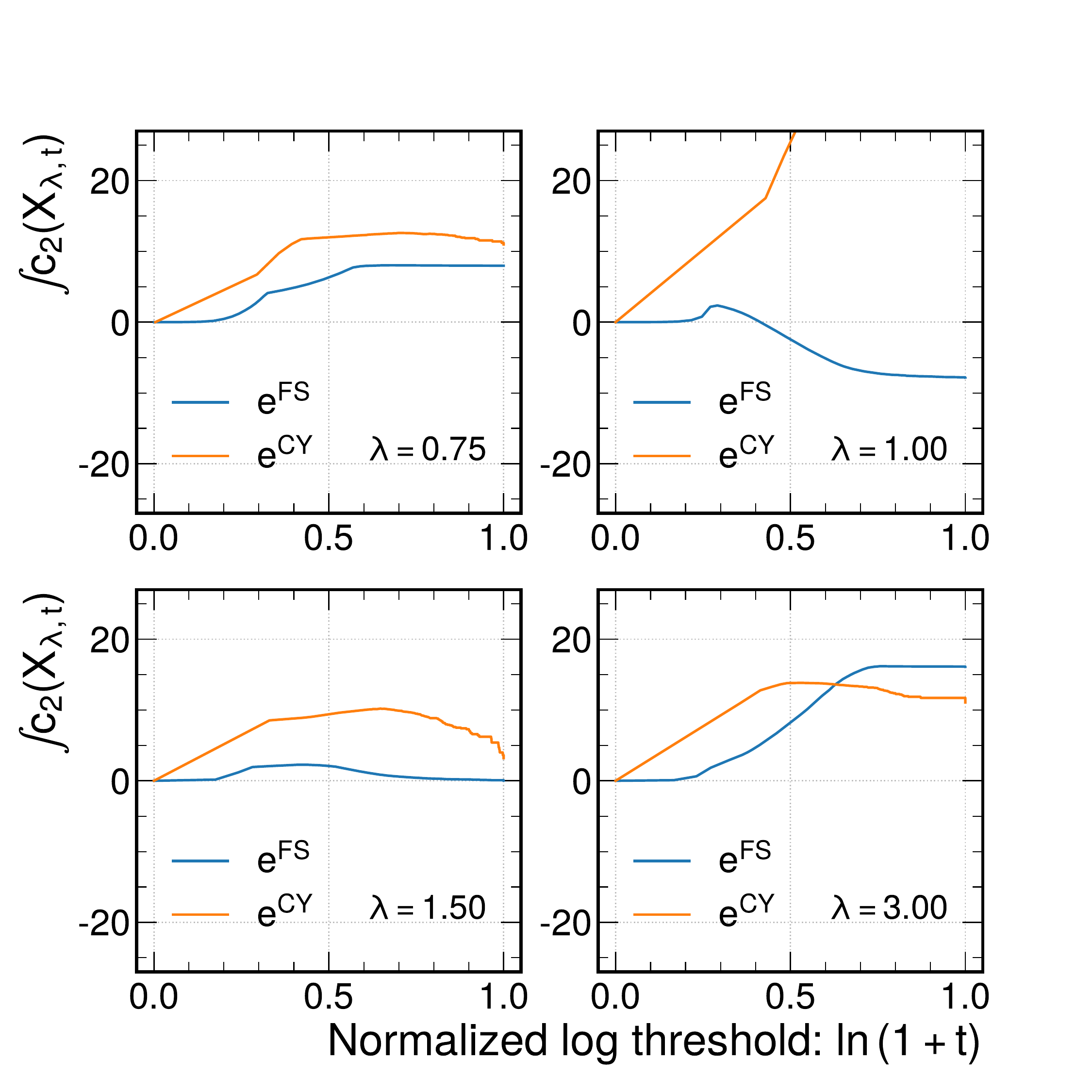}\\[15pt]
	\caption{Asymptotic behavior of $F_2$ for both Fubini--Study and Calabi--Yau metrics.}\label{fig:f2Plots}.
\end{figure*}
As can be observed from Figure~\ref{fig:f2Plots}, the $F_2^\textsf{FS}$, corresponding to the Fubini--Study metric, converges to the values shown in Table~\ref{table:c1c1c2FS}, whereas $F_2^\textsf{CY}$ seems to converge in all cases with the exception of $\lambda=1$. Comparing to the learning curves in Figure~\ref{fig:cefalu_sigma_loss}, we see that $\lambda=1$ exhibits the poorest training (\textit{i.e.}, highest sigma loss). We must also highlight the case of $\lambda=3$ as it gives the best results in terms of convergence and also yields the smallest loss among the singular models, this can be due to the fact that we have the smallest number of high curvature regions for $\lambda=3$. 

Similar behavior can be observed with $F_1$, for which the result of numerical computation is shown on Figure~\ref{fig:f1Plots}.
\begin{figure*}[htb]
	\centering
	\includegraphics[width=0.7\textwidth]{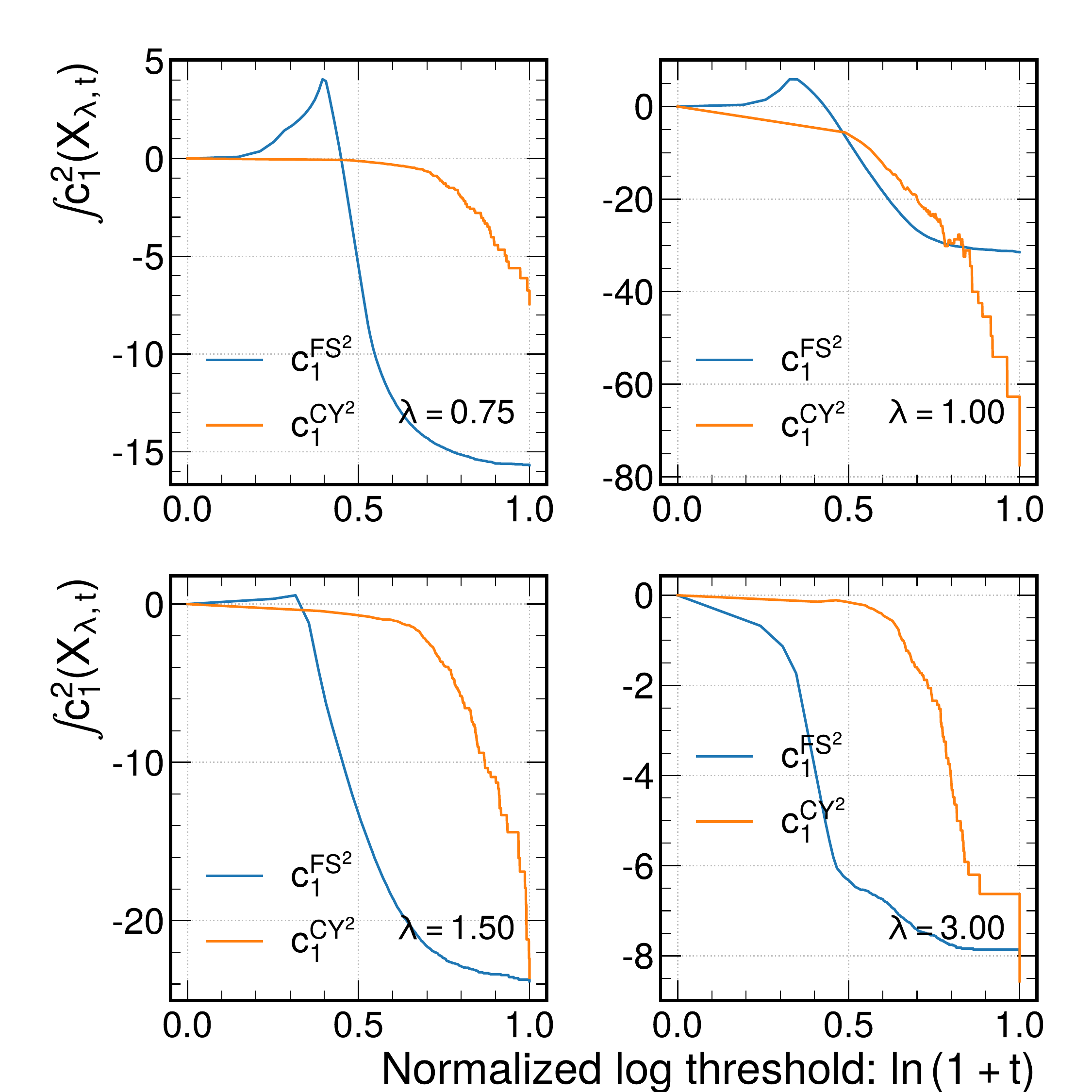}\\[15pt]
    \vspace*{4mm}
	\caption{Asymptotic behavior of $F_1$ for both Fubini--Study and Calabi--Yau metrics.}\label{fig:f1Plots}.
\end{figure*}
As it can be seen in that  figure, the quantity $F_1$ computed using the machine learned metric
becomes unstable and diverges once normalized $\log(1+t)$ is greater than $0.5$. Observing the histogram on Figure~\ref{fig:c1c1_hist} corresponding to $c_1^2$ shows a similar disagreement: the Fubini--Study and Calabi--Yau metric produce different results.

As a cautionary note, we observe that local data derived from the metric, such as the extent to which the Monge--Amp\`ere equation is satisfied pointwise, does not necessarily indicate whether the metric is able to reliably recover global, topological properties of $X$. Therefore, such local quantities may be unable to provide a complete diagnostic of the phenomenological suitability of the metric. 

\begin{figure*}[htb]
	\centering
	\includegraphics[width=0.9\textwidth]{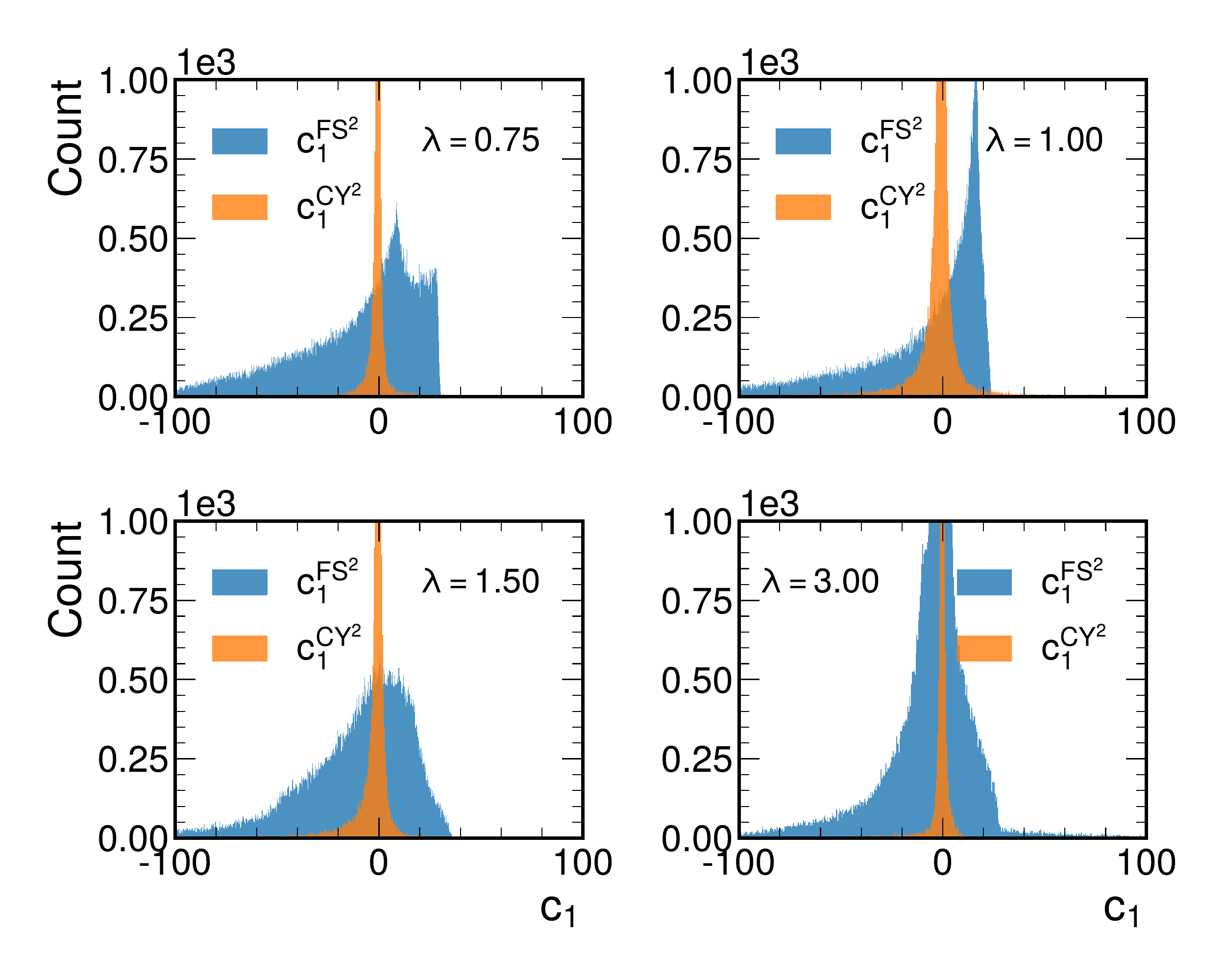}
	\caption{Distribution of the values of $c_1(J_\lambda)^2$ using both $J_\lambda^\textsf{FS}$ and $J_\lambda^\textsf{CY}$.}\label{fig:c1c1_hist}.
\end{figure*}

\comment{

\subsection{Dwork--Cefal\'{u} deformations}
In the Cefal\'{u} pencil, all singularities are isolated and of $A_1$ type. The simplest extension of this deformation which also includes singularities of dimension $>0$ is the $2$-parameter family of projective varieties given by:
\begin{gather}\label{eq:dworkCefaluFamily}
	X_{\lambda,\epsilon}	 = \left\{\sum_{i=0}^3z_i^4 - \frac{\lambda}{3}\left[\sum_{i=0}^3 z_i^2\right]^2 - \epsilon \prod_{i=0}^3 z_i = 0\right\}= \left\{p_{\lambda,\epsilon}(z) = 0\right\}\subseteq \mathbb{P}_\IC^3
\end{gather}
We may find the equation in $\lambda,\epsilon$ parametrizing the singular varieties by computing Gr\"{o}bner basis of ideal $I_{\lambda,\epsilon} = (p_{\lambda,\epsilon}(z) ) + J_{p_{\lambda,\epsilon}}$ where $J_{p_{\lambda,\epsilon}}$ is the Jacobian ideal of $p_{\lambda,\epsilon}$. The elimination is given by:
\begin{gather}
	G(I_{\lambda,\epsilon})	\cap \mathbb{C}[z_3] = \\\notag=(\epsilon^2 - 16)(9-9\lambda + 2\lambda^2)(9\epsilon^4 (\lambda - 3) - 768(3-4\lambda)^2(\lambda - 1) - 32\epsilon^2\lambda (27 - 36\lambda + 8\lambda^2))z_3^9 =\\\notag = f(\lambda,\epsilon)z_3^9
\end{gather}
Thus, the singular $X_{\lambda,\epsilon}$ are given by the roots of $f(\lambda,\epsilon)$. Consider $(\lambda,\epsilon) = (3/2,-4)$. In this case, we have:
\begin{gather}
	\int_{{X_{3/2,-4}}_s}	c_2(J_{3/2,-4}^\textsf{FS}) = 12,\qquad \int_{{X_{3/2,-4}}_s} c_1(J_{3/2,-4}^\textsf{FS})^2 = 36
\end{gather}
While this does violate Conjecture~\ref{conj1}, it satisfies $3\deg{c_2(J_{3/2,-4}^\textsf{FS})} = \deg{c_1(J_{3/2,-4}^\textsf{FS})^2}$. Moreover, looking at the values at each points, we observe a stronger condition:
\begin{gather}
	3c_2(J_{3/2,-4}^\textsf{FS}) = c_1(J_{3/2,-4}^\textsf{FS})^2
\end{gather}
As it is evident from the histogram shown on Figure~\ref{fig:dworkCefaluHist}.
\begin{figure*}[htb]
	\centering
	\includegraphics[width=0.6\textwidth]{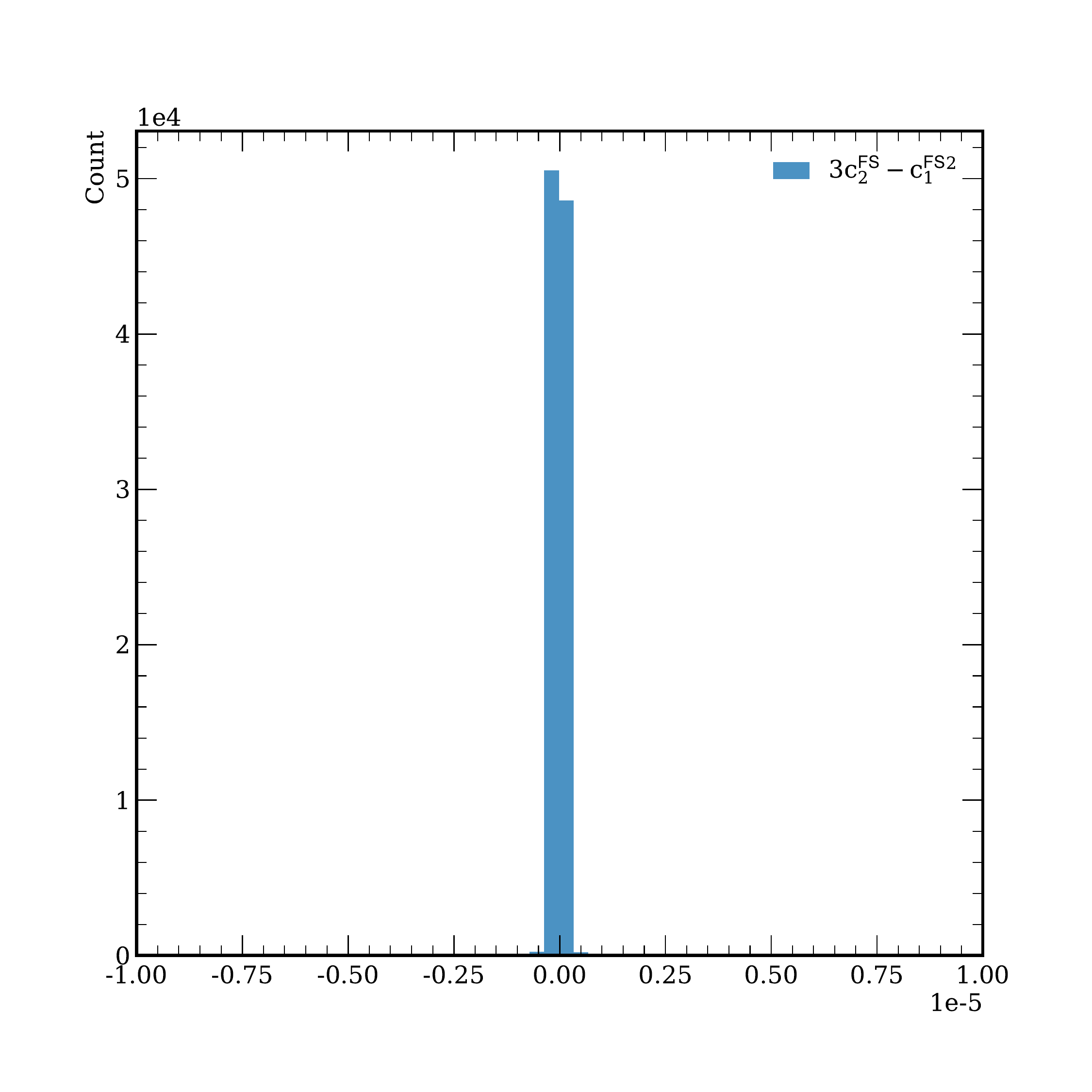}
	\caption{Distribution of the values of the difference $3c_2(J_{3/2,-4}^\textsf{FS}) - c_1(J_{3/2,-4}^\textsf{FS})^2$.}\label{fig:dworkCefaluHist}
\end{figure*}

Note that the prescription $\lambda = 0$ and $\epsilon = 4\psi$ of the 2-parameter deformation family~\eqref{eq:dworkCefaluFamily} defines the Dwork pencil of quartics:
\begin{gather}
	X_{\psi} := \left\{\sum_{i=0}^3z_i^4 - 4\psi \prod_{i=0}^3 z_i = 0\right\}\subseteq \mathbb{P}_\IC^4
\end{gather}
Note that $X_\psi$ is singular for $\psi^5 = 1$ and has total of $16$ singularities of $A_1$ type. Using the application of Prop.~\ref{prop:pfaffian}, we then obtain:
\begin{gather}\label{eq:dworkEuler}
	\int_{X_{\psi,s}}e(J) = 24 - 2\times 16 = -8,\quad\text{for}~\psi^5 = 1
\end{gather}
A similar asymptotic behavior to the one discussed in Section~\ref{sec:convergence} is shown in Figures~\ref{fig:f2Dwork} and~\ref{fig:f1Dwork}.

\begin{figure*}[htb]
	\centering
	\includegraphics[width=0.8\textwidth]{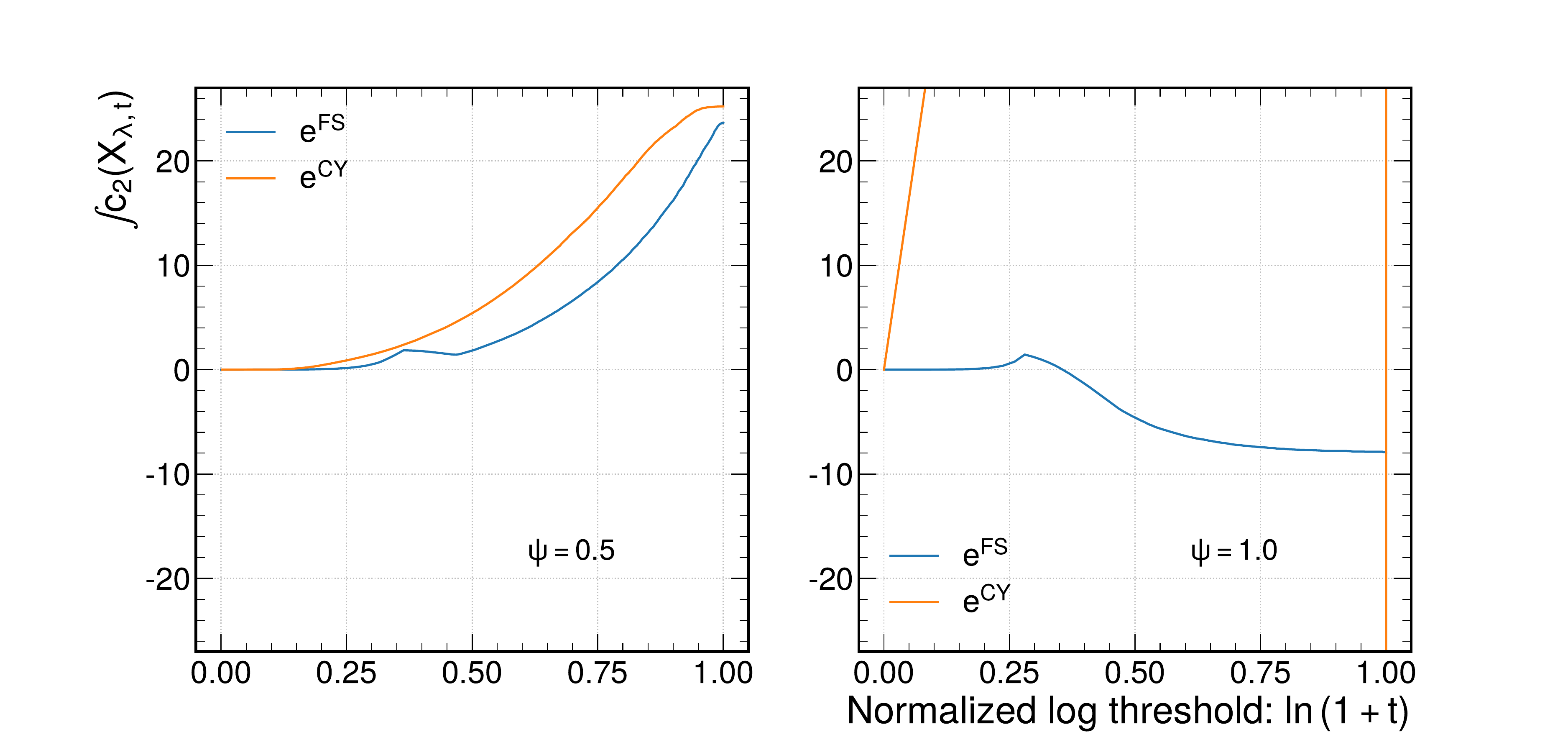}
    \vspace*{4mm}
	\caption{Asymptotic behavior of $F_2$ for the Dwork quartic.}\label{fig:f2Dwork}.
\end{figure*}

\begin{figure*}[htb]
	\centering
	\includegraphics[width=0.7\textwidth]{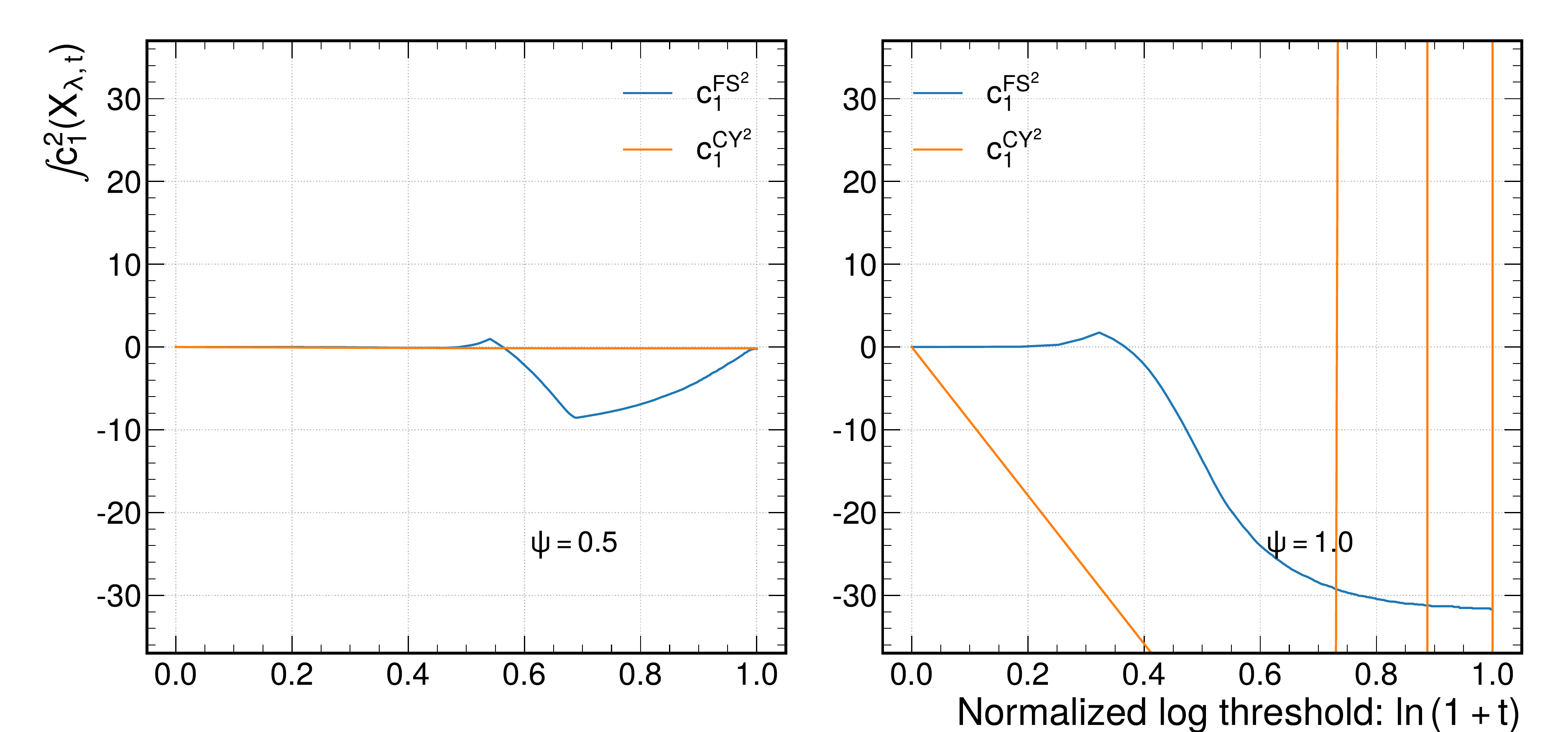}
    \vspace*{4mm}
	\caption{Asymptotic behavior of $F_1$ for the Dwork quartic.}\label{fig:f1Dwork}.
\end{figure*}

The distribution of the Euler density for the Dwork Quartic also satisfies the same positivity condition:
\begin{gather}
	e(J_\psi^\textsf{FS}) \geq 0\quad\text{and}\quad e(J_\psi^\textsf{CY})\gtrsim 0
\end{gather}

The distribution of $c_1(J_\psi)^2$ is shown on Figure~\ref{fig:c1c1_dwork_hist}.
\begin{figure*}[htb]
	\centering
	\includegraphics[width=0.8\textwidth]{graphics/c1c1_hist_comparison0.5,1.0 (1).pdf}
	\caption{Distribution of $c_1(J_\psi)^2$.}\label{fig:c1c1_dwork_hist}.
\end{figure*}

}

\subsection{Dwork quintic}
Similarly to the Dwork quartic, the Dwork quintic can be defined as a $1$-parameter family of Calabi--Yau threefolds:
\begin{gather}
    Z_\psi = \left\{\sum_{i=0}^4 z_i^5 - 5\psi \prod_{i=0}^4 z_i = 0\right\}\subseteq \mathbb{P}_\IC^4 ~.
\end{gather}
Recall that the Dwork quintic family is singular iff the complex structure parameter $\psi\in\IC$ is a $5$-th root of unity.
\begin{figure*}[htb]
    \centering
    \includegraphics[width=0.6\textwidth]{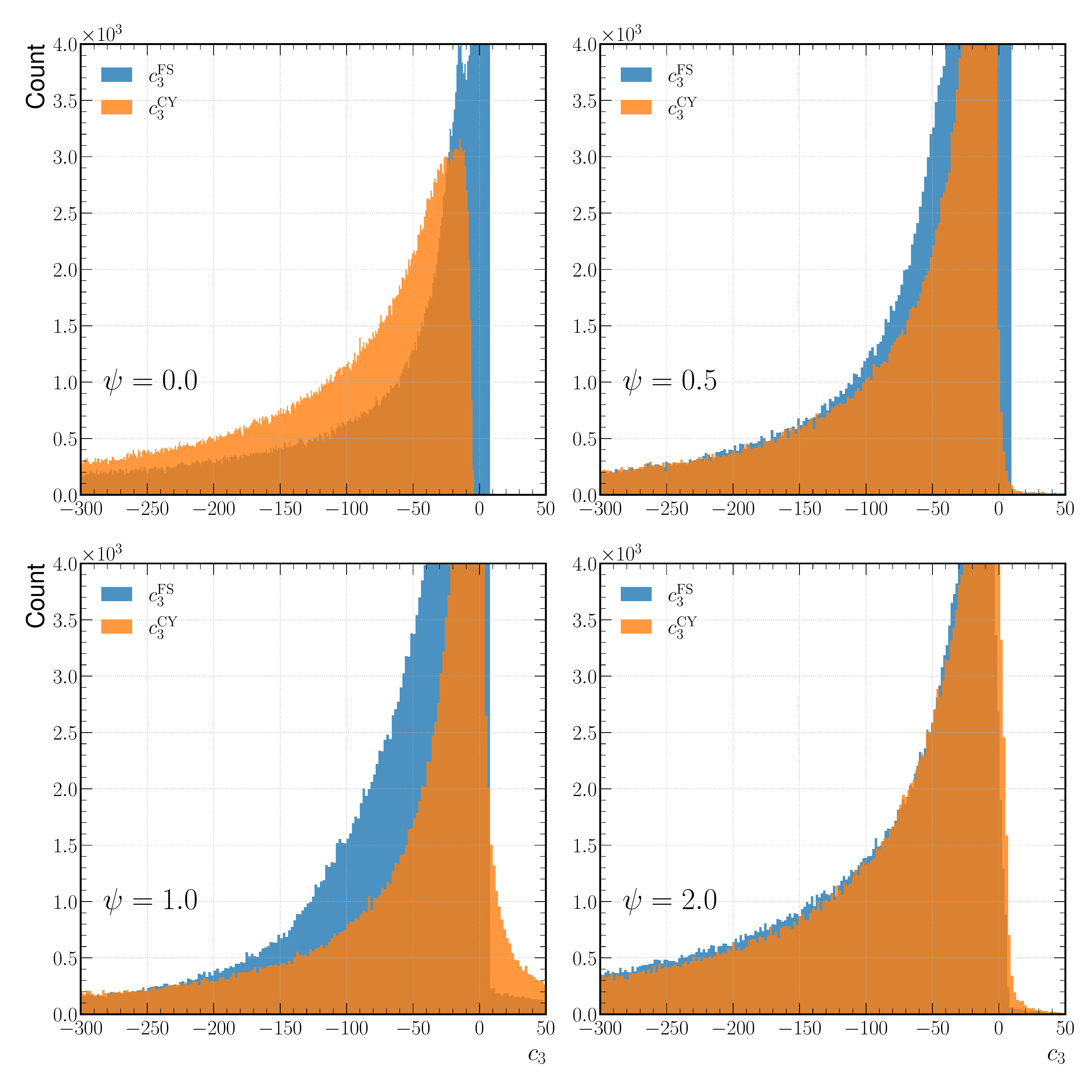}
    \caption{Curvature distributions $c_3(J_\psi)$ for the Fubini--Study and machine learned metric approximations. These distributions were obtained with the \texttt{JAX} implementation.}\label{fig:dworkQuintic_c3jax}
\end{figure*}

We may generalize Proposition~\ref{prop:pfaffian} to $n$-folds with isolated $A_1$ singularities as follows:
\begin{prop} Let $X\subseteq{\mathbb{P}}^m$ be a possibly singular projective variety with curvature form $J$ defined on the smooth locus $X_s$ of $X$. If $|\mathrm{Sing}~X|<\infty$ and the singularities have type $A_1$, then:
\begin{gather}
    \int_{X_s} e(J) = \deg{c_F(X)} - 2(-1)^{\dim{X}}|\mathrm{Sing}~X| ~.
\end{gather}
\end{prop}
\begin{proof}
\begin{gather}
    \int_{X_s} e(J) = \chi(X) - (-1)^{\dim{X}}|\mathrm{Sing}~X| = \deg{c_F(X)} - 2(-1)^{\dim{X}}|\mathrm{Sing}~X| ~.
\end{gather}
\end{proof}
Similarly as for Cefal\'{u}/Dwork quartic pencils, we consider the histograms of different possible top forms generated by the products of the characteristic forms. The histograms are shown on Figure~\ref{fig:dworkQuintic_c1c1c1}.

\begin{figure*}[htb]
    \centering
    \includegraphics[width=0.7\textwidth]{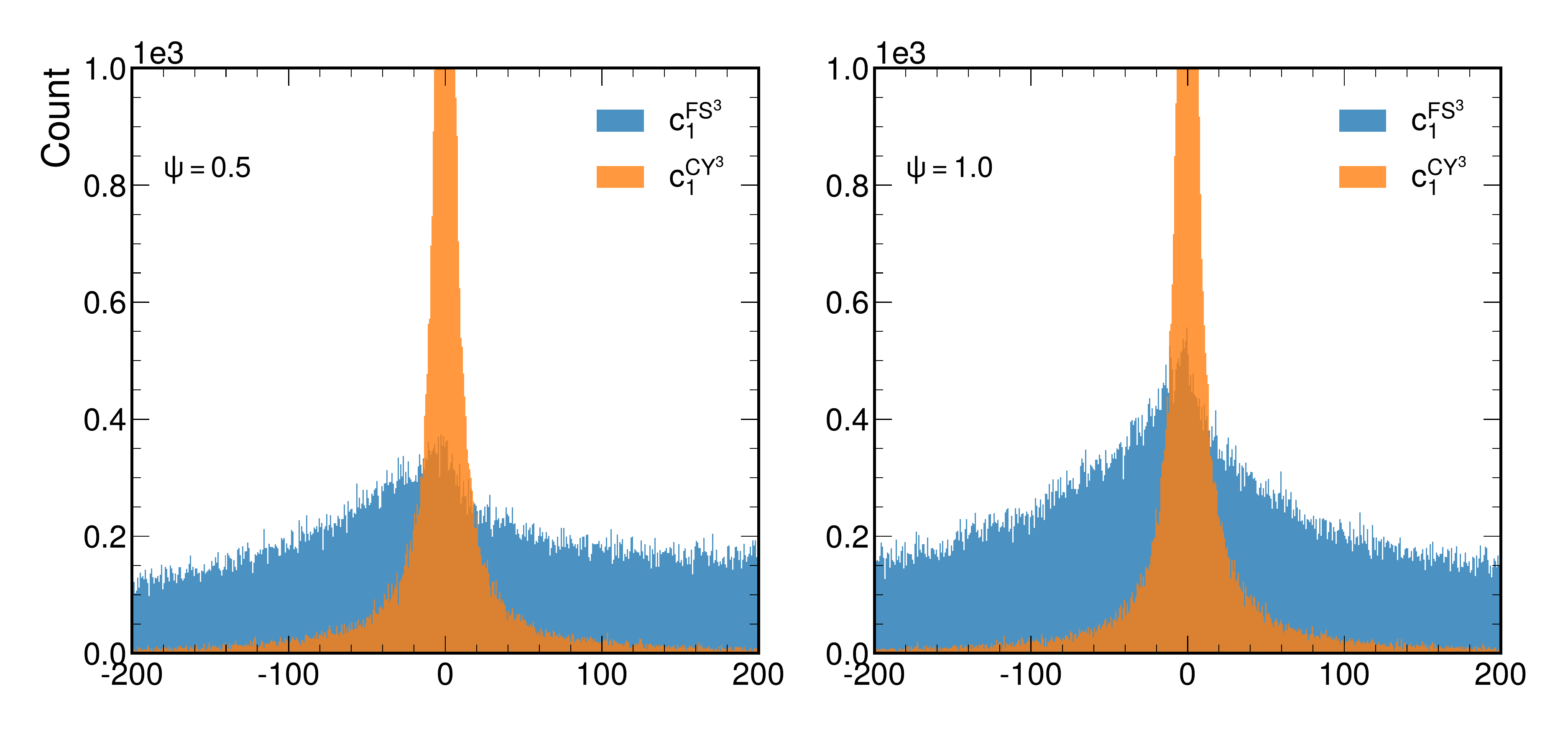}
    \caption{Distribution of the values of the $c_1^3$ for both $J_\psi^\mathrm{FS}$ and $J_\psi^\mathrm{CY}$ at $\psi\in\{1/2,~1\}$.}\label{fig:dworkQuintic_c1c1c1}
\end{figure*}

\subsubsection{Toric quintic vs.\ CICY quintic}\label{sec:cicytoric}
There are $7,890$ Calabi--Yau threefolds realized as complete intersections of polynomial equations in products of projective space~\cite{candelas:1987kf}.
Some of these are also Calabi--Yau hypersurfaces in toric varieties obtained from triangulations of reflexive polytopes in the Kreuzer--Skarke list~\cite{Kreuzer:2000xy}.
The quintic in $\mathbb{P}^4$ is one such geometry with a double description.
Thus, we can calculate a numerical Ricci-flat metric using \texttt{cymetric} by considering this manifold in either language.
The machine learned metric as a complete intersection Calabi--Yau (CICY) is better than the one obtained from the toric description, where in the latter a random choice of coefficients is used for the degree five defining equation, which is consistent with the general observation that the accuracy of the numerical method is significantly dependent on the choice of a point in the complex structure moduli space.

\begin{figure*}[htb]
    \centering
    \includegraphics[width=1.0\textwidth]{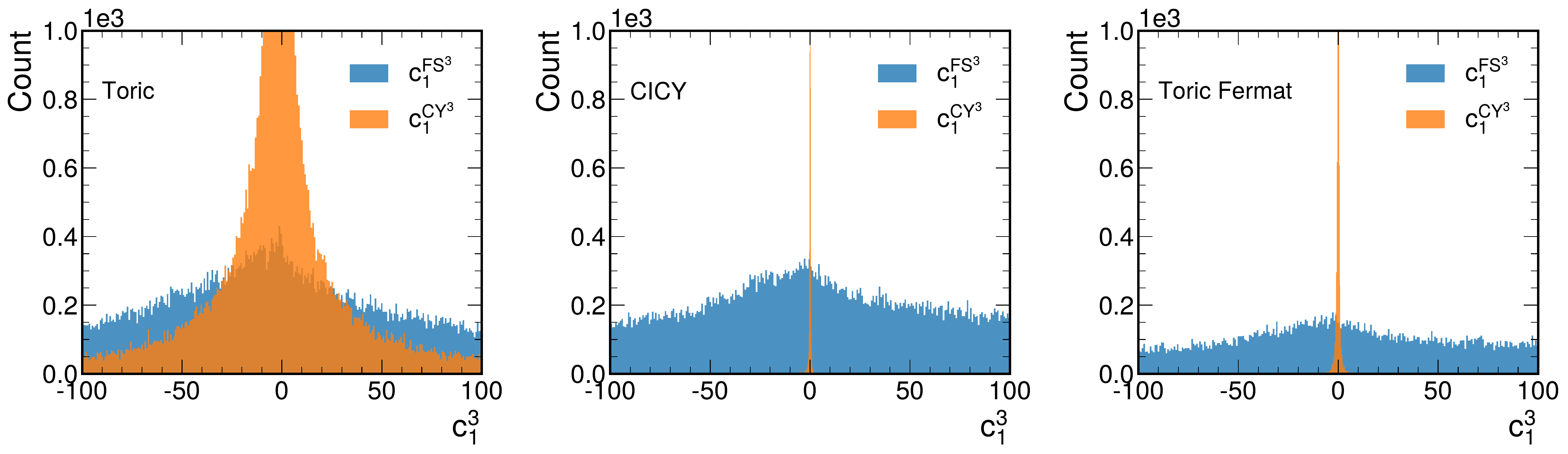}
    \vspace*{-3mm}
    \caption{Distribution of curvature densities $c_1^3$ for the toric quintic with random choice of coefficients on the left, CICY Fermat quintic ($\psi=0$) in the middle, and toric quintic at the same point in complex structure moduli as Fermat quintic on the right.}
    \label{fig:toric_cicy_comparison}
\end{figure*}
Indeed, by fixing the coefficients to be the same as that of the Fermat quintic, we see that the histogram for $c_1^3$, shown on the right in Figure~\ref{fig:toric_cicy_comparison}, is similar to that of CICY Fermat quintic shown in the middle in Figure~\ref{fig:toric_cicy_comparison}.

\setlength{\abovecaptionskip}{10pt}
\begin{table}[htb]
$$
\begin{array}{||@{~~}c@{~~} c@{~~} c@{~~} c@{~~}||} 
 \hline
\mathrm{Type} & \deg c_3(J) & \deg c_1(J)\wedge c_2(J) & \deg{c_1(J)^3}  \\ [0.5ex] 
 \hline\hline
 \mathrm{ToricFS}  & -183.87 & 16.24& 16.58 \\
 \mathrm{Toric}  & -193.46 & 11.31 & 8.27\\
 \mathrm{QuinticFS} & -196.43 & 1.28 & -0.28 \\
 \mathrm{Quintic}  & -203.55 & 0.09 & 0.00 \\[1ex] 
 \hline
\end{array}
$$
\caption{Values of the integrals of the possible top characteristic forms for threefold. The integrals were evaluated using MC integration with $100,000$ points.}
\end{table}

\subsection{Epilogue: Globally defined $\phi$-models and spectral neural networks  }
\label{sec:spectralNNs}
Note that the neural network architecture described above does not define a global smooth function $\phi\colon X \rightarrow \mathbb{R}$. This in turn implies that $J^\text{CY}$ and $J^\mathrm{ref}$ are not necessarily in the same K\"ahler class under the Ansatz $J^\text{CY} = J^\mathrm{ref} + \partial\overline{\partial}\phi$, in fact, $J^\mathrm{CY}$ might not even be well-defined. We observe this discrepancy by numerically studying the global topological characteristics of $(X,J^\mathrm{CY})$. In particular, from Figure~\ref{fig:eulerChar}, it is evident that the numerical computation of the topological Euler characteristic $\chi(X)$, computed using~\eqref{eq:topChar} deviates significantly from the expected value of $\chi(X) = 24$ for non-singular $X$.
\begin{gather}\label{eq:topChar}
	\chi(X) = \int_{X} \mathrm{Pf}(J^\mathrm{CY}) \approx \frac{(-2i)^2}{2!}\frac{1}{N}\sum_{i=1}^N \frac{w(p_i) c_2(p_i)}{(\Omega\wedge\overline{\Omega})(p_i)} ~.
\end{gather}
Furthermore, note that if $\phi$ had been a global function on $X$, then, the difference of Pfaffians must be exact~\cite{cibotaru2021odd, moroianu2022higher}:
\begin{gather}
	\mathrm{Pf}(J^\mathrm{CY}) - \mathrm{Pf}(J^\mathrm{ref})	 = \mathrm{d}\,\mathrm{TPf}(J^\mathrm{CY}, J^\mathrm{ref}) ~.
\end{gather}
Thus, the numerical approximation~\eqref{eq:topChar} should produce results which are within the expected $\chi(X) = 24$ within the margin of the error of numerical integration error.

Recently there has been some progress in designing $\IC$-homogeneous and holomorphic neural networks~\cite{douglas2021holomorphic, Douglas:2020hpv}. The idea is to construct pair-wise products of the homogeneous coordinates and apply activation functions which are holomorphic. This allows one to define networks $\phi$ (called \emph{biholomorphic networks}) such that:
\begin{gather}
	\forall \lambda\in\IC\colon \phi(\lambda\cdot z) = |\lambda|^{2k}\phi(z) ~, \qquad\text{for some}~k\in\mathbb{N} ~.
\end{gather}
However, although homogeneous, $\phi$ does not define a global function on $X$, but a section of $\mathcal{O}_X(k)\otimes \overline{\mathcal{O}_X(k)}$.

Motivated by the observations above, we propose a modified method which allows us to define $\phi$ to be a global function in $\mathcal{C}^\infty(X)$. To define the setup, let $X$ be a CICY defined as a zero locus of homogeneous polynomials $\{f_i\}_{i=1,\dots,N}$ where $f_i \in \IC[Z_0,\dots,Z_{n_i}]$, thus $X$ lies in $\mathbb{P}_\IC^{n_1}\times \dots \times \mathbb{P}_\IC^{n_N}$. For each component $\mathbb{P}_\IC^{n_i}$ of the product, define a mapping:
\begin{gather}
	\alpha_{n_i}\colon \mathbb{P}_\IC^{n_i}	\longrightarrow \mathbb{C}^{n_i+1,n_i+1} ~,
\end{gather}
whose action on a general point $p\in [Z_0\colon Z_1\colon \dots\colon Z_{n_i}]\in\mathbb{P}_\IC^{n_i}$ is defined as:
\begin{gather}
	\alpha_{n_i}(p) = \left[\begin{matrix}
		\displaystyle \frac{Z_0 \overline{Z_0}}{|Z|^2} && \displaystyle\frac{Z_0 \overline{Z_1}}{|Z|^2} && \dots && \displaystyle\frac{Z_0 \overline{Z_{n_i}}}{|Z|^2}  \\
		\displaystyle\frac{Z_1 \overline{Z_0}}{|Z|^2}  && \displaystyle\frac{Z_1 \overline{Z_1}}{|Z|^2} && \dots && \displaystyle\frac{Z_1 \overline{Z_{n_i}}}{|Z|^2} \\
		\vdots && \vdots && \ddots && \vdots \\
		\displaystyle \frac{Z_{n_i} \overline{Z_0}}{|Z|^2} && \displaystyle\frac{Z_{n_i} \overline{Z_1}}{|Z|^2} && \dots && \displaystyle\frac{Z_{n_i} \overline{Z_{n_i}}}{|Z|^2}
	\end{matrix}\right] ~.
\end{gather}
Note that $\alpha_{n_i}$ is a well-defined global smooth function on $\mathbb{P}_\IC^{n_i}$ and thus its restriction $\alpha_{n_i}\vert_{X}$ is a well-defined smooth function on $X$. Furthermore, note that the components of $\alpha_{n_i}$ correspond to $k_\phi = 1$ basis used in~\cite{PhysRevD.103.106028} to build the eigenfunctions of Laplace operator $\Delta$, thus we shall refer to the layer of the neural network which applies $\alpha_n$ as a spectral layer and the corresponding neural network - a spectral neural network. We then decompose $\alpha_{n_i}$ into a real and imaginary components:
\begin{gather}
	\mathrm{ReIm}\colon \alpha_{n_i}(p) \mapsto \left(\mathrm{Re}\circ \alpha_{n_i}(p),~\mathrm{Im}\circ \alpha_{n_i}(p)\right)/\text{``redundancies''} ~.
\end{gather}
Bearing in mind that $\alpha_{n_i}$ is a Hermitian matrix, we see that $\mathrm{ReIm}\circ\alpha_{n_i}$ is made up of $(n_i+1)^2$ independent real entries. We then define $\rho\colon \mathbb{R}^{\sum_i(n_i+1)^2}\longrightarrow \mathbb{R}$ to be a neural network with $d$-layers and $W_i$-nodes in layer $i \in \{1,\dots, d\}$.
Finally, we combine these mappings to define $\phi\colon X\rightarrow\mathbb{R}$ as:
\begin{gather}
	\phi := \rho\circ \begin{pmatrix}
		\mathrm{ReIm}\circ \alpha_{n_1}\vert_X,~ \mathrm{ReIm}\circ \alpha_{n_2}\vert_X,~\dots,~\mathrm{ReIm}\circ \alpha_{n_N}\vert_X
	\end{pmatrix} ~.
\end{gather}
That $\phi$ is well-defined on $X$ is trivial, since for all $p\,{=}\, ([Z_0^{(1)}\colon\dots\colon Z_{n_1}^{(1)}],\dots, [Z_0^{(N)}\colon \dots \colon Z_{n_N}^{(N)}])$ in $X$, and for every $\lambda = (\lambda_1,\dots,\lambda_{N})\in{\IC^\times}^N$, we have:
\begin{gather}
	\phi(\lambda\cdot p)	 = \phi((\lambda_1 p_1,\dots, \lambda_N p_N)) = \rho\circ \begin{pmatrix}
		\mathrm{ReIm}\circ \alpha_{n_1}(\lambda_1 p_1),~\dots,~\mathrm{ReIm}\circ \alpha_{n_N}(\lambda_N p_N)
		\end{pmatrix} = \\ \notag =\rho\circ \begin{pmatrix}
		\mathrm{ReIm}\circ \alpha_{n_1}(p_1),~\dots,~\mathrm{ReIm}\circ \alpha_{n_N}(p_N) 
		\end{pmatrix} = \phi(p) ~.
\end{gather}

The spectral neural network architecture of $\phi$ is shown on Figure~\ref{fig:spectralNN}. Note that the spectral networks with $0$ hidden layers and $x\mapsto\ln(x)$ activation function are equivalent to algebraic metrics with $k=1$~\cite{Headrick:2005ch, Headrick:2009jz}.
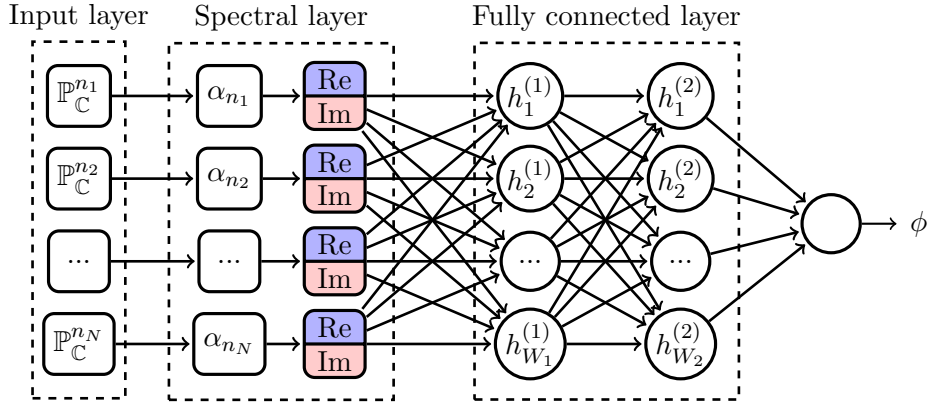
\begin{figure}[htb]
\centering
\begin{tikzpicture}[shorten >=1pt,->,draw=black, line width=1]
	\tikzstyle{RR}=[very thick, rounded corners, draw=black, minimum size=23];
	\tikzstyle{Split} = [rectangle split, rectangle split parts=2, rectangle split part fill={blue!30,red!20}, rounded corners, draw=black, very thick, minimum size=23, inner sep=2pt, text centered];
	\tikzstyle{neuron}=[very thick,circle,draw=black,minimum size=22,inner sep=0.5,outer sep=0.6]

	\foreach \name / \y in {1,...,4}
		\node[RR] (I-\name) at (0,-1.1*\y cm) {
			\ifnum \y < 3
				$\mathbb{P}^{n_\y}_\mathbb{C}$
			\fi
			\ifnum \y = 3
				$...$
			\fi
			\ifnum \y = 4
				$\mathbb{P}^{n_N}_\mathbb{C}$
			\fi
			};

	\foreach \name / \y in {1,...,4}
		\node[RR] (BH-\name) at (\layersep, -1.1*\y cm) {
			\ifnum \y < 3
				$\mathrm{\alpha}_{n_\y}$
			\fi
			\ifnum \y = 3
				$...$
			\fi
			\ifnum \y = 4
				$\mathrm{\alpha}_{n_N}$
			\fi
		};

	\foreach \i in {1,...,4}
		\path (I-\i) edge (BH-\i);

	\foreach \name / \y in {1,...,4}
		\node [Split] (ReIm-\name) at (1.7*\layersep, -1.1*\y cm)  {$\mathrm{Re}$\nodepart{two} $\mathrm{Im}$};

	\foreach \i in {1,...,4}
		\path (BH-\i) edge (ReIm-\i);

	\foreach \name / \y in {1,...,\widthNN}
		\path[yshift=(\widthNN-4)*0.5cm]
			node[neuron] (N1-\name) at (3*\layersep, -1.1*\y cm) {
				\ifnum \y < \numexpr\widthNN-1\relax
					$h_\y^{(1)}$
				\fi
				\ifnum \y = \numexpr\widthNN-1\relax
					...
				\fi
				\ifnum \y > \numexpr\widthNN-1\relax
					$h_{W_1}^{(1)}$
				\fi
			};
			
	\foreach \name / \y in {1,...,\widthNN}
		\path[yshift=(\widthNN-4)*0.5cm]
			node[neuron] (N2-\name) at (4*\layersep, -1.1*\y cm) {
				\ifnum \y < \numexpr\widthNN-1\relax
					$h_\y^{(2)}$
				\fi
				\ifnum \y = \numexpr\widthNN-1\relax
					...
				\fi
				\ifnum \y > \numexpr\widthNN-1\relax
					$h_{W_2}^{(2)}$
				\fi
			};

	\foreach \i in {1,...,4}
		\foreach \j in {1,...,\widthNN}
			\path (ReIm-\i) edge (N1-\j);
	
	\foreach \i in {1,...,\widthNN}
		\foreach \j in {1,...,\widthNN}
			\path (N1-\i) edge (N2-\j);

	\path[yshift=0.5cm]
		node[neuron, pin={[pin edge={->, black, thick}]right:$\phi$}] (phi) at (5*\layersep, -1.1*3) {};

	\foreach \i in {1,...,\widthNN}
		\path (N2-\i) edge (phi);

	\node[draw, dashed, rectangle, minimum height=135pt, fit=(I-1) (I-4), label={[align=center]above:Input layer}] {};

	\node[draw, dashed,rectangle, minimum width=85pt, minimum height=135pt, fit=(BH-1) (ReIm-4), label={[align=center]above:Spectral layer}] {};
	
	\node[draw, dashed, rectangle, minimum height=135pt, minimum width=100pt, fit=(N1-1) (N2-\widthNN), label={[align=center]above:Fully connected layer}] {};

\end{tikzpicture}
\vspace{0.4cm}
\caption{Spectral neural network architecture: Prior to the fully connected neural network we introduce the spectral layer, taking real and imaginary parts of $\mathbb{C}^*$-invariant quantities. }\label{fig:spectralNN}
\end{figure}
In Figure~\ref{fig:eulerCharAll}, we present 
the values for the Euler number obtained using the spectral neural networks. The numerical integrals computed using the spectral networks sketched in Figure~\ref{fig:spectralNN} are mostly close to the expected Euler characteristic, within the margin of the error due to the Monte Carlo integration. 
\begin{figure*}[htb]
	\centering
	\includegraphics[width=1.0\textwidth]{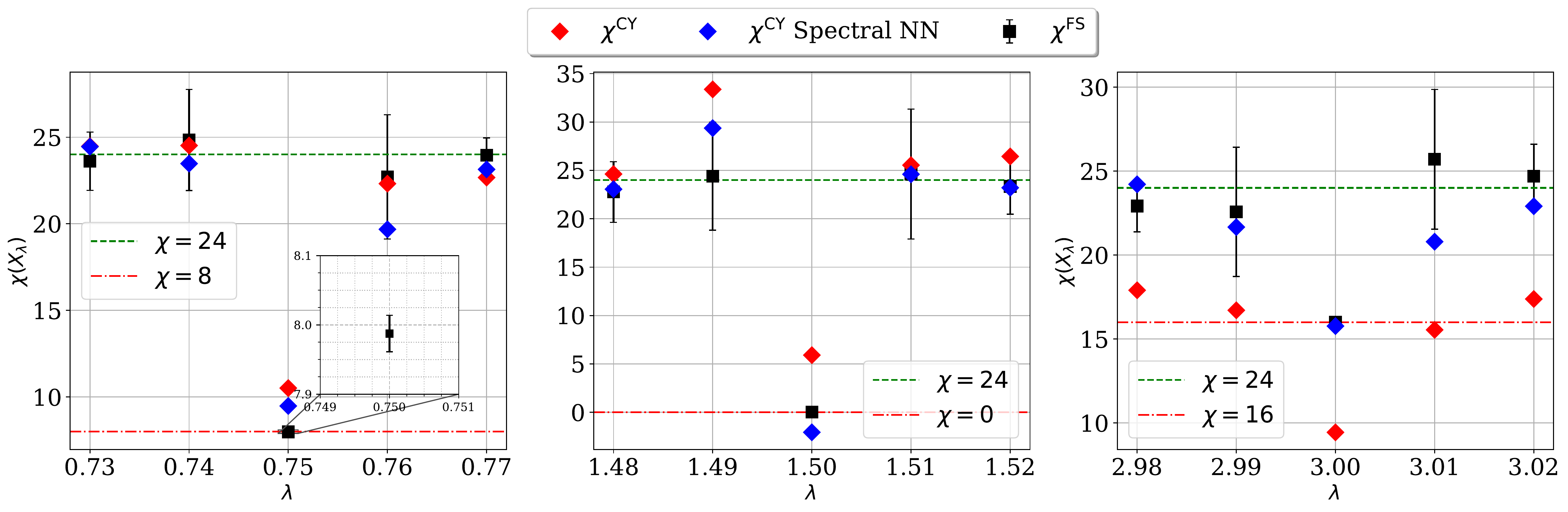}
	\caption{Numerical values of~\eqref{e:chiK3} along the Cefal\'{u} pencil. Black points and error bars showing a $95\%$ confidence interval are associated to Fubini--Study results, while the red and blue dots correspond to the machine learned metric approximation using fully-connected and spectral networks, respectively. See Appendix~\ref{a:numIntegration} for the details on integration.}\label{fig:eulerCharAll}
\end{figure*}

\begin{figure*}[htb]
    \centering
    \includegraphics[width=0.6\textwidth]{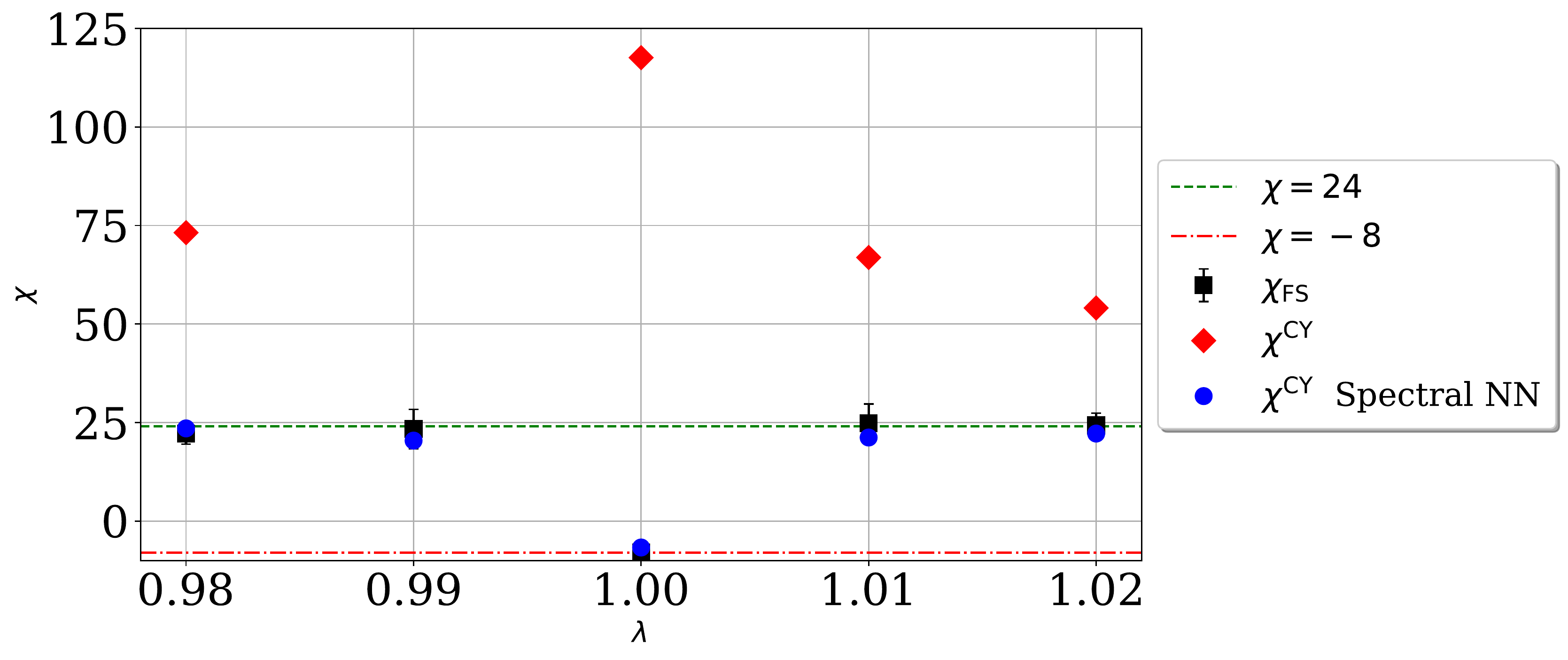}
    \caption{Numerical values of~\eqref{e:chiK3} along the Cefal\'{u} pencil near $\lambda = 1$. The plot markers are the same as in Figure~\ref{fig:eulerChar}. The value of $\chi^\text{CY}$ using fully-connected network at $\lambda = 0.99$ is off the chart: $\chi^\text{CY} \approx 85510$.}\label{fig:eulerCharLambda1}
\end{figure*}

The numerical values of~\eqref{e:chiK3} along the Cefal\'{u} pencil near $\lambda = 1$ are shown on the Figure~\ref{fig:eulerCharLambda1}. There are no convergence issues for spectral network Euler number at $\lambda=1$, in sharp contrast with the fully connected neural network. The convergence plot in shown in Figure~\ref{fig:convl1}.
\begin{figure*}[htb]
	\centering
	\includegraphics[width=.99\textwidth]{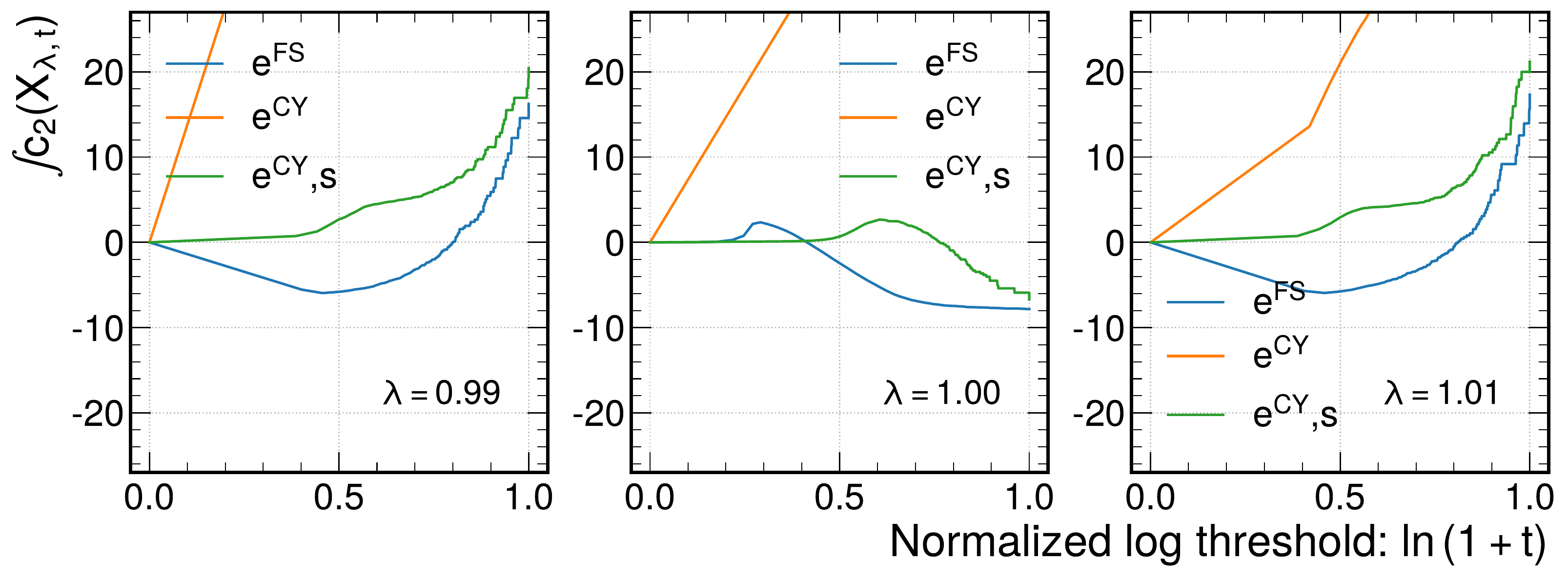}.
	\caption{Convergence plot for $c_2(X_\lambda)$ around $\lambda=1$; the spectral network results (green, ``s'' subscript) show significant improvement.}\label{fig:convl1}
\end{figure*}

\begin{table}\centering
\begin{tabular}{||ccccc||}
\hline
      $\mathbf{\lambda}$ & $~|\mathrm{Sing}~X|~$ & $~\mathrm{deg}~c_2(J_\lambda^\mathrm{FS}~)$ & ~$\mathrm{deg}~c_2(J_\lambda^\mathrm{CY})$  &~ $\mathrm{deg}~c_2(J_\lambda^\mathrm{CY})$ \\
& &  &  w/ Spectral Layer~ &~ w/o Spectral Layer~\\
      \hline\hline
      $0$ & $0$ & $24$ & $24$ & $24$ \\
      $3/4$ & $8$ & $7.99\pm 0.03$ & $9.48$ & $10.52$ \\
      $1$ & $16$ & $-7.99\pm 0.08$ & $-6.71$ & $117.59$ \\
      $3/2$ & $12$ & $0.0\pm 0.1$ & $-2.07$ & $5.90$ \\
      $3$ & $4$ & $16.00\pm 0.09$ & $15.77$ & $9.44$ \\ \hline
  \end{tabular}
    \caption{Values of the Monte Carlo approximations of the integrals of $c_2(X_\lambda)$. Note that the results computed using spectral networks are closer to the expected value. The points that were sampled in order to compute above integrals are different than the ones sampled in Figure~\ref{fig:eulerChar}.}\label{tab:eulercomparison}
\end{table}

The plots in Figure~\ref{fig:convl1} highlight a sharp contrast among the $\phi$-model neural network and the spectral neural network. As one includes points close to the singularities, the curvature for the $\phi$-model starts diverging. This suggests that as one approaches the singularity the neural network is not to be trusted. This is not the case for the spectral networks. Here, as one continues adding points to the curvature integral it keeps following the Fubini--Study. In Table~\ref{tab:eulercomparison}, we notice that spectral neural networks and Fubini--Study produce similar results for the curvature.

\section{Conclusions}

In this work we have considered two families of Calabi--Yau manifolds: the Cefal\'u family of quartics and the more broadly studied Dwork quintic family. For both of these, we have developed the algorithms to compute topological quantities derived from their corresponding Chern characters. This implementation can be easily extended to the whole CICY dataset. Our algorithms utilize some of the neural network approximation models of the \texttt{cymetric} package: the so called \texttt{PhiModel}. We also employ our own \texttt{JAX} implementation of this. 

Computation of topological quantities is a crucial fitness check for numerical Calabi--Yau metrics. At first one might think that these relatively straightforward computations automatically work out as they are metric independent. However, one has to bear in mind that the possible neural network approximations constitute a far broader set of solutions than that of globally defined K\"ahler metrics. Choosing smooth activation functions for the neural network ensures that the metric is smooth over each of the patches. Similarly, if the metric is obtained from the so-called \texttt{PhiModel}, over each patch one has $dJ=0$, satisfying some local form of K\"ahlerity. In the matching of patches, however, it is not guaranteed \textit{a priori} that the perturbation $\phi_{NN}$ respects the K\"ahler transformation rules of the seed K\"ahler potential (in our case, the Fubini--Study potential). That is an inherent issue with these numerical approximations and for the cases in which this situation is non-negligible we expect significant deviations when computing topological quantities. 

Deviations from the expected values are not necessarily due to the approximation: they might also be due to the sample size. In the smooth Calabi--Yau cases, we have observed that the neural network approximation to the flat metric behaves well over the entire manifold, and convergence to the right Euler number is achieved. In some of the singular cases, the neural network approximation gives curvature values in the vicinity of the singularities, that once weighted in the Euler number integral produce divergent results. While this is expected, it is also problematic, as increasing the number of points inevitably brings us close to the singularities, leading ultimately to unacceptable metrics. We notice that these situations can occur in the case of the $\phi$-model neural network for some values of $\lambda$. As a surprising result, we observe that for the spectral networks, the curvatures near the singularities remain in check and that the global computations match the Fubini--Study results within numerical errors. We observed that for all values of $\lambda$ considered here there is convergence for the Euler number as one increases the number of points. This suggests that in the singular cases, the spectral neural network is a good generalization of the metric in high curvature regions. The best results were obtained in cases where the number of points used for training is of similar order as the number of points used for the integration. 

Spectral networks seem to be able to provide reliable topological data even for singular manifolds. As a proof of principle we demonstrated how the spectral network predictions for the Euler number are within error bars for the Cefal\'{u} family of quartics. The spectral networks also lead to smaller sigma losses compared to a standard neural network approximation: As we see in Figure~\ref{fig:cefalu_spectral_vs_fullyconnected_at_0}, the $\sigma$ loss for the spectral network at $\lambda=0$ is below $10^{-3}$. This is consistent with the level of accuracy attained for $k=8$ in~\cite{Headrick:2009jz}. 

\begin{figure*}[htb]
    \centering
    \includegraphics[width=0.7\textwidth]{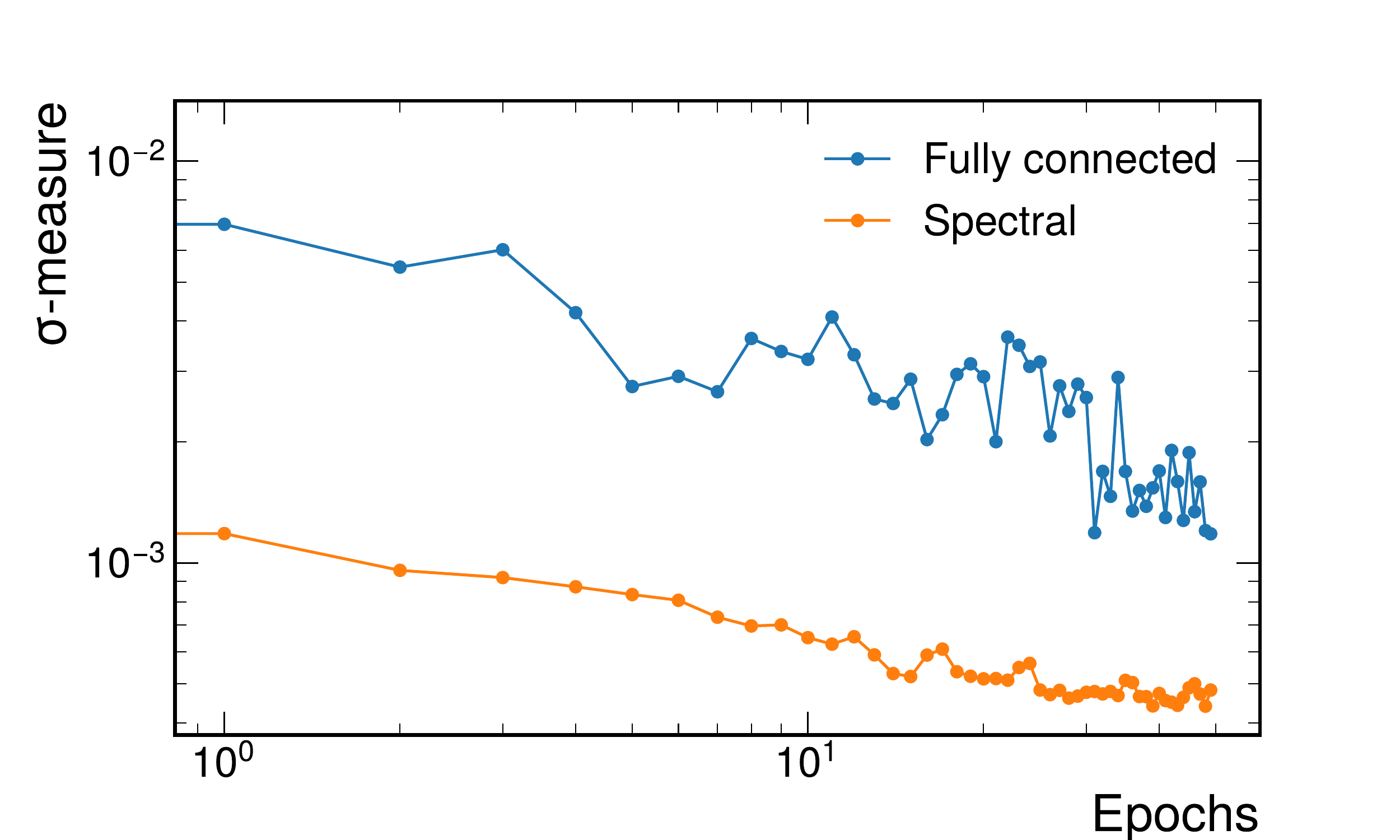}
    \vspace{5mm}
    \caption{Evolution of $\sigma$ loss at $\lambda = 0$ using fully-connected and spectral networks indicates a better performance for the latter. The $\sigma$ loss is evaluated on the validation set. This network is elaborated further in~\cite{spectralNetworks}.}\label{fig:cefalu_spectral_vs_fullyconnected_at_0}
\end{figure*}

A globally consistent Calabi--Yau metric must be able to produce the right value for any topological value on the manifold. Having numerical approximations implies having errors in the computation. A natural question is what error values are tolerable and how they can be related to the error (loss) function in the numerical approximation. Clearly, topological quantities are not sensible to flatness, but as we already highlighted, they are a crucial check for the global consistency of the metric approximation. This global consistency if of utmost importance, particularly in computation of ``global'' quantities such as the Yukawa couplings for a given string compactification model. 

Our work makes use of the \texttt{JAX} $\int_Q c_3$,
where the third Chern form $c_3$ is derived from the curvature two-form on a Hermitian manifold. Using a three-layer densely connected neural network with 64 units in each layer as the $\phi$-approximant,\footnote{The numerical experiments were performed on an Intel 16-Core Xeon Gold 5218 machine with an Nvidia A100 40GB GPU.} employing $1024$ points in the Monte Carlo integration, this computation takes $387 \text{ms} \pm 7.54\ \text{ms}$ using \texttt{JAX} compared to $4.41 \text{s} \pm 18.3\ \text{ms}$ (mean $\pm$ standard deviation over $7$ runs.)

We intend to open source our codebase as a fully-fledged package in a forthcoming publication~\cite{spectralNetworks}.

\section*{Acknowledgments}
We thank Carl Henrik Ek, Mario Garcia-Fernandez, Mathis Gerdes, Matt Headrick, and Fabian Ruehle, for helpful discussions, Oisin Kim for critical and instructive suggestions, and Katrin Wendland for discussions on the ``$\approx$''-identification subtleties in~\eqref{e:8->16} and in Figure~\ref{f:Cefalu}.
We are grateful to the anonymous referee for helpful and constructive comments on the text.
PB would like to thank the Hamilton Institute and Mathematics Department at Trinity College Dublin for their hospitality.
GB also wants to thank the 2022 IAIFI Summer School including the computational resources provided through the school.
TH is grateful to the Department of Mathematics, University of Maryland, College Park, MD and the Physics Department of the Faculty of Natural Sciences of the University of Novi Sad, Serbia, for the recurring hospitality and resources.
PB and GB are supported in part by the Department of Energy grant DE-SC0020220.
VJ is supported by the South African Research Chairs Initiative of the Department of Science and Innovation and the National Research Foundation.
CM is supported by
a Fellowship with the Accelerate Science program at the Computer Laboratory, University
of Cambridge. JT is supported by a studentship with the Accelerate Science Program. 

\appendix

\section{Numerical integration}\label{a:numIntegration}
The sampled points are not uniformly distributed with respect to the desired Ricci-flat metric~\cite{braun2008calabi, shiffman1999distribution}. Instead, the density sampled points $\rho$ is built in such a way that 
\be
    \int_{X}{\rm dVol}_\textsf{FS}\;\rho=1 ~,
    \label{eq:one}
\ee
having the condition that the sampled points are uniformly distributed with respect to the Fubini--Study metric $\rho\sim  1/\mathrm{dVol}_\textsf{FS}$, we obtain the following expression
\be
\rho=\frac{1}{\rm Vol_\textsf{CY}}\,\frac{{\rm dVol}_\textsf{CY}}{{\rm dVol}_\textsf{FS}} ~.
\ee
with ${\rm dVol}_\textsf{CY}=\Omega\wedge\bar{\Omega}$. In this fashion, integration of a given function over the Calabi--Yau one obtains 
\be
\int_X {\rm dVol}_\textsf{CY}\; f(z,\bar{z})={\rm Vol}_\textsf{CY}\int_X {\rm dVol}_\textsf{FS}\; f(z,\bar{z})\,\rho ~.
\label{eq:two}
\ee
Let's now consider the integration as a finite sum of $N_p$ points uniformly distributed with respect to the Fubini--Study metric,~\eqref{eq:one} reads
\be
    \int_{X}\rho\;{\rm dVol}_\textsf{FS}=\frac{{\rm Vol}_\textsf{FS}}{N_p}\sum_i \rho(p_i)=1 ~.
\ee
Similarly,~\eqref{eq:two} reads:
\be
\int_X {\rm dVol}_\textsf{CY}\; f(z,\bar{z})={\rm Vol}_\textsf{CY}\frac{{\rm Vol}_\textsf{FS}}{N_p}\sum_i f(p_i)\,\rho(p_i)=\frac{{\rm Vol}_\textsf{CY}}{\sum_i \rho(p_i)}\sum_i f(p_i)\,\rho(p_i) ~.
\label{eq:two2}
\ee
Special care needs to be taken when dealing with topological quantities, such as 
\be 
\int_X c_3=\int_X {\rm dVol}_\textsf{FS}\; \frac{c_3}{{\rm dVol}_\textsf{FS}}=\frac{{\rm Vol}_\textsf{FS}}{N_p}\sum_{i} \frac{c_3}{{\rm dVol}_\textsf{FS}}\rho(p_i) ~.
\ee%
Using a different expression for the Euler number (from the Chern--Gauss--Bonnet theorem)
\be
\chi=\int_X {\rm dVol}_g\; K(R_g) ~,
\ee%
where $K(R_g)$ is a function of the Riemann tensor. Then if we want to compute the Euler number for the Fubini--Study metric we get:
\be
\chi_\textsf{FS}=\frac{{\rm Vol}_\textsf{FS}}{N_p} \sum_i K(R_\textsf{FS}) ~.
\ee
For the machine learned numerical metric 
\be
\chi_\textsf{ML}=\frac{{\rm Vol}_\textsf{FS}}{N_p} \sum_i K(R_\textsf{ML}) \frac{{\rm dVol}_\textsf{ML}}{{\rm dVol}_\textsf{FS}} ~.
\ee

\section{Plurisubharmonic property (metric positivity)}\label{sec:psh}
In order for $J_\phi = J_\text{FS} + \partial\overline{\partial}\phi$ to define a K\"ahler metric, the corresponding Riemannian metric must be positive definite, that is, $\phi\in C^\infty(X)$ must be $J_\text{FS}$-psh function, that is, for every point $p\in X$, there is a neighborhood $U$ of $p$, s.t. locally $J\vert_U = \partial\overline{\partial}(\psi + \phi)$, where $\psi+\phi$ is psh on $U$. Equivalently, $\phi$ is $J_\text{FS}$-psh if we have:
\begin{gather}
    J_\text{FS} + \partial\overline{\partial}\phi \geq 0 ~,
\end{gather}
in the sense of currents. The strict positivity is guaranteed by non-degeneracy of $J_\phi$. We may numerically check this, by computing the eigenvalues of the Hermitian matrix of coefficients of $J_\phi$. Then, $J_\phi>0$ iff all of the eigenvalues $\lambda$ of its matrix of coefficients are positive.

\begin{figure*}[htb]
    \centering
    \includegraphics[width=0.8\textwidth]{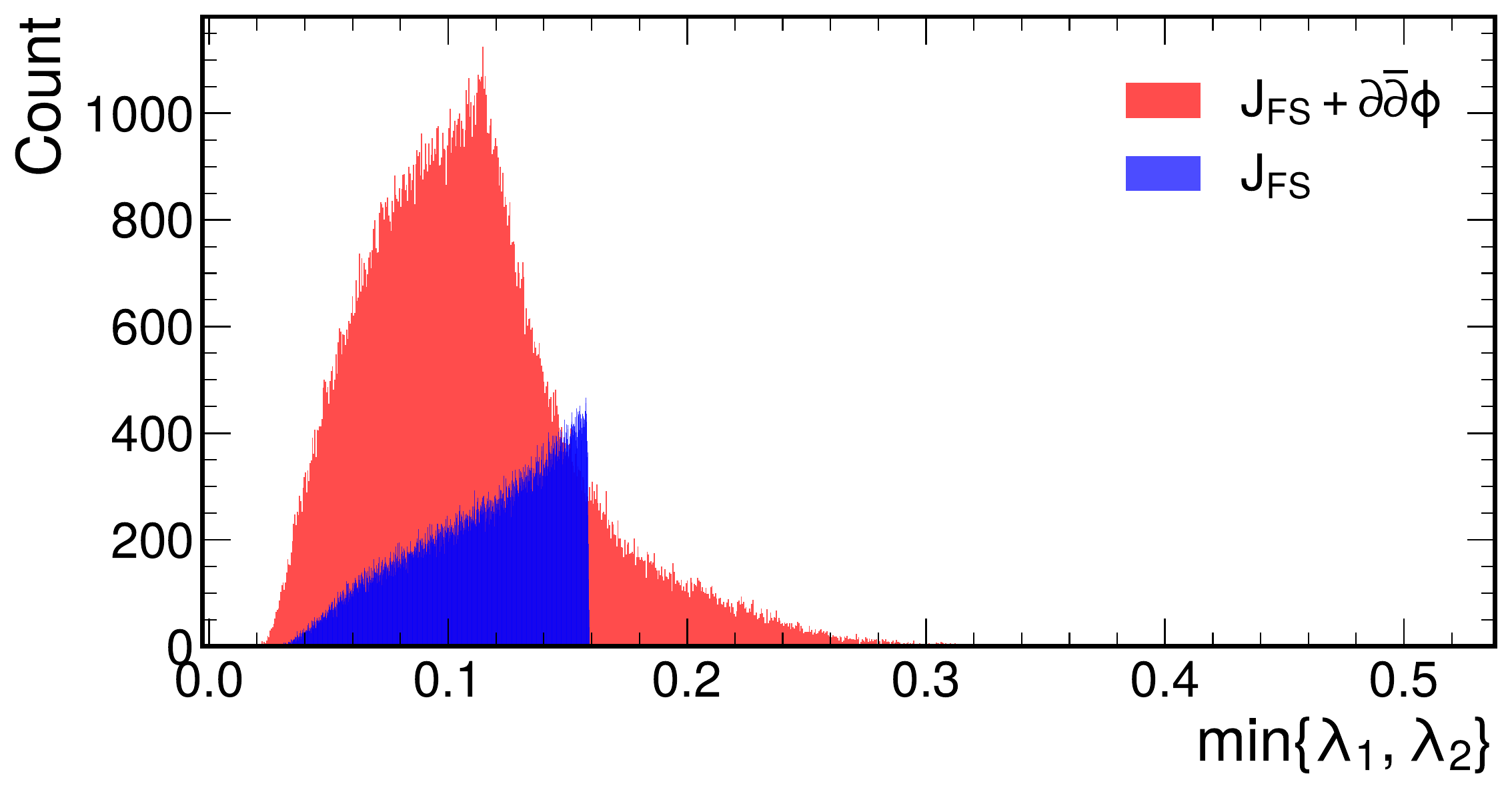}
    \caption{Distribution of the minimum of the eigenvalues over $X_\lambda$ for $\lambda = 0.99$.}\label{fig:psh_0.99}
\end{figure*}

The numerical result for $\lambda = 0.99$ using a spectral network $\phi$ compared against $J_\text{FS}$ is shown on Figure~\ref{fig:psh_0.99}. We observe a similar positivity of eigenvalues for different values of $\lambda$ in the Cefal\'{u} pencil. However, it must be noted that although this is a necessary condition for $\phi$ to be $J_\text{FS}$-plurisubharmonic, this is not sufficient.
\section{Classical volume terms dominate}
As an additional crosscheck we must ensure that the numerical quantities obtained are the most relevant even after including quantum corrections. Writing the \kae form in a basis for $H^{(1,1)}(X,\ZZ)$ as $J=t^i J_i$ we obtain the volume as 
\be 
V=\int_X J \wedge J \wedge J = k_{ijk} t^i t^j t^k ~,
\ee
where $k_{ijk}$ are the triple intersection numbers. Take the \kae potential for the quintic 
\be
K(z,\bar{z})=\frac{t}{2\pi}{\rm log}(z\bar{z})
\ee
in this manner, $J(z,\bar{z})=t J_0$, and the volume 
\be
V=t^3 V_0\,, \quad V_0=\int_X J_0\wedge J_0 \wedge J_0 ~.
\ee
The Kahler potential for the IIB flux compactification on the Calabi--Yau is given (with the first $\alpha^\prime$ corrections is given by~\cite{Cicoli:2008va}
\be
K=-2\,{\rm log}\left(V-\frac{\chi(X)\,\zeta(3)}{4(4\pi)^3 g_s^{3/2}}\right)-{\rm log}(S+\bar{S})-{\rm log}\left(-i\int_X \Omega \wedge \bar{\Omega}\right) ~.
\ee
The term $\chi(X)\,\zeta(3)/4(4\pi)^3 g_s^{3/2}$ is the first $\alpha^\prime$ correction term. Note that it does not scale with $t$ as the Euler number is topological. In general, $\alpha^\prime$ does not have a $t$ dependence. Note that $g\sim t$, and hence $g^{-1}\sim t^{-1}$. In this manner, the Christoffel symbols $\Gamma \sim g^{-1}\partial g$ and therefore the Riemann tensors are independent of $t$.

\bibliographystyle{JHEP}
\bibliography{ref}

\providecommand{\href}[2]{#2}\begingroup\raggedright\begin{thebibliography}{10}

\bibitem{Candelas:1985en}
P.~Candelas, G.~T. Horowitz, A.~Strominger and E.~Witten, \emph{{Vacuum
  Configurations for Superstrings}},
  \href{http://dx.doi.org/10.1016/0550-3213(85)90602-9}{\emph{Nucl. Phys. B}
  {\bf 258} (1985) 46--74}.

\bibitem{Candelas:1987rx}
P.~Candelas and S.~Kalara, \emph{{Yukawa Couplings for a Three Generation
  Superstring Compactification}},
  \href{http://dx.doi.org/10.1016/0550-3213(88)90271-4}{\emph{Nucl. Phys. B}
  {\bf 298} (1988) 357--368}.

\bibitem{Green:1987mn}
M.~B. Green, J.~Schwarz and E.~Witten, \emph{{Superstring Theory. Vol. 2: Loop
  Amplitudes, Anomalies and Phenomenology}}.
\newblock Cambridge University Press, 1988.

\bibitem{yau1977calabi}
S.-T. Yau, \emph{Calabi's conjecture and some new results in algebraic
  geometry}, {\emph{Proceedings of the National Academy of Sciences} {\bf 74}
  (1977) 1798--1799}.

\bibitem{Yau:1978}
S.-T. Yau, \emph{On the {R}icci curvature of a compact {K}{\"a}hler manifold
  and the complex {M}onge-{A}mp\`ere equation. {I}.}, {\emph{Comm.\ Pure.\
  Appl.\ Math.} {\bf 31} (1978) 339--411}.

\bibitem{Calabi:1954}
E.~Calabi, \emph{{The space of K\"ahler metrics}}, {\emph{Proc.\ Int.\ Cong.\
  Math.\ Amsterdam} {\bf 2} (1954) 206--207}.

\bibitem{Kachru:2018van}
S.~Kachru, A.~Tripathy and M.~Zimet, \emph{{K3 metrics from little string
  theory}},  \href{https://arxiv.org/abs/1810.10540}{{\tt 1810.10540}}.

\bibitem{Kachru:2020tat}
S.~Kachru, A.~Tripathy and M.~Zimet, \emph{{K3 metrics}},
  \href{https://arxiv.org/abs/2006.02435}{{\tt 2006.02435}}.

\bibitem{volkert2013space}
K.~Volkert, \emph{Space forms: a history},  in \emph{The Manifold Atlas}.
\newblock available online at:
  \url{http://www.map.mpim-bonn.mpg.de/Space_forms:_a_history}, 2013.

\bibitem{Ashmore:2019wzb}
A.~Ashmore, Y.-H. He and B.~A. Ovrut, \emph{{Machine Learning
  Calabi\textendash{}Yau Metrics}},
  \href{http://dx.doi.org/10.1002/prop.202000068}{\emph{Fortsch. Phys.} {\bf
  68} (2020) 2000068}, [\href{https://arxiv.org/abs/1910.08605}{{\tt
  1910.08605}}].

\bibitem{Anderson:2020hux}
L.~B. Anderson, M.~Gerdes, J.~Gray, S.~Krippendorf, N.~Raghuram and F.~Ruehle,
  \emph{{Moduli-dependent Calabi-Yau and SU(3)-structure metrics from Machine
  Learning}},  \href{https://arxiv.org/abs/2012.04656}{{\tt 2012.04656}}.

\bibitem{Douglas:2020hpv}
M.~R. Douglas, S.~Lakshminarasimhan and Y.~Qi, \emph{{Numerical Calabi-Yau
  metrics from holomorphic networks}},
  \href{https://arxiv.org/abs/2012.04797}{{\tt 2012.04797}}.

\bibitem{Jejjala:2020wcc}
V.~Jejjala, D.~K. Mayorga~Pena and C.~Mishra, \emph{{Neural network
  approximations for Calabi-Yau metrics}},
  \href{http://dx.doi.org/10.1007/JHEP08(2022)105}{\emph{JHEP} {\bf 08} (2022)
  105}, [\href{https://arxiv.org/abs/2012.15821}{{\tt 2012.15821}}].

\bibitem{Larfors:2021pbb}
M.~Larfors, A.~Lukas, F.~Ruehle and R.~Schneider, \emph{{Learning Size and
  Shape of Calabi-Yau Spaces}},  \href{https://arxiv.org/abs/2111.01436}{{\tt
  2111.01436}}.

\bibitem{Ashmore:2021ohf}
A.~Ashmore, L.~Calmon, Y.-H. He and B.~A. Ovrut, \emph{{Calabi-{Y}au Metrics,
  Energy Functionals and Machine-Learning}},
  \href{https://arxiv.org/abs/2112.10872}{{\tt 2112.10872}}.

\bibitem{Larfors:2022nep}
M.~Larfors, A.~Lukas, F.~Ruehle and R.~Schneider, \emph{{Numerical metrics for
  complete intersection and Kreuzer\textendash{}Skarke Calabi\textendash{}Yau
  manifolds}}, \href{http://dx.doi.org/10.1088/2632-2153/ac8e4e}{\emph{Mach.
  Learn. Sci. Tech.} {\bf 3} (2022) 035014},
  [\href{https://arxiv.org/abs/2205.13408}{{\tt 2205.13408}}].

\bibitem{Ashmore:2020ujw}
A.~Ashmore, \emph{{Eigenvalues and eigenforms on Calabi-Yau threefolds}},
  \href{https://arxiv.org/abs/2011.13929}{{\tt 2011.13929}}.

\bibitem{Ashmore:2021qdf}
A.~Ashmore and F.~Ruehle, \emph{{Moduli-dependent {KK} towers and the swampland
  distance conjecture on the quintic {C}alabi-{Y}au manifold}},
  \href{http://dx.doi.org/10.1103/PhysRevD.103.106028}{\emph{Phys. Rev. D} {\bf
  103} (2021) 106028}, [\href{https://arxiv.org/abs/2103.07472}{{\tt
  2103.07472}}].

\bibitem{cymetric}
M.~Larfors, A.~Lukas, F.~Ruehle and R.~Schneider, \emph{\texttt{cymetric}},
\newblock available online at:
  \url{https://github.com/pythoncymetric/cymetric}, 2021.

\bibitem{Kreuzer:2000xy}
M.~Kreuzer and H.~Skarke, \emph{{Complete classification of reflexive polyhedra
  in four-dimensions}}, {\emph{Adv.Theor.Math.Phys.} {\bf 4} (2002)
  1209--1230}, [\href{https://arxiv.org/abs/hep-th/0002240}{{\tt
  hep-th/0002240}}].

\bibitem{donaldson2001}
S.~Donaldson, \emph{Scalar curvature and projective embeddings, {I}},
  \href{http://dx.doi.org/10.4310/jdg/1090349449}{\emph{J. Differential Geom.}
  {\bf 59} (11, 2001) 479--522}.

\bibitem{donaldson2005some}
S.~K. Donaldson, \emph{Some numerical results in complex differential
  geometry}, {\emph{arXiv preprint math/0512625} (2005) }.

\bibitem{Douglas:2006rr}
M.~R. Douglas, R.~L. Karp, S.~Lukic and R.~Reinbacher, \emph{{Numerical
  Calabi-Yau metrics}}, \href{http://dx.doi.org/10.1063/1.2888403}{\emph{J.
  Math. Phys.} {\bf 49} (2008) 032302},
  [\href{https://arxiv.org/abs/hep-th/0612075}{{\tt hep-th/0612075}}].

\bibitem{Headrick:2005ch}
M.~Headrick and T.~Wiseman, \emph{{Numerical Ricci-flat metrics on K3}},
  \href{http://dx.doi.org/10.1088/0264-9381/22/23/002}{\emph{Class. Quant.
  Grav.} {\bf 22} (2005) 4931--4960},
  [\href{https://arxiv.org/abs/hep-th/0506129}{{\tt hep-th/0506129}}].

\bibitem{Headrick:2009jz}
M.~Headrick and A.~Nassar, \emph{{Energy functionals for Calabi-Yau metrics}},
  \href{http://dx.doi.org/10.4310/ATMP.2013.v17.n5.a1}{\emph{Adv. Theor. Math.
  Phys.} {\bf 17} (2013) 867--902},
  [\href{https://arxiv.org/abs/0908.2635}{{\tt 0908.2635}}].

\bibitem{Cui:2019uhy}
W.~Cui and J.~Gray, \emph{{Numerical Metrics, Curvature Expansions and
  Calabi-Yau Manifolds}},
  \href{http://dx.doi.org/10.1007/JHEP05(2020)044}{\emph{JHEP} {\bf 05} (2020)
  044}, [\href{https://arxiv.org/abs/1912.11068}{{\tt 1912.11068}}].

\bibitem{Douglas:2021zdn}
M.~R. Douglas, \emph{{Holomorphic feedforward networks}},
  \href{http://dx.doi.org/10.4310/PAMQ.2022.v18.n1.a7}{\emph{Pure Appl. Math.
  Quart.} {\bf 18} (2022) 251--268},
  [\href{https://arxiv.org/abs/2105.03991}{{\tt 2105.03991}}].

\bibitem{Douglas:2006hz}
M.~R. Douglas, R.~L. Karp, S.~Lukic and R.~Reinbacher, \emph{{Numerical
  solution to the Hermitian Yang-Mills equation on the Fermat quintic}},
  \href{http://dx.doi.org/10.1088/1126-6708/2007/12/083}{\emph{JHEP} {\bf 12}
  (2007) 083}, [\href{https://arxiv.org/abs/hep-th/0606261}{{\tt
  hep-th/0606261}}].

\bibitem{Braun:2007sn}
V.~Braun, T.~Brelidze, M.~R. Douglas and B.~A. Ovrut, \emph{{Calabi-Yau Metrics
  for Quotients and Complete Intersections}},
  \href{http://dx.doi.org/10.1088/1126-6708/2008/05/080}{\emph{JHEP} {\bf 05}
  (2008) 080}, [\href{https://arxiv.org/abs/0712.3563}{{\tt 0712.3563}}].

\bibitem{Braun:2008jp}
V.~Braun, T.~Brelidze, M.~R. Douglas and B.~A. Ovrut, \emph{{Eigenvalues and
  Eigenfunctions of the Scalar Laplace Operator on Calabi-Yau Manifolds}},
  \href{http://dx.doi.org/10.1088/1126-6708/2008/07/120}{\emph{JHEP} {\bf 07}
  (2008) 120}, [\href{https://arxiv.org/abs/0805.3689}{{\tt 0805.3689}}].

\bibitem{Anderson:2010ke}
L.~B. Anderson, V.~Braun, R.~L. Karp and B.~A. Ovrut, \emph{{Numerical
  Hermitian Yang-Mills Connections and Vector Bundle Stability in Heterotic
  Theories}}, \href{http://dx.doi.org/10.1007/JHEP06(2010)107}{\emph{JHEP} {\bf
  06} (2010) 107}, [\href{https://arxiv.org/abs/1004.4399}{{\tt 1004.4399}}].

\bibitem{Anderson:2011ed}
L.~B. Anderson, V.~Braun and B.~A. Ovrut, \emph{{Numerical Hermitian Yang-Mills
  Connections and Kahler Cone Substructure}},
  \href{http://dx.doi.org/10.1007/JHEP01(2012)014}{\emph{JHEP} {\bf 01} (2012)
  014}, [\href{https://arxiv.org/abs/1103.3041}{{\tt 1103.3041}}].

\bibitem{Catanese:2021aa}
F.~Catanese, \emph{Kummer quartic surfaces, strict self-duality, and more},
  \href{https://arxiv.org/abs/2101.10501}{{\tt 2101.10501}}.

\bibitem{Candelas:1990rm}
P.~Candelas, X.~C. De~La~Ossa, P.~S. Green and L.~Parkes, \emph{{A Pair of
  Calabi-Yau manifolds as an exactly soluble superconformal theory}},
  \href{http://dx.doi.org/10.1016/0550-3213(91)90292-6}{\emph{Nucl. Phys. B}
  {\bf 359} (1991) 21--74}.

\bibitem{Candelas:1990qd}
P.~Candelas, X.~C. De~la Ossa, P.~S. Green and L.~Parkes, \emph{{An Exactly
  soluble superconformal theory from a mirror pair of Calabi-Yau manifolds}},
  \href{http://dx.doi.org/10.1016/0370-2693(91)91218-K}{\emph{Phys. Lett. B}
  {\bf 258} (1991) 118--126}.

\bibitem{Nahm:1999ps}
W.~Nahm and K.~Wendland, \emph{A hiker's guide to {K3}: Aspects of ${N}=(4,4)$
  superconformal field theory with central charge $c = 6$},
  \href{http://dx.doi.org/10.1007/PL00005548}{\emph{Commun. Math. Phys.} {\bf
  216} (2001) 85--138}, [\href{https://arxiv.org/abs/hep-th/9912067}{{\tt
  hep-th/9912067}}].

\bibitem{Wendland:2003ma}
K.~Wendland, \emph{{On Superconformal field theories associated to very
  attractive quartics}},  in \emph{{Les Houches School of Physics: Frontiers in
  Number Theory, Physics and Geometry}}, pp.~223--244, 2007.
\newblock \href{https://arxiv.org/abs/hep-th/0307066}{{\tt hep-th/0307066}}.
\newblock \href{http://dx.doi.org/10.1007/978-3-540-30308-4_5}{DOI}.

\bibitem{Berglund:2022dgb}
P.~Berglund and T.~H\"ubsch, \emph{{Hirzebruch Surfaces, Tyurin Degenerations
  and Toric Mirrors: Bridging Generalized Calabi-Yau Constructions}},
  \href{https://arxiv.org/abs/2205.12827}{{\tt 2205.12827}}.

\bibitem{Reid:2022aa}
M.~Reid, ``The {D}u {V}al singularities ${A}_n$, ${D}_n$, ${E}_6$, ${E}_7$,
  ${E}_8$.'' \url{https://homepages.warwick.ac.uk/~masda/surf/more/DuVal.pdf},
  October 9, 2022.

\bibitem{hubsch1992calabi}
T.~Hubsch, \emph{Calabi--{Y}au manifolds: A Bestiary for Physicists}.
\newblock World Scientific, 1992.

\bibitem{Atiyah:1989ty}
M.~Atiyah and G.~B. Segal, \emph{On equivariant {E}uler characteristics},
  \href{http://dx.doi.org/https://doi.org/10.1016/0393-0440(89)90032-6}{\emph{J.
  Geom. Phys.} {\bf 6} (1989) 671--677}.

\bibitem{candelas1988lectures}
P.~Candelas, \emph{Lectures on complex manifolds},  in \emph{Superstrings and
  grand unification} (T.~Pradhan, ed.), 1988 Winter School on high energy
  physics; Puri (India).
\newblock World Sci. Publishing, Singapore, 1988.

\bibitem{Ballmann}
W.~Ballmann, \emph{Lectures on K\"ahler Manifolds}.
\newblock EMS Press, 1st~ed., 2006.

\bibitem{10.2307/1971080}
R.~D. MacPherson, \emph{Chern classes for singular algebraic varieties},
  {\emph{Annals of Mathematics} {\bf 100} (1974) 423--432}.

\bibitem{10.1155/S1073792894000498}
P.~Aluffi, \emph{{MacPherson's and Fulton's Chern classes of hypersurfaces}},
  \href{http://dx.doi.org/10.1155/S1073792894000498}{\emph{International
  Mathematics Research Notices} {\bf 1994} (06, 1994) 455--465},
  [\href{https://arxiv.org/abs/https://academic.oup.com/imrn/article-pdf/1994/11/455/6768367/1994-11-455.pdf}{{\tt
  https://academic.oup.com/imrn/article-pdf/1994/11/455/6768367/1994-11-455.pdf}}].

\bibitem{10.2307/118036}
P.~Aluffi, \emph{Chern classes for singular hypersurfaces}, {\emph{Transactions
  of the American Mathematical Society} {\bf 351} (1999) 3989--4026}.

\bibitem{helmer2016algorithms}
M.~Helmer, \emph{{Algorithms to compute the topological Euler characteristic,
  Chern--Schwartz--MacPherson class and Segre class of projective varieties}},
  {\emph{Journal of Symbolic Computation} {\bf 73} (2016) 120--138}.

\bibitem{aluffi2019chern}
P.~Aluffi, \emph{The {C}hern--{S}chwartz--{M}ac{P}herson class of an embeddable
  scheme},  in \emph{Forum of Mathematics, Sigma}, vol.~7, Cambridge University
  Press, 2019.

\bibitem{parusinski1998characteristic}
A.~Parusinski and P.~Pragacz, \emph{Characteristic classes of hypersurfaces and
  characteristic cycles}, {\emph{arXiv preprint math/9801102} (1998) }.

\bibitem{aluffi2019pfaffian}
P.~Aluffi and M.~Goresky, \emph{Pfaffian integrals and invariants of singular
  varieties}, {\emph{arXiv preprint arXiv:1901.06312} (2019) }.

\bibitem{articleCyclesPolaires}
R.~Piene, \emph{Cycles polaires et classes de {C}hern pour les vari{\'e}t{\'e}s
  projectives singuli{\`e}res. (on polar cycles and {C}hern classes of singular
  projective varieties)}, {\emph{Travaux en Cours.} {\bf 37} (01, 1988) }.

\bibitem{siersma2022polar}
D.~Siersma and M.~Tib{\u{a}}r, \emph{Polar degree of hypersurfaces with
  1-dimensional singularities}, {\emph{Topology and its Applications} {\bf 313}
  (2022) 107992}.

\bibitem{huh2014milnor}
J.~Huh, \emph{Milnor numbers of projective hypersurfaces with isolated
  singularities}, {\emph{Duke Mathematical Journal} {\bf 163} (2014)
  1525--1548}.

\bibitem{https://doi.org/10.1002/cpa.3160310304}
S.-T. Yau, \emph{On the {R}icci curvature of a compact {K}{\"a}hler manifold
  and the complex {M}onge-{A}mp{\'e}re equation, {I}},
  \href{http://dx.doi.org/https://doi.org/10.1002/cpa.3160310304}{\emph{Communications
  on Pure and Applied Mathematics} {\bf 31} (1978) 339--411},
  [\href{https://arxiv.org/abs/https://onlinelibrary.wiley.com/doi/pdf/10.1002/cpa.3160310304}{{\tt
  https://onlinelibrary.wiley.com/doi/pdf/10.1002/cpa.3160310304}}].

\bibitem{doi:10.1073/pnas.74.5.1798}
S.-T. Yau, \emph{Calabi's conjecture and some new results in algebraic
  geometry}, \href{http://dx.doi.org/10.1073/pnas.74.5.1798}{\emph{Proceedings
  of the National Academy of Sciences} {\bf 74} (1977) 1798--1799},
  [\href{https://arxiv.org/abs/https://www.pnas.org/doi/pdf/10.1073/pnas.74.5.1798}{{\tt
  https://www.pnas.org/doi/pdf/10.1073/pnas.74.5.1798}}].

\bibitem{Guggenheimer1952}
H.~Guggenheimer, \emph{\"{U}ber vierdimensionale {E}insteinr\"{a}ume},
  \href{http://dx.doi.org/10.1007/bf02296683}{\emph{Experientia} {\bf 8} (Nov.,
  1952) 420--421}.

\bibitem{roulleau2019generalized}
X.~Roulleau, \emph{On generalized {K}ummer surfaces and the orbifold
  {B}ogomolov-{M}iyaoka-{Y}au inequality}, {\emph{Transactions of the American
  Mathematical Society} {\bf 371} (2019) 7651--7668}.

\bibitem{jax2018github}
J.~Bradbury, R.~Frostig, P.~Hawkins, M.~J. Johnson, C.~Leary, D.~Maclaurin
  et~al., \emph{{JAX}: composable transformations of {P}ython+{N}um{P}y
  programs}, 2018.
\newblock \url{http://github.com/google/jax}.

\bibitem{spectralNetworks}
P.~Berglund, G.~Butbaia, T.~H\"ubsch, V.~Jejjala, O.~Kim, D.~Mayorga Pe\~na
  et~al., \emph{{Work in progress}}, .

\bibitem{M2}
D.~R. Grayson and M.~E. Stillman, ``Macaulay2, a software system for research
  in algebraic geometry.'' Available at
  \url{http://www.math.uiuc.edu/Macaulay2/}.

\bibitem{candelas:1987kf}
P.~Candelas, A.~Dale, C.~Lutken and R.~Schimmrigk, \emph{{Complete Intersection
  Calabi-Yau Manifolds}},
  \href{http://dx.doi.org/10.1016/0550-3213(88)90352-5}{\emph{Nucl.Phys.} {\bf
  B298} (1988) 493}.

\bibitem{cibotaru2021odd}
D.~Cibotaru and S.~Moroianu, \emph{Odd {P}faffian forms}, {\emph{Bulletin of
  the Brazilian Mathematical Society, New Series} {\bf 52} (2021) 915--976}.

\bibitem{moroianu2022higher}
S.~Moroianu, \emph{Higher transgressions of the {P}faffian}, {\emph{Revista
  Matem{\'a}tica Iberoamericana} {\bf 38} (2022) 1425--1452}.

\bibitem{douglas2021holomorphic}
M.~R. Douglas, \emph{Holomorphic feedforward networks}, {\emph{arXiv preprint
  arXiv:2105.03991} (2021) }.

\bibitem{PhysRevD.103.106028}
A.~Ashmore and F.~Ruehle, \emph{Moduli-dependent {KK} towers and the swampland
  distance conjecture on the quintic {C}alabi-{Y}au manifold},
  \href{http://dx.doi.org/10.1103/PhysRevD.103.106028}{\emph{Phys. Rev. D} {\bf
  103} (May, 2021) 106028}.

\bibitem{braun2008calabi}
V.~Braun, T.~Brelidze, M.~R. Douglas and B.~A. Ovrut, \emph{Calabi-yau metrics
  for quotients and complete intersections}, {\emph{Journal of High Energy
  Physics} {\bf 2008} (2008) 080}.

\bibitem{shiffman1999distribution}
B.~Shiffman and S.~Zelditch, \emph{Distribution of zeros of random and quantum
  chaotic sections of positive line bundles}, {\emph{Communications in
  Mathematical Physics} {\bf 200} (1999) 661--683}.

\bibitem{Cicoli:2008va}
M.~Cicoli, J.~P. Conlon and F.~Quevedo, \emph{{General Analysis of LARGE Volume
  Scenarios with String Loop Moduli Stabilisation}},
  \href{http://dx.doi.org/10.1088/1126-6708/2008/10/105}{\emph{JHEP} {\bf 10}
  (2008) 105}, [\href{https://arxiv.org/abs/0805.1029}{{\tt 0805.1029}}].

\end{thebibliography}\endgroup

\end{document}